\documentclass[12pt]{article}
\usepackage[letterpaper, margin=1in]{geometry}

\usepackage{times}
\usepackage{bm}
\usepackage{natbib}
\RequirePackage[OT1]{fontenc}

\usepackage{mathtools}
\usepackage{url}
\usepackage{appendix, array}
\usepackage{graphics,bbm,bm}
\usepackage{graphicx}
\usepackage{enumerate}
\usepackage{amsfonts,animate,subfig}
\usepackage{amssymb, amsthm}
\usepackage{rotating}
\usepackage{algorithm}  % algorithm environment
\usepackage{algpseudocode}  % algorithmic environment
\usepackage{lscape}
\usepackage[export]{adjustbox}
\usepackage[normalem]{ulem}
\usepackage{enumitem}
\usepackage[usenames,dvipsnames]{xcolor}

\makeatletter
\newcommand*{\bdiv}{%
  \nonscript\mskip-\medmuskip\mkern5mu%
  \mathbin{\operator@font div}\penalty900\mkern5mu%
  \nonscript\mskip-\medmuskip
}
\makeatother

\renewcommand{\hat}[1]{\widehat{#1}}
\renewcommand{\tilde}[1]{\widetilde{#1}}

\renewcommand{\P}{\mathbb{P}}
\newcommand{\R}{\mathbb{R}}

\newcommand{\A}{\mathbf{A}}
\newcommand{\B}{\mathbf{B}}

\newcommand{\bp}{\mathbf{p}}

\newcommand{\bSigma}{\boldsymbol{\Sigma}}
\newcommand{\C}{\mathbf{C}}
\newcommand{\D}{\mathbf{D}}

\newcommand{\I}{\mathbf{I}}

\newcommand{\M}{\mathbf{M}}
\newcommand{\p}{\mathbf{p}}
\newcommand{\Q}{\mathbf{Q}}

\renewcommand{\S}{\mathcal{S}}

\renewcommand{\u}{\mathbf{u}}
\renewcommand{\v}{\mathbf{v}}

\newcommand{\w}{\mathbf{w}}
\newcommand{\X}{\mathbf{X}}
\newcommand{\x}{\mathbf{x}}

\newcommand{\y}{\mathbf{y}}

\newcommand{\z}{\mathbf{z}}

\newcommand{\bbeta}{\boldsymbol{\beta}}

\newcommand{\bepsilon}{\boldsymbol{\epsilon}}

\newcommand{\bOmega}{\boldsymbol{\Omega}}

\newcommand{\bphi}{\boldsymbol{\phi}}

\newcommand{\btheta}{\boldsymbol{\theta}}

\newcommand{\s}[1]{\mathcal{#1}}
\newcommand{\mbf}[1]{\mathbf{#1}}

\newtheorem{proposition}{Proposition}[section]

\title{Regression of binary network data with exchangeable latent errors}

\author{Frank W. Marrs$^1$ and Bailey K. Fosdick$^2$ 
}
\date{%
    $^1$Statistical Sciences, Los Alamos National Laboratory,\\ 
    \textit{fmarrs3@lanl.gov}\\[1ex]%
    $^2$Department of Biostatistics \& Informatics, Colorado School of Public Health\\
    \today
}
\begin{document}

\maketitle

\begin{abstract}
Undirected, binary network data consist of indicators of symmetric relations between pairs of actors. 
Regression models of such data allow for the estimation of effects of exogenous covariates on the network and for prediction of unobserved data.
Ideally, estimators of the regression parameters should account for the inherent dependencies among relations in the network that involve the same actor.
To account for such dependencies,  researchers have developed a host of latent variable network models, however, estimation of many latent variable network models is computationally onerous and which model is best to base inference upon may not be clear. 
We propose the Probit Exchangeable (PX) model for undirected binary network data that is
based on an assumption of exchangeability, which is common to many of the latent variable 
network models in the literature. 
The PX model can represent the first two moments of any exchangeable network model.
We leverage the EM algorithm to obtain an approximate maximum likelihood estimator of the PX model that is extremely computationally efficient.
Using simulation studies, we demonstrate the improvement in estimation of regression coefficients of the proposed model over existing latent variable network models. In an analysis of purchases of politically-aligned books, we demonstrate political polarization in purchase behavior and show that the proposed estimator significantly reduces runtime relative to estimators of latent variable network models, while maintaining predictive performance. \\

\end{abstract}

\section{Introduction}
Undirected binary network data  measure the presence or absence of a relationship between pairs of actors and have recently become extremely common in the social and biological sciences. 
Some examples of data that are naturally represented as undirected binary networks are international relations among countries \citep{fagiolo2008topological}, gene co-expression \citep{zhang2005general}, and interactions among students  \citep{han2016using}. 
We focus on an example of politically-aligned books, where a relation exists between two books if they were frequently purchased by the same person on Amazon.com.  
Our motivations are estimation of the effects of exogenous covariates, such as the effect of alignment of political ideologies of pairs of books on the propensity for books to be purchased by the same consumer, and the related problem of predicting unobserved relations using book ideological information. For example, predictions of relations between new books and old books could be used to recommend new books to potential purchasers.

A binary, undirected network
$\left\{ y_{ij} \in \{0,1 \}: i,j\in\{1,...,n\}, i < j \right\}$, which we abbreviate $\{y_{ij} \}_{ij}$,
may be represented as an $n \times n$ symmetric adjacency matrix which describes the presence or absence of relationships between unordered pairs of $n$ actors.
The diagonal elements of the matrix $\{y_{ii}: i \in \{1,...,n\}\}$ are assumed to be undefined, as we do not consider actor relations with him/herself. We use $\y$ to refer to the $\binom{n}{2}$ vector of network relations formed by a columnwise vectorization of the upper triangle of the matrix corresponding to $\{y_{ij} \}_{ij}$.

A regression model for the probability of observing a binary outcome is the probit model, which can be expressed 
\begin{align}
    \P(y_{ij} = 1 ) = \P \left( \x_{ij}^T \bbeta + \epsilon_{ij} > 0 \right), \quad  \label{eq_probit_prob}
\end{align}
where $\epsilon_{ij}$ is a mean-zero normal random error, $\x_{ij}$ is a fixed vector of covariates corresponding to relation $ij$, and $\bbeta$ is a vector of coefficients to be estimated. 
When each entry in the error network $\{ \epsilon_{ij} \}_{ij}$ is independent of the others, estimation of the probit regression model in \eqref{eq_probit_prob} is straightforward and proceeds via standard gradient methods for maximum likelihood estimation of generalized linear models \citep{greene2003econometric}.
The assumption of independence of $\{ \epsilon_{ij} \}_{ij}$ may be appropriate when the mean $\{ \x_{ij}^T \bbeta \}_{ij}$ represents nearly all of the dependence in the network $\{ y_{ij}\}_{ij}$.
However, network data naturally contain excess dependence beyond the mean: the errors $\epsilon_{ij}$ and $\epsilon_{ik}$ both concern actor $i$ (see \cite{faust1994social}, e.g., for further discussion of dependencies in network data). In the context of the political books data set, the propensity 
of  ``Who's Looking Out For You?'' by Bill O'Reilly to be purchased by the same reader as ``Deliver Us from Evil'' by Sean Hannity may be similar to the propensity of ``Who's Looking Out For You?'' and ``My Life'' by Bill Clinton to be co-purchased simply because ``Who's Looking Out For You?'' is a popular book.
Or, in a student friendship network,
the friendship that Julie makes with Steven may be related to the friendship that Julie makes with Asa due to Julie's  gregariousness.  Unlike the case of typical linear regression, the estimator that maximizes the likelihood of the 
generalized linear regression model in \eqref{eq_probit_prob}, when assuming independence of each entry in the error network $\{ \epsilon_{ij} \}_{ij}$, is not unbiased for $\bbeta$.
Ignoring the excess dependence in $\{\epsilon_{ij} \}_{ij}$ can thus be expected to result in poor estimation of $\bbeta$ and poor out-of-sample predictive performance. We observe this phenomenon in the simulation studies and analysis of the political books network (see Sections~\ref{sec:sim}~and~\ref{sec:data}, respectively). Thus,  estimators of $\bbeta$ and $\P(y_{ij} = 1)$ in \eqref{eq_probit_prob} for the network $\{ y_{ij} \}_{ij}$ should ideally account for the excess dependence of network data. A host of regression models exist in the literature that do just this; we briefly review these here.

A method used to account for excess dependence in regression of binary network data is the estimation of generalized linear mixed models, which were first introduced for repeated measures studies \citep{ stiratelli1984random, breslow1993approximate}. In these models, a random effect, i.e. latent variable, is estimated for each individual in the study, to account 
for possible individual variation. \cite{warner1979new} used latent variables to account for excess network dependence when analyzing  data with continuous measurements of relationships between actors, and  \cite{holland1981exponential} extended their approach to networks consisting of binary observations.
\cite{hoff2002latent} further extended this approach to include nonlinear functions of latent variables, and since then, many variations have been proposed \citep{handcock2007model, hoff2008modeling, sewell2015latent}. We refer to parametric network models wherein the observations are independent conditional on  random latent variables as ``latent variable network models,'' which we discuss in detail in Section~\ref{sec:latent}.
Separate latent variable approaches may lead to vastly different estimates of $\bbeta$, and it may not be clear which model's estimate of $\bbeta$, or prediction, to choose.  Goodness-of-fit checks are the primary method of assessing latent variable network model fit \citep{hunter2008goodness}, however,  selecting informative statistics is a well known challenge.
Finally, latent variable network models are typically computationally burdensome to estimate, often relying on Markov chain Monte Carlo methods.

Another approach to estimating covariate effects on network outcomes is the estimation of exponential random graph models, known as ERGMs. ERGMs represent the probability of relation formation using a generalized exponential family distribution, $\P(y_{ij} = 1) \propto \text{exp}(\mathbf{t}(\y_{ij}, \x_{ij})^T \theta)$, where $\theta$ is a vector of parameters to be estimated. In this flexible formulation, the effects of the exogenous covariates are included in the network statistics $\mathbf{t}(\y_{ij}, \x_{ij})$. ERGMs also account for excess network dependence using the network statistics $\mathbf{t}(\y_{ij}, \x_{ij})$, such as counts of the number of observed triangles or the number of ``2-stars'' -- pairs of indicated relations that share an actor.
ERGMs were developed by \cite{frank1986markov} and \cite{snijders2006new}, and are typically estimated using Markov chain Monte Carlo (MCMC) approximations to posterior distributions \citep{snijders2002markov, handcock:statnet, handcock:ergm}. ERGMs have been shown to be prone to place unrealistic quantities of probability mass on networks consisting of all `1's or all `0's \citep{handcock2003assessing, schweinberger2011instability}, and the estimation procedures may be slow to complete \citep{caimo2011bayesian}. Further, parameter estimates typically cannot be generalized to populations outside the observed network \citep{shalizi2013consistency}. %Since we are interested in general regressions of network data where generalization to populations outside the observed network may in fact be the exact goal, we do not examine ERGMs further. 

A final approach to account for excess network dependence is to explicitly model the  correlation among network observations. This is the approach we take in this paper.  In this approach, an unobserved normal random variable, $z_{ij}$, is proposed to underlie each data point, such that $y_{ij} = \mathbf{1}[z_{ij} > 0]$ for $\z \sim {\rm N}(\X \bbeta, \bOmega(\btheta))$. 
In this formulation, excess dependence due to the network is accounted for in $\bOmega$.
The parameters $\bbeta$ and $\btheta$ of the distribution of the unobserved normal random variables $\{ z_{ij} \}_{ij}$ may be estimated using likelihood methods. For example, \cite{ashford1970multi} propose likelihood ratio hypothesis tests and \cite{ochi1984likelihood} give closed-form parameter estimators for studies of repeated observations on the same individual, such that $\bOmega(\btheta)$ is block diagonal.  In more general scenarios, such as unrestricted correlation structures,  methods such as semi-parametrics \citep{connolly1988conditional}, pseudo-likelihoods \citep{le1994logistic}, and MCMC approximations to EM algorithms \citep{chib1998analysis, li2008likelihood} are employed for estimation.

In this paper, we propose the Probit Exchangeable (PX) Model, a parsimonious regression model for undirected binary network data based on an assumption of exchangeability of the unobserved normal random variables $\{ z_{ij} \}_{ij}$.  The assumption of exchangeability is pervasive in random network models and, in fact, underlies many of the latent variable network models (see Section~\ref{sec:exch} for a detailed discussion of exchangeability)\footnote{We consider infinite exchangeability such that the exchangeable generating process is valid for arbitrarily large numbers of actors $n$, as in \cite{hoover1979relations} and \cite{aldous1981representations}.}. We show that, under exchangeability, the excess network dependence in $\{ z_{ij} \}_{ij}$ may be quantified using a single parameter $\rho$ such that $\bOmega(\btheta) = \bOmega(\rho)$. This fact remains regardless of the particular exchangeable generating model, and thus, our approach can be seen as subsuming exchangeable latent network variable models, at least up to the second moment of their latent distributions. 
The proposed model may be rapidly estimated using an expectation-maximization (EM) algorithm to attain a numerical approximation to the maximum likelihood estimator, where we make approximations in the expectation step for runtime considerations.  The estimation scheme we employ is similar to those used to estimate generalized linear mixed models in the literature \citep{littell2006sas, gelman2006data}. 

This paper is organized as follows. As latent variable network models are strongly related to our work, we review them in detail in Section~\ref{sec:latent}. We provide supporting theory for exchangeable random network models and their connections to latent variable network models in Section~\ref{sec:exch}. In Section~\ref{sec:PX}, we define the PX model and then the estimation thereof in Section~\ref{sec:estimation}. In Section~\ref{sec:pred_proc}, we give a method for making predictions on unobserved relations.  %\fr{(Uncertainty section~\ref{sec:uncertainty}).} 
We provide simulation studies demonstrating consistency of the proposed estimation algorithm, and demonstrating the improvement with the proposed model over latent variable network models in estimating $\bbeta$ in Section~\ref{sec:sim}. We analyze a network of political books in Section~\ref{sec:data}, demonstrating the reduction in runtime when using the PX model, and compare its out-of-sample performance to existing latent variable network models. A discussion with an eye toward future work is provided in Section~\ref{sec:disc}.

%%%%
\section{Latent variable network models}
\label{sec:latent}
In this section, we briefly summarize a number of latent variable network models in the literature that are used to capture excess dependence in network observations. All latent variable network models we consider here may be written in the common form
\begin{align}
&\P(y_{ij} = 1 ) = \P \left( \mu_{ij} + f_{\btheta} (\v_i, \v_j) + \xi_{ij} > 0 \right),  \quad
\label{eq_gen_lv} \\
&\v_i \stackrel{iid}{\sim} (\mathbf{0}, \bSigma_v), 
\hspace{.3in}
\xi_{ij} \stackrel{iid}{\sim} {\rm N}(0, \sigma^2), 
\nonumber
\end{align}
where $\v_i \in \R^K$ with mean $\mathbf{0}$ and covariance matrix $\bSigma_v$, and $\mu_{ij}$ is fixed. We avoid specifying a distribution for the latent vectors $\{ \v_{i} \}_{i=1}^n$, although they are often taken to be multivariate Gaussian. We set the total variance of the latent variable representation to be $1 = \sigma^2 + var[f_{\btheta} (\v_i, \v_j)]$, since it is not identifiable. The function of the latent variables $f_{\btheta} : \R^K \times \R^K \rightarrow \R$, parametrized by $\btheta$, serves to distinguish the latent variable network models discussed below. 
Regression latent variable network models are formed when the latent mean is represented as a linear function of exogenous covariates $\x_{ij} \in \R^p$, such that $\mu_{ij} = \x_{ij}^T \bbeta$.
The latent nodal random vectors $\{ \v_{i} \}_{i=1}^n$ represent excess network dependence -- beyond the mean $\mu_{ij}$. Since relations $y_{ij}$ and $y_{ik}$ share latent vector $\v_i$ corresponding to shared actor $i$, and thus, $y_{ij}$ and $y_{ik}$ have related distributions through the latent function $f_{\btheta} (\v_i, \v_j)$. 
Many popular models for network data may be represented as in \eqref{eq_gen_lv}, such as the social relations model, the latent position model, and the latent eigenmodel.

\subsection{Social relations model}
The social relations model was first developed for continuous, directed network data
\citep{warner1979new, wong1982round, snijders1999social}. In the social relations model for binary network data \citep{hoff2005bilinear}, $f_{\btheta}(\v_i, \v_j) = \v_i + \v_j$ and $\v_i = a_i \in \mathbb{R}$ for each actor $i$, such that
\begin{align}
&\P(y_{ij} = 1 ) = \P \left( \x_{ij}^T \bbeta + a_i + a_j + \xi_{ij} > 0 \right),  \quad
\label{eq_binary_srm} \\
&a_i \stackrel{iid}{\sim} (0, \sigma^2_a),
\hspace{.3in}
\xi_{ij} \stackrel{iid}{\sim} {\rm N}(0, \sigma^2). \nonumber
% \nonumber
\end{align}
Each actor's latent variable $\{ a_i \}_{i=1}^n$ may be thought of as the actor's sociability: large values of $a_i$ correspond to actors with a higher propensity to form relations in the network. The random $\{ a_i \}_{i=1}^n$ in \eqref{eq_binary_srm} also account for the excess correlation in network data; any two relations that share an actor, e.g. $y_{ij}$ and $y_{ik}$, are marginally correlated.

\subsection{Latent position model}
A more complex model for representing excess dependence in social network data is the latent position model
\citep{hoff2002latent}. The latent position model extends the idea of the social relations model 
by giving each actor $i$ a latent position $\u_i$ in a Euclidean latent space, for example $\R^{K}$. Then, actors whose latent positions are closer together in Euclidean distance are more likely to share a relation: 
\begin{align}
&\P(y_{ij} = 1 ) = \P \left( \x_{ij}^T \bbeta + a_i + a_j - || \u_i - \u_j ||_2 + \xi_{ij} > 0\right), \quad \label{eq_binary_dist} \\
&a_i \stackrel{iid}{\sim} (0, \sigma^2_a),
\hspace{.3in}
\u_i \stackrel{iid}{\sim} (0, \bSigma_u),
\hspace{.3in}
\xi_{ij} \stackrel{iid}{\sim} {\rm N}(0, \sigma^2). \nonumber
\end{align}
In the form of \eqref{eq_gen_lv}, the latent position model contains latent random vector $\v_{i} = [a_i, \u_i]^T \in \R^{K+1}$, and $f_{\btheta}(\v_{i}, \v_j ) = a_i + a_j - || \u_i - \u_j ||_2$.
\cite{hoff2002latent} show that the latent position model is capable of representing transitivity, that is, when $y_{ij} = 1$ and $y_{jk} = 1$, it is more likely that $y_{ik} = 1$. Models that are transitive often display a pattern observed in social network data: a friend of my friend is also my friend \citep{wasserman1994social}.  

\subsection{Latent eigenmodel}
The latent eigenmodel also associates each actor with  a latent position $\u_i$ in a latent Euclidean space, however the inner product between latent positions (weighted by  symmetric parameter matrix $\Lambda$) measures the propensity of actors $i$ and $j$ to form a relation, rather than the distance between positions \citep{hoff2008modeling}:
\begin{align}
&\P(y_{ij} = 1 ) = \P \left( \x_{ij}^T \bbeta + a_i + a_j  + \u_i^T \Lambda \u_j + \xi_{ij} > 0 \right), \quad
\label{eq_binary_eigent} \\
&a_i \stackrel{iid}{\sim} (0, \sigma^2_a),
\hspace{.3in}
\u_i \stackrel{iid}{\sim} (0, \bSigma_u),
\hspace{.3in}
\xi_{ij} \stackrel{iid}{\sim} {\rm N}(0, \sigma^2). \nonumber
\end{align}
In the context of \eqref{eq_gen_lv}, the function $f_{\btheta}(\v_{i}, \v_{j}) = a_i + a_j  + \u_i^T \Lambda \u_j $ for the latent eigenmodel, where the parameters $\btheta$ are the entries in $\Lambda$ and $\v_{i} = [a_i, \u_i]^T \in \R^{K+1}$.
\cite{hoff2008modeling} shows that the latent eigenmodel is capable of representing transitivity, and that the latent eigenmodel generalizes the latent position model given sufficiently large dimension of the latent vectors $K$. 

In addition to transitivity, a second phenomenon observed in social networks is structural  equivalence, wherein different groups of actors in the network form relations in a similar manner to others in their group. One form of structural equivalence is associative community structure, where the social network may be divided into groups of nodes that share many relations within group, but relatively few relations across groups. Such behavior is common when cliques are formed in high school social networks, or around subgroups in online social networks. A form of structural equivalence is when actors in a given group are more likely to form relations with actors in other groups than with actors in their own group, for example, in networks of high-functioning brain regions when performing cognitively demanding tasks \citep{betzel2018non}. 
Two models that are aimed at identifying subgroups of nodes that are structurally equivalent are the latent class model of \cite{nowicki2001estimation} and the mixed membership stochastic blockmodel \citep{airoldi2008mixed}. \cite{hoff2008modeling} shows that the latent eigenmodel is capable of representing stochastic equivalence in addition to transitivity, and that the latent eigenmodel generalizes latent class models given sufficiently large dimension of the latent vectors $K$. For this reason, we focus on the latent eigenmodel, and the simpler social relations model, as reference models in this paper.

\subsection{Drawbacks}
The latent variable network models discussed in this section were developed based on the patterns often observed in real world social networks. Latent variable network models contain different terms to represent the social phenomena underlying these patterns, and thus, 
different models may lead to substantially different estimates of $\bbeta$. It may not be clear which model's estimate of $\bbeta$, or which model's prediction of $\{ y_{ij} \}_{ij}$, is best. Generally, latent variable network models are evaluated using goodness-of-fit checks \citep{hunter2008goodness}, rather than rigorous tests, and it is well-known that selecting informative statistics for the goodness-of-fit checks is challenging. The latent variable network models described in this section are typically estimated using a Bayesian Markov chain Monte Carlo (MCMC) approach, which may be slow, especially for large data sets. Some recent advances do directly attempt to maximize the likelihood of network models with latent spaces \citep{ma2020universal, zhang2022joint}, however, public software implementations of these methods do not appear available, and they require certain covariate types (relation-level and actor-level, respectively) and certain latent space structures, such as the Euclidean distance latent space. 
%%%%

%%%%
\section{Exchangeable network models}
\label{sec:exch}
To motivate the formulation of the proposed model, we briefly discuss the theory of exchangeable random network models and their relationship to latent variable network models. A 
random network model for $\{\epsilon_{ij} \}_{ij}$ is \emph{exchangeable} if the distribution of $\{\epsilon_{ij} \}_{ij}$ is invariant to permutations of the actor labels, that is, if 
\begin{align}
\P\left( \{\epsilon_{ij} \}_{ij} \right) = \P \left( \{\epsilon_{\pi(i) \pi(j) } \}_{ij} \right), \quad \label{eq:prop_exch}
\end{align}
for any permutation $\pi(.)$. 
There is a rich theory of exchangeable network models, dating back to work on exchangeable random matrices \citep{hoover1979relations, aldous1981representations}, upon which we draw in this section. 

All the latent variable network models discussed in Section~\ref{sec:latent} have latent error networks $\{\epsilon_{ij} \}_{ij}$ that are exchangeable, where we define
$\epsilon_{ij} = f_{\btheta}(\v_{i}, \v_{j}) + \xi_{ij}$ from \eqref{eq_gen_lv}, the random portion of a general latent variable network model. Further,
under constant mean $\mu_{ij} = \mu$, all the latent variable network models for the observed network $\{ y_{ij} \}_{ij}$ in Section~\ref{sec:latent} are exchangeable.  In fact, 
any exchangeable network model may be represented by a latent variable network model. Specifically, the theory of exchangeable network models states that every exchangeable random network model may be represented in the following form (see, for example, \cite{lovasz2006limits, kallenberg2006probabilistic}):
\begin{align}
&\P(\epsilon_{ij} = 1 ) = \P \left( \mu + h (u_i, u_j) + \xi_{ij} > 0 \right), \quad  \label{eq_graphon_normal} \\
&u_i \stackrel{iid}{\sim} \text{Uniform}(0, 1),
\hspace{.3in}
\xi_{ij} \stackrel{iid}{\sim} \rm{N}(0, \sigma^2),
\nonumber
\end{align}
where the function  $h : [0,1] \times [0,1] \rightarrow \R$ has finite integral $\int_{[0,1] \times [0,1]} h(u,v) du dv < \infty$ and serves to distinguish the various exchangeable network models.
It can be shown that \eqref{eq_graphon_normal} is equivalent to the graphon representation of exchangeable random network models, where the graphon is the canonical probabilistic object of exchangeable random network models \citep{lovasz2006limits, borgs2014p}.
%and the $\xi_{ij}$ are traditionally uniformly distributed
Noting that we may always map the  random scalar $u_i$ to some random vector $\v_i$,
the expression in \eqref{eq_graphon_normal} illustrates how every exchangeable random network model may be represented by a latent variable network model in the sense of \eqref{eq_gen_lv}.

\subsection{Covariance matrices of exchangeable network models}
The expression in \eqref{eq_graphon_normal} shows that any exchangeable network model for binary network data must correspond to a latent random network $\{\epsilon_{ij} \}_{ij}$ that is continuous and exchangeable. 
 The covariance matrix of \emph{any} undirected exchangeable network model has the same form and contains at most two unique nonzero values (\cite{marrs2017standard} shows that directed exchangeable network models with continuous values all have covariance matrices of the same form with at most five unique nonzero terms). This fact can be seen by simply considering the ways that any pair of relations can share an actor. In addition to a variance, the remaining covariances are between relations that do and do not share an actor:
\begin{align}
var[\epsilon_{ij}] = \sigma^2_\epsilon, \hspace{.3in} cov[\epsilon_{ij}, \epsilon_{ik}] := \rho, \hspace{.3in} cov[\epsilon_{ij}, \epsilon_{kl}] = 0, \quad \label{eq_cov_exch}
\end{align} 
where the indices $i,j,k,$ and $l$ are unique. 
It is easy to see the second equality holds for any pair of relations that share an actor by the exchangeability property, i.e. by permuting the actor labels. The third equality results from the fact that the only random elements in \eqref{eq_graphon_normal} are the actor random variables $u_i$, $u_j$, and the random error $\xi_{ij}$. When the random variables corresponding to two relations $\epsilon_{ij}$ and $\epsilon_{kl}$ share no actor, the pair of relations are independent by the generating process. 
Finally, we note that exchangeable network models have relations that are marginally identically distributed, and thus relations therein have the same expectation and variance. 
That said, in the generalized linear regression case of \eqref{eq_gen_lv}, the means $\mu_{ij} = \x^T_{ij} \bbeta$ are non-constant and thus the observations $\{y_{ij} \}_{ij}$ are not exchangeable; only the latent error network $\{ \epsilon_{ij} \}_{ij}$ is exchangeable in the generalized linear regression case. In the proposed model, rather than put forth a particular parametric model for the latent network $\{ \epsilon_{ij} \}_{ij}$, we simply model the covariance structure outlined in \eqref{eq_cov_exch}, which is sufficient to represent the covariance structure of \emph{any} exchangeable network model for the errors. 
%%%%

%%%%
\section{The Probit Exchangeable (PX) model}
\label{sec:PX}
In this section, we propose the probit exchangeable network regression model, which we abbreviate the ``PX'' model. In the PX model, the vectorized mean of the network is characterized by a linear combination of covariates, $\X \bbeta$,  where $\bbeta$ is a $p$-length vector of coefficients that are the subject of inference and $\X$ is a $\binom{n}{2} \times p$ matrix of covariates.  The excess network dependence beyond that captured in $\X\bbeta$ is represented by an unobservable mean zero error vector $\bepsilon$, a vectorization of $\{ \epsilon_{ij} \}_{ij}$, that is exchangeable in the sense of \eqref{eq:prop_exch}. The PX model is
\begin{align}
&\P(y_{ij} = 1) = \P \left(\x_{ij}^T\bbeta +  \epsilon_{ij} > 0 \right),  %\hspace{.2in} i,j \in \{ 1,2,\ldots,n\}, i \neq j, 
\quad \label{eq:PX}\\
&\bepsilon \sim {\rm N}( \mbf{0}, \bOmega),  \nonumber  %\label{eq_graphon_normal}
\end{align}
where we note that the variance of $\epsilon_{ij}$ is not identifiable, and thus we choose $var[\epsilon_{ij}] = 1$ without loss of generality. We focus on normally-distributed unobserved errors $\bepsilon$ in this paper, however, other common distributions, such as the logistic distribution, could be used. We note that the normal distribution assumption implies that \eqref{eq:PX} is a typical probit regression model, but with correlation among the observations due to network structure.

As discussed in Section~\ref{sec:exch}, under the exchangeability assumption, the covariance matrix of the latent error network $var[\bepsilon] = \bOmega$ has at most two unique nonzero parameters. Taking $var[\epsilon_{ij}] = 1$,  the covariance matrix of $\bepsilon$ has a single parameter $\rho = cov[\epsilon_{ij}, \epsilon_{ik}]$.
We may thus write
\begin{align}
\bOmega(\rho) = \S_1 + \rho \s{S}_2, \quad \label{eq:omega}
\end{align}
where we define the binary matrices $\{ \S_{i} \}_{i=1}^3$ indicating unique entries in $\bOmega$. The matrix $\S_1$ is a diagonal matrix indicating the locations of the variance in $\bOmega$, and $\S_2$ and $\S_3$ indicate the locations in $\bOmega$  corresponding to the covariances $cov[\epsilon_{ij}, \epsilon_{ik}] $, and $cov[\epsilon_{ij}, \epsilon_{kl}]$, respectively, where the indices $i,j,k,$ and $l$ are unique.

The PX model unifies many of the latent variable network models discussed in Sections~~\ref{sec:latent}~and~\ref{sec:exch}.
Similar to \eqref{eq_graphon_normal}, the PX model may be seen as representing the covariance structure of the latent variables $\{ f_{\btheta}(\v_i, \v_j) + \xi_{ij} \}_{ij}$ with $\{ \epsilon_{ij} \}_{ij}$, the unobservable error network of the PX model in \eqref{eq:PX}. As both networks $\{ f_{\btheta}(\v_i, \v_j) + \xi_{ij} \}_{ij}$ and $\{ \epsilon_{ij} \}_{ij}$ are exchangeable, they have covariance matrices of the same form (see discussion in Section~\ref{sec:exch}).
As every exchangeable random network model may be represented by a latent variable network model, the PX model may represent the latent correlation structure of \emph{any} exchangeable network model, yet without specifying a particular exchangeable model. Further, we now show that the PX model is equivalent to the social relations model under certain conditions.
\begin{proposition}
\label{prop_srm}
Suppose that the random effects $\{a_i \}_{i=1}^n$ for the social relations model in \eqref{eq_binary_srm} are normally distributed. 
Then, there exists $\rho \in [0, 1/2]$ such that 
$\{y_{ij} \}_{ij}$ in the PX model in \eqref{eq:PX} is equal in distribution to $\{y_{ij} \}_{ij}$ as specified by the social relations model in \eqref{eq_binary_srm}.
\end{proposition}
\begin{proof}
As the PX and social relations models are probit regression models with the same mean structure, given by $\X \bbeta$, 
it is sufficient to show that their latent covariance matrices are equivalent, that is, that $var[ \{a_i + a_j + \xi_{ij} \}_{ij} ] = var[ \{ \epsilon_{ij} \}_{ij} ]$. By exchangeability, the latent covariance matrices of the PX and social relations models have the same form and by assumption have variance 1. It is easy to see that, given $\sigma^2_a \le 1$ (a necessary condition for $var[\epsilon_{ij}]=1$), we may take $\rho = \sigma^2_a/2$ for the PX model, which establishes equality in the model distributions. 
\end{proof}
Exact distributional equivalence between the PX model and latent variable models other than the social relations model will typically not hold. For example, the latent eigenmodel in \eqref{eq_binary_eigent} includes non-Gaussian random variables, so that exact distributional equivalence is impossible. Similarly, it appears likely that the general latent variable model in \eqref{eq_gen_lv} may generate non-Gaussian random variables through the function $f_{\btheta} (\v_i, \v_j)$. Importantly however, there does exist $\rho$ such that the covariance of the latent errors of every pair of relations, $cov[\epsilon_{ij}, \, \epsilon_{kl}]$, is equal to the covariance of the latent errors in \emph{any} exchangeable latent variable model, $cov[f_{\btheta} (\v_i, \v_j) + \xi_{ij},\, f_{\btheta} (\v_k, \v_l) + \xi_{kl}]$. Hence, the PX model may be seen as a generalized exchangeable latent variable model that focuses all modelling effort on the first two moments of the data.

Proposition~\ref{prop_srm} states that the PX model and social relations model are equivalent under normality of their latent error networks. 
In principle, the social relations model is simply a generalized linear mixed model, however, existing software packages, such as \texttt{lme4} in \texttt{R} \cite{lme4}, do not appear to accommodate the random effects specification of the social relations model in \eqref{eq_binary_srm} since the indices $i$ and $j$ pertain to random effects $a_i$ and $a_j$ from the same set (as opposed to $a_i$ and $b_j$ in a random crossed design). Nevertheless, the estimation scheme proposed in Section~\ref{sec:estimation}  employs the same strategies as those commonly used to estimate generalized linear mixed models  \citep{littell2006sas, gelman2006data}. In the estimation algorithm in \texttt{lme4}, the marginal likelihood of the data is approximated and then maximized using numerical approximations with respect to $\bbeta$ and random effects variance, for example $\sigma^2_a$ in the social relations model. Rather than an approximate likelihood, we propose maximizing the true likelihood with respect to $\bbeta$ and $\rho$,  yet also use numerical approximations to accomplish this maximization.

It is important to note that, although the latent errors $\{ \epsilon_{ij} \}_{ij}$ in the PX model form an exchangeable random network, the random network $y_{ij}$ represented by the PX model is almost certainly not  exchangeable. For example, each $y_{ij}$ may have a different marginal expectation $\Phi(\x_{ij}^T \bbeta)$. Then, the relations in the network are not marginally identically distributed, which is a necessary condition for exchangeability. Further, the covariances between pairs of relations, say $y_{ij}$ and $y_{ik}$, 
depend on the marginal expectations:
\begin{align}
cov[y_{ij}, y_{ik} ] &= E \left[ y_{ij} y_{ik} \right] - E \left[ y_{ij} \right]  E \left[ y_{ik}  \right] = \int_{-\x^T_{ij} \bbeta}^\infty \int_{-\x^T_{ik} \bbeta}^\infty dF_\rho - \Phi(\x^T_{ij} \bbeta)  \Phi(\x^T_{ik} \bbeta). \nonumber
% \nonumber
\end{align}
Here, $dF_\rho$ is the bivariate standard normal distribution with correlation $\rho$. Since the covariance $cov[y_{ij}, y_{ik} ] $ depends on the latent means $\x^T_{ij} \bbeta$ and $\x^T_{ik} \bbeta$, $cov[y_{ij}, y_{ik}  ] $ is only equal to $cov[y_{ab}, y_{ac} ]$ when the latent means are equal. As a result, although the covariance matrix of the unobserved errors $\bOmega$ is of a simple form with entries $\{1, \rho, 0 \}$,
the covariances between elements of the vector of observed relations $\y$ are heterogeneous (in general) and depend on $\rho$ in a generally more complicated way.
%%%%

%%%%
\section{Estimation}
\label{sec:estimation}
In this section, we propose an estimator of $\{\bbeta, \rho \}$ in the PX model that approximates  the maximum likelihood estimator (MLE). The algorithm we propose is based on the expectation-maximization (EM) algorithm \citep{dempster1977maximum}. 
Although the covariance matrix for the PX model is highly structured, as in \eqref{eq:omega}, a closed-form expression for the MLE does not appear available. While we explored pseduo-likelihood pairwise approximations (also called ``composite likelihoods'' in some literature) to the complete PX likelihood \citep{heagerty1998composite}, we found no substantial advantage  -- neither in performance nor runtime -- over the proposed estimation scheme in this paper.  

The proposed estimation algorithm consists of alternating computation of the expected complete likelihood with maximization with respect to $\rho$ and $\bbeta$, iterating until convergence. 
Since the algorithm iterates expectation and two maximization steps, we term it the EMM algorithm. 
To improve algorithm efficiency, we initialize $\bbeta$ at the ordinary probit regression estimator (assuming independence of the latent errors), and initialize $\rho$ with a mixture estimator based on possible values of $\rho$ such that $\bOmega$ is positive definite, as detailed in Appendix~\ref{sec:rho_init}.
The complete EMM algorithm is presented in Algorithm~\ref{alg:sub}.  In the following text, we detail the EMM algorithm, beginning with maximization with respect to $\rho$, and then proceeding to maximization with respect to $\bbeta$. We define $\gamma_i = E[\bepsilon^T \s{S}_i \bepsilon \mid \y,  \hat{\rho}^{(\nu)}, \hat{\bbeta}^{(\nu)} ] / |\Theta_i|$, where $\Theta_i$ is the set of relation pairs indicated by binary matrices $\s{S}_i$. 
By default, we typically set $\tau = 10^{-2}$ and $\delta = 10^{-1}$.

\begin{algorithm}
\caption{EMM estimation of the PX model}
\label{alg:sub}

\begin{enumerate}
\setcounter{enumi}{-1}
\item \textbf{Initialization:} \\ Initialize $\hat{\bbeta}^{(0)}$ using probit regression assuming independence and initialize $\hat{\rho}^{(0)}$ as described in Appendix~\ref{sec:rho_init}. Set positive convergence threshold $\tau$, scaling $\delta \in [0,1]$
%$\tau_\beta$, $\tau_\rho$, 
and set iteration $\nu = 0$. 

\item \textbf{Expectation step:} \\
Given $\hat{\rho}^{(\nu)}$ and $\hat{\bbeta}^{(\nu)}$, compute $E[\bepsilon \, | \, \y, \hat{\rho}^{(\nu)}, \hat{\bbeta}^{(\nu)} ]$ using the procedure described in Appendix~\ref{sec:appx_beta}, and approximate $\{ \gamma_i \}_{i=1}^3$ as described in  Appendix~\ref{sec:rho_linear_appx}.

\item \textbf{Maximization with respect to }$\rho$: \\
Given $s=0$ and $\hat{\rho}^{(\nu, \, s)} = \hat{\rho}^{(\nu)}$, $\hat{\bbeta}^{(\nu)}$, and $\{ \gamma_i \}_{i=1}^3$, compute $\hat{\rho}^{(\nu, \, s + 1)}$ by alternating \eqref{eq:rho_est} and \eqref{eq:lambda_est} until $\rho$ changes by less than $\delta \tau$. Set $\rho^{(\nu+1)}$ equal to the final $\rho$ value. 

\item \textbf{Maximization with respect to }$\bbeta$: \\
Compute the updated estimate
\[\hat{\bbeta}^{(\nu + 1)} = \hat{\bbeta}^{(\nu)} + (\X^T \bOmega^{-1} \X)^{-1} \X^T \bOmega^{-1} E[\bepsilon \, | \, \y,  \hat{\rho}^{(\nu)}, \hat{\bbeta}^{(\nu)}]. \nonumber \]

\item If $\max \{ | \hat{\bbeta}^{(\nu + 1)} - \hat{\bbeta}^{(\nu )} | / \hat{\bbeta}^{(\nu )},  \ | \hat{\rho}^{(\nu+1)} - \hat{\rho}^{(\nu)}| / \hat{\rho}^{(\nu)}\} > \tau$, then increment $\nu$ by 1 and return to Step 1. Otherwise, end. 
\end{enumerate}
\end{algorithm}

\subsection{Expectation}
Consider the log-likelihood, $\ell_\z$, of the latent continuous random vector $\z$. Taking the expectation of $\ell_\z$ conditional on $\y$, the expectation step for a given iteration $\nu$ of the EM algorithm is
\begin{align}
&E[\ell_\z  \, | \,  \y,  \rho=\hat{\rho}^{(\nu)}, \bbeta=\hat{\bbeta}^{(\nu)} ]  = \nonumber \\
& \hspace{.25in} -\frac{1}{2}{ \rm log } 2 \pi |\bOmega| - \frac{1}{2} E\left[ (\z - \X \bbeta)^T \bOmega^{-1} (\z - \X \bbeta) \, | \, \y,  \rho=\hat{\rho}^{(\nu)}, \bbeta=\hat{\bbeta}^{(\nu)} \right], \quad \label{eq:EM_exp}
\end{align}
where $\hat{\rho}^{(\nu)}$ and $\hat{\bbeta}^{(\nu)}$ are the estimators of $\rho$ and $\bbeta$ at iteration $\nu$. 
In discussing the maximization step for $\rho$, we will show that that the $\rho$ update depends on the data through the expectations denoted by $\gamma_i$ for $i \in \{1,2,3 \}$. In discussing the maximization step for $\beta$, we will show that that the $\beta$ update depends on the data only through the expectation $E[\bepsilon \, | \, \y, \hat{\rho}^{(\nu)}, \hat{\bbeta}^{(\nu)} ]$. \\

\noindent\textbf{Approximations:} \\
The computation of $E[\bepsilon \, | \, \y, \hat{\rho}^{(\nu)}, \hat{\bbeta}^{(\nu)}]$ in \eqref{eq:beta_update} is nontrivial, as it is a $\binom{n}{2}$-dimensional truncated multivariate normal integral. We exploit the structure of $\bOmega$ to compute $E[\bepsilon \, | \, \y, \hat{\rho}^{(\nu)}, \hat{\bbeta}^{(\nu)}]$ using the law of total expectation. A Newton-Raphson algorithm, along with an approximate matrix inverse, are employed to compute an approximation of $E[\bepsilon \, | \, \y, \hat{\rho}^{(\nu)}, \hat{\bbeta}^{(\nu)}]$. Details of the implementation of thse approximations are given in Appendix~\ref{sec:appx_beta}.

The expectations  $\{ \gamma_i \}_{i=1}^3$ require the computation of $\binom{n}{2}$-dimensional truncated multivariate normal integrals, which are onerous for even small networks.
Thus, we make two approximations to  $\{ \gamma_i \}_{i=1}^3$ to reduce the runtime of the EMM algorithm. First, we compute the expectations conditioning only on the entries in $\y$ that correspond to the entries in $\bepsilon$ being integrated, for example, instead of computing $E[\epsilon_{jk} \epsilon_{lm} \, | \, \y]$, we compute $E[\epsilon_{jk} \epsilon_{lm} \, | \, y_{jk}, y_{lm}]$. This first approximation is most appropriate when $\rho$ is small, since $y_{lm}$ is maximally informative for $\epsilon_{jk}$ when $\rho$ is large (for $l,m,j,$ and $k$ distinct). Second, we find empirically that  $\gamma_2 = E[\bepsilon^T \S_2 \bepsilon \, | \, \y ] / |\Theta_2|$ is approximately linear in $\rho$, since this sample mean of conditional expectations concentrates around a linear function of $\rho$. % concentrates around 
Thus, we compute  $\gamma_2$ for $\rho=0$ and $\rho = 1$, and use a line connecting these two values to compute $\gamma_2$ for arbitrary values of $\rho$ (see evidence of linearity of $\gamma_2$ for the political books network in Appendix~\ref{sec:app_data}). 
The details of the approximations to  $\{ \gamma_i \}_{i=1}^3$  are given in Appendix~\ref{sec:rho_linear_appx}.

\subsection{Maximization with respect to $\rho$}
\label{sec:max_rho}

To derive the maximization step for $\rho$, we use the method of Lagrange multipliers, since differentiating \eqref{eq:EM_exp} directly with respect to $\rho$ gives complex nonlinear equations that are not easily solvable. 
We first define the set of parameters $\{\phi_i \}_{i=1}^3$, representing the variance and two possible covariances in $\bOmega$,
\begin{align}
var[\epsilon_{ij}] = \phi_1, \hspace{.3in}  cov[\epsilon_{ij}, \epsilon_{ik}] = \phi_2 = \rho, \hspace{.3in} cov[\epsilon_{ij}, \epsilon_{kl}] = \phi_3, \nonumber
\end{align} 
where the indices $i,j,k,$ and $l$ are distinct. 
In addition, we let $\p =  [p_1, p_2, p_3]$ parametrize the precision matrix $\bOmega^{-1} = \sum_{i=1}^3 p_i \S_i$, which has the same form as the covariance matrix $\bOmega$  (see  \cite{marrs2017standard} for a similar result when $\{ \epsilon_{ij} \}_{ij}$ forms a directed network). The objective function, incorporating the restrictions that $\phi_1 = 1$ and $\phi_3 = 0$, is \begin{align}
Q_\y(\bphi) := E[\ell_\z \, | \, \y ] + \frac{1}{2}\lambda_1 (\phi_1 - 1) + \frac{1}{2}\lambda_3 \phi_3, \nonumber
\end{align} 
where $\bphi = [\phi_1, \phi_2, \phi_3]$ and the `$\frac{1}{2}$' factors are included to simplify algebra. Then, differentiating $Q_\y$ with respect to $\p$, $\lambda_1$, and $\lambda_3$,
the estimators for $\rho$, $\{ \lambda_1, \lambda_3 \}$ are
\begin{align}
\hat{\rho} &= \gamma_2 -   \frac{1}{|\Theta_2|}
\begin{bmatrix} 
\frac{\partial \phi_1}{\partial p_2}  & \frac{\partial \phi_3}{\partial p_2}
\end{bmatrix}^T
\begin{bmatrix} 
\lambda_1 \\
\lambda_3
\end{bmatrix} \quad \label{eq:rho_est} \\
\begin{bmatrix} 
\hat{\lambda}_1 \\
\hat{\lambda}_3
\end{bmatrix}
&= \begin{bmatrix} 
\frac{\partial \phi_1}{\partial p_1}  & \frac{\partial \phi_3}{\partial p_1} \\
\frac{\partial \phi_1}{\partial p_3}  & \frac{\partial \phi_3}{\partial p_3}
\end{bmatrix}^{-1}
\begin{bmatrix} 
|\Theta_1|  & 0 \\
0  & |\Theta_3|
\end{bmatrix}
\begin{bmatrix} 
\gamma_1 - 1 \\
\gamma_3
\end{bmatrix}, \quad \label{eq:lambda_est}
\end{align} 
where again $ \Theta_i$ is the set of pairs of relations $(jk, lm)$ that share an actor in the $i^\text{th}$ manner, for $i \in \{1,2,3 \}$. For instance, $\Theta_2$ consists of pairs of relations of the form $(jk, jl)$, where $j,k,$ and $l$ are distinct indices. In \eqref{eq:rho_est} and \eqref{eq:lambda_est}, the partial derivatives $\left\{ \partial \phi_i / \partial p_j \right\}$ are available in closed form and are easily computable in $O(1)$ time using the forms of $\bOmega$ and $\bOmega^{-1}$. See Appendix~\ref{sec:undir_cov_mat} for details.

Alternation of the estimators for $\rho$ and $\{ \lambda_1, \lambda_3 \}$ in \eqref{eq:rho_est} and \eqref{eq:lambda_est} constitutes a block coordinate descent for $\rho = \phi_2$ subject to the constraints $\phi_1 = 1$ and $\phi_3 = 0$. 
This block coordinate descent makes up the maximization step of the EMM algorithm for $\rho$.

\subsection{Maximization with respect to ${\beta}$}
\label{sec:max_bbeta}

The maximizating step with respect to $\bbeta$ in the EMM algorithm can be obtained directly. 
Setting the derivative of \eqref{eq:EM_exp} with respect to $\bbeta$ equal to zero, 
the maximization step for $\bbeta$ is
\begin{align}
\hat{\bbeta}^{(\nu+1)} & = \hat{\bbeta}^{(\nu)} +
 \left( \X^T \bOmega^{-1} \X \right)^{-1} \X^T \bOmega^{-1} E[\bepsilon \, | \, \y,  \hat{\rho}^{(\nu)}, \hat{\bbeta}^{(\nu)} ], \quad \label{eq:beta_update}
\end{align}
where we use the identity $\bepsilon = \z - \X\bbeta$.
In Appendix~\ref{sec:undir_cov_mat}, we show that the leading terms of the unique entries in $\Omega^{-1}$, $\p$, depend only on $\rho$ through a multiplicative factor, 
\begin{align}
\p \approx f(\rho)[g_1(n),\, g_2(n), \,g_3(n)]^T. \nonumber
\end{align}
Thus, we may factor $f(\rho)$ out of \eqref{eq:beta_update}, and the $\bbeta$ maximization is asymptotically $\rho$-free (except for the expectation term). Similarly, the maximization with respect to $\rho$ in Section~\ref{sec:max_rho} is free from $\bbeta$ except for the expectation term. Hence, only a single maximization step with respect to each $\bbeta$ and $\rho$ is required for each expectation.

\subsection{Consistency of the EMM estimator}
\label{sec:consistency}
The complete multivariate normal likelihood for $\z$ is a non-curved, identifiable likelihood. Then, it is known that each  expectation and maximization step in an EM algorithm increases the current likelihood value \citep{wu1983convergence}. Whenever there is a unique, single local maximum, the EM algorithm yields consistent, and efficient, estimators. We make a series of approximations to the expectations in the EMM algorithm to reduce computational demands, so that the theory of EM estimator convergence may not be directly applicable. Yet, we find that the EMM estimators, $\{\hat{\beta}_{EMM}, \hat{\rho}_{EMM} \}$ maintain consistency. Taking the leading terms of $\hat{\rho}_{EMM}$, 
\begin{align}
    \hat{\rho}_{EMM} &= \frac{1}{2} + \frac{1}{n^3} \sum_{jk, lm \in \Theta_2} E[\epsilon_{jk} \epsilon_{lm} \mid y_{jk}, y_{lm} ] - \frac{1}{n^2} \sum_{jk} E[\epsilon_{jk}^2 \mid y_{jk} ]\ldots \quad \nonumber \\ %\label{eq:rho_bc_em_MAIN} \\
    & \ldots - \frac{2}{n^4} \sum_{jk, lm \in \Theta_3} E[\epsilon_{jk} \mid y_{jk}] E[\epsilon_{lm} \mid y_{lm}] + O(n^{-1}). \quad \nonumber
\end{align}
which has expectation $E[\hat{\rho}_{EMM}] = \rho + O(n^{-1})$. Then, consistency can be established by showing that the variance of $\hat{\rho}_{EMM}$ tends to zero. We provide details in Appendix~\ref{sec:theory}.
The estimator $\hat{\beta}_{EMM}$ is particularly difficult to analyze, as $E[\epsilon_{jk} \mid \y]$ depends on every entry in $\y$, and because we approximate this expectation. We provide a sketch for a proof of consistency in Appendix~\ref{sec:theory} by bounding the distance between the EMM estimator for $\beta$ and the true MLE, 
$||\hat{\beta}_{EMM} - \hat{\beta}_{MLE} ||_2^2$, using an easier-to-analyze estimator for $\beta$ which replaces $E[\epsilon_{jk} \mid \y]$ with $E[\epsilon_{jk} \mid y_{jk}]$.
As in the argument for consistency of $\hat{\rho}_{EMM}$, we establish consistency of the bounding estimator by showing the expectation is asymptotically equal to the true value of $\beta$ and that the variance of the bounding estimator tends to zero. 
We also discuss performance of $\hat{\beta}_{EMM}$ under model misspecification, showing that it maintains consistency even under violation of the normality and exchangeability assumptions. 
%%%%

%%%%
\section{Prediction}
\label{sec:pred_proc}
In this section, we describe how to use the PX model, and the approximations in service of Algorithm~\ref{alg:sub}, to make predictions for an unobserved network relation without undue computational cost.
The predicted value we seek is the probability of observing $y_{jk} = 1$ given all the other values $\y_{-jk}$, where $\y_{-jk}$ is the vector of observations $\y$ excluding the single relation  $jk$. 
As in estimation, the desired probability is again equal to a $\binom{n}{2}$-dimensional multivariate truncated normal integral, which is computationally burdensome. Thus, we approximate the desired prediction probability
\begin{align}
\mathbb{P}( y_{jk} = 1 \, | \, \y_{-jk} ) 
%&= E \left[ \mathbbm{1} [\epsilon_{jk} > - \x^T_{jk} \bbeta  ] \, | \, \y_{-jk}, \X \right], \\
&= E \left[E \left[ \mathbf{1}[\epsilon_{jk}  > -\x^T_{jk} \bbeta  ] \, | \, {\bepsilon}_{-jk} \right] \, | \, \y_{-jk} \right], \quad \label{eq:phat} \\
&\approx \Phi \left(\frac{ E[{\epsilon}_{jk} \, | \, \y ]  + \x^T_{jk} \bbeta}{ \sigma_n} \right).\nonumber 
%\\ \nonumber
\end{align}
The approximation in \eqref{eq:phat} is based on the fact that  $[\epsilon_{jk} \, | \, {\bepsilon}_{-jk}]$ is normally distributed:
\begin{align}
&\epsilon_{jk} \, | \, {\bepsilon}_{-jk} \sim {\rm N}(m_{jk}, \sigma^2_n), \quad \label{eq_condl_distr_undir1}\\
m_{jk} &= - {\sigma_n^2} \mathbf{1}_{jk}^T \left( p_2 \s{S}_2 + p_3 \s{S}_3 \right) \tilde{\bepsilon}_{-jk}, 
\hspace{.2in} \sigma^2_n = \frac{1}{p_1 }, \nonumber 
% \nonumber
\end{align}
where $\mathbf{1}_{jk}$ is the vector of all zeros with a one in the position corresponding to relation $jk$ and, for notational simplicity, we define $\tilde{\bepsilon}_{-jk}$ is the vector $\bepsilon$ with a zero in the entry corresponding to relation $jk$. We note that the diagonal of the matrix $p_2 \s{S}_2 + p_3 \s{S}_3$ consists of all zeros so that $m_{jk}$ is free of $\epsilon_{jk}$. Then, the inner expectation in \eqref{eq:phat} is 
\begin{align}
E \left[ \mathbf{1}[\epsilon_{jk}  > -\x^T_{jk} \bbeta  ] \, | \, {\bepsilon}_{-jk} \right] = \Phi \left( \frac{m_{jk} + \x^T_{jk} \bbeta }{\sigma_n}\right). \quad \label{eq:inner_exp}
\end{align}
Of course, $m_{jk}$ depends on ${\bepsilon}_{-jk}$ which is unknown, and thus, we replace $m_{jk}$ with its conditional expectation $E[m_{jk} \, | \, \y_{-jk}] = E[\epsilon_{jk} \, | \, \y_{-jk}]$.

Computing $E[\epsilon_{jk} \, | \, \y_{-jk}]$ is extremely difficult, however computing $E[\epsilon_{jk} \, | \, \y]$ proves feasible if we exploit the structure of $\bOmega$.  Thus, we approximate the desired expectation by imputing $y_{jk}$ with the mode of the observed data:
\begin{align}
E \left[ \epsilon_{jk}  \, | \, \y_{-jk} \right] \approx
E \left[ \epsilon_{jk}  \, | \, \y_{-jk}, y_{jk} = y^* \right] = E \left[ \epsilon_{jk}  \, | \, \y \right], \quad \label{eq:impute_ymode}
\end{align}
where $y^*$ is the mode of $\y_{-jk}$. The error due to this approximation is small and  shrinks as $n$ grows. Substituting \eqref{eq:impute_ymode} for $m_{jk}$ in \eqref{eq:inner_exp} gives the final expression in \eqref{eq:phat}.
%%%%

%%%%
\section{Simulation studies}
\label{sec:sim}
In this section, we describe three simulation studies. The first verifies that the performance of the EMM estimator in Algorithm~\ref{alg:sub} provides improvement over standard probit regression. The second simulation study verifies consistency of the EMM estimators of $\bbeta$, and compares the performance of these estimators to the estimators of 
$\bbeta$ from the social relations model and the latent eigenmodel. The third simulation study evaluates the robustness of the PX model, and EMM algorithm, to the assumption that the latent random variables are normally distributed. 

For both simulation studies, we generated data with mean consisting of three covariates and an intercept: 
\begin{align}
y_{ij} = \mathbf{1} \Big[\beta_0 + \beta_1 \mathbf{1}[x_{1i} \in C] \mathbf{1}[x_{1j} \in C]+\beta_2|x_{2i}-x_{2j}|+\beta_3 x_{3ij}
+ \epsilon_{ij} > 0 \Big]. \quad \label{eq:sim_gen_model}
\end{align}
In the model in \eqref{eq:sim_gen_model}, $\beta_0$ is an intercept; $\beta_1$ is a coefficient on a binary indicator of whether individuals $i$ and $j$ both belong to a pre-specified class $C$; $\beta_2$ is a coefficient on the absolute difference of a continuous, actor-specific covariate $x_{2i}$; and $\beta_3$ is that for a pair-specific continuous covariate $x_{3ij}$. 
We fixed  $\bbeta = [\beta_0, \beta_1, \beta_2, \beta_3]^T$ at a single set of values.  Since the accuracy of estimators of $\boldsymbol{\beta}$ may depend on $\X$, we generated $20$ random design matrices $\X$ for each sample size  of $n \in \{ 20, 40, 80 \}$ actors. We emphasize that, although these may appear to be only moderately-sized networks, each consists of $\binom{n}{2} \in \{ 190, \, 780, \, 3160\}$ observations.
For each design matrix we simulated 100 error realizations of $\{ \epsilon_{ij} \}_{ij}$, 
with distribution that depended on the generating model. When generating from the PX model, half of the total variance in $\epsilon_{ij}$ was due to correlation $\rho = 1/4$, and the remaining half was due to the unit variance of $\epsilon_{ij}$. When generating from the latent eigenmodel in \eqref{eq_binary_eigent}, one third the variance in  $\epsilon_{ij}$ was due to each term $a_i + a_j$, $\u_i^T \Lambda \u_j$, and $\xi_{ij}$, respectively. 
For additional details of the simulation study procedures, see Appendix~\ref{sec:sim_cons}. 

\subsection{Evaluation of approximations in Algorithm~\ref{alg:sub}}
\label{sec:sim1}
% To evaluate the efficacy of the approximations described in the estimation procedure in Algorithm~\ref{alg:sub}, we conducted a simulation study comparing EMM and MLE estimators of $\beta$. We estimated $\beta$ in the standard probit model assuming independence between observations (which we abbreviate ``std. probit'') as a baseline. 
%We fixed $\X$ for the study to contain a single covariate (column) of independent Bernoulli$(p=0.3)$ random variables.  

To evaluate the efficacy of the approximations described in the estimation procedure in Algorithm~\ref{alg:sub}, we simulated from \eqref{eq:sim_gen_model} for a single $\X$ with $n=40$ (larger $n$ caused multivariate normal integral failures in \texttt{R}).
We simulated 100 networks from the PX model in \eqref{eq:PX} using this $\X$, for each value of $\rho \in \{0.1, 0.25, 0.4 \}$ (we note that we require $\rho < 1/2$ for the error covariance matrix $\Omega$ to be positive definite). For each realization, we estimated $\beta$ in the PX model using EMM in Algorithm~\ref{alg:sub}. 
To estimate $\beta$ in the standard probit model, we used the function \texttt{glm} in \texttt{R}. To compute the MLE, we numerically optimized the data log-likelihood using the Broyden–Fletcher–Goldfarb–Shanno (BFGS) algorithm as implemented in the \texttt{optim} function in \texttt{R}, initializing at the true values of $\{ \beta, \rho\}$. 
%Since this numerical optimization is computationally onerous, we simulated networks of size $n=15$ for this study.

\begin{figure}[ht]
\centering
\begin{tabular}{cc}
\includegraphics[width=.44\textwidth]{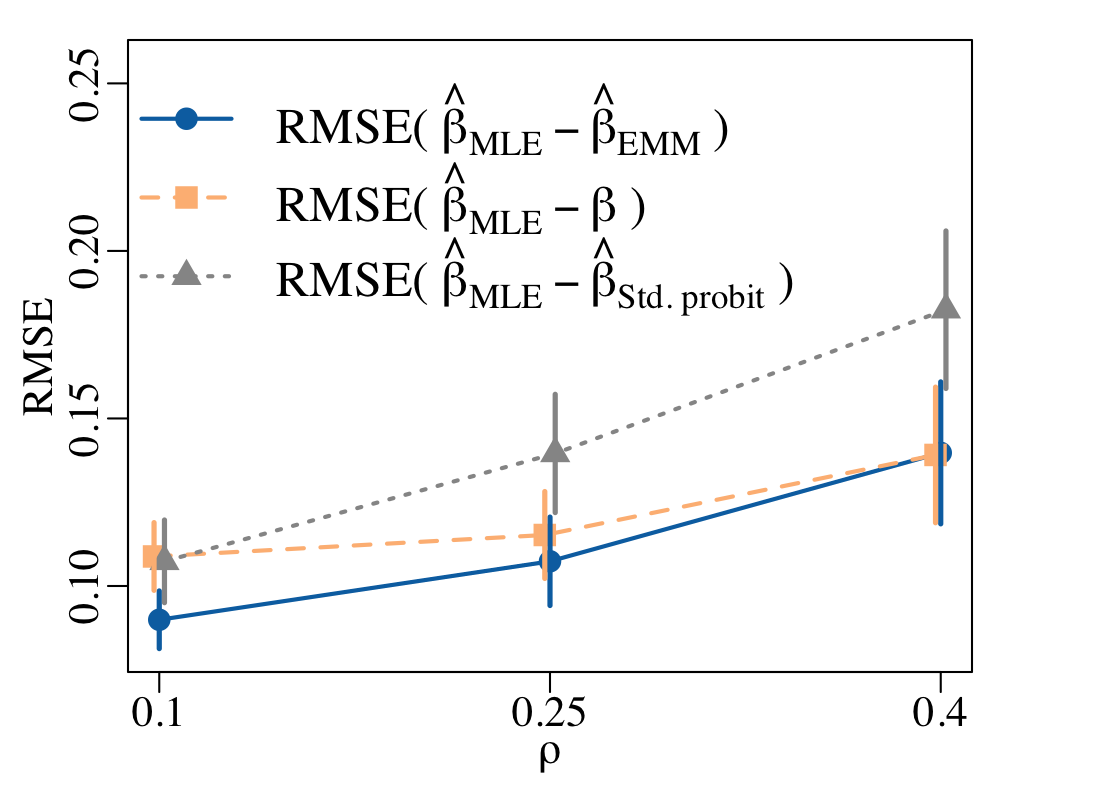}
\hspace{-.1in} 
& 
\hspace{-.1in} 
\includegraphics[width=.44\textwidth]{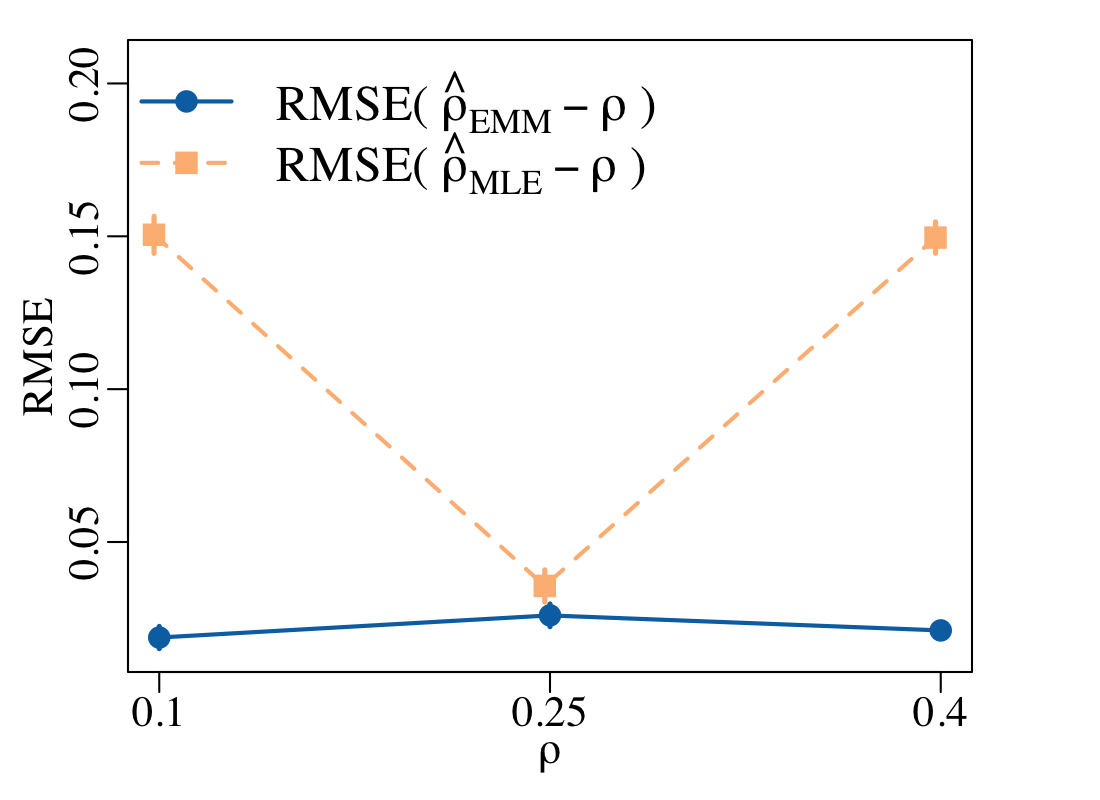} \\
\end{tabular}
  \caption{The left panel depicts performance in estimating $\beta$: RMSE between the EMM estimator and the MLE (RMSE$(\hat{\beta}_{MLE} - \hat{\beta}_{EMM})$), 
  between the MLE and the truth (RMSE$(\hat{\beta}_{MLE} - {\beta})$), and between the MLE and the standard probit estimator (RMSE$(\hat{\beta}_{MLE} - \hat{\beta}_{Std. probit})$). 
  The right panel depicts performance in estimating $\rho$: RMSE between the MLE and the EMM estimator  (RMSE$(\hat{\rho}_{MLE} - \hat{\rho}_{EMM})$) and between the MLE and the truth 
  (RMSE$(\hat{\rho}_{MLE} - \rho)$).
  The RMSEs are plotted as a function of the true values of $\rho$, and solid vertical lines denote Monte Carlo error bars. Some points obscure their Monte Carlo error bars. } 
    \label{fig:mle_sim}
  \end{figure}

In the left panel of Figure~\ref{fig:mle_sim}, we evaluate the performance of the EMM estimator  by comparing the root mean square error (RMSE) between the EMM coefficient estimate, $\hat{\beta}_{EMM}$, and the MLE obtained by the optimization procedure $\hat{\beta}_{MLE}$. As a baseline, we compute the RMSE between $\hat{\beta}_{MLE}$ and the true value $\beta$. If the approximations in the EMM algorithm are small, we expect the RMSE between $\hat{\beta}_{EMM}$ and $\hat{\beta}_{MLE}$ to be much smaller than the RMSE between $\hat{\beta}_{MLE}$ and $\beta$. Generally, the RMSE between 
$\hat{\beta}_{EMM}$ and $\hat{\beta}_{MLE}$ is  smaller than the RMSE between $\hat{\beta}_{MLE}$ and $\beta$. However, the discrepancy between the two RMSEs decreases as the true $\rho$ grows. As a reference, the MSE between $\hat{\beta}_{Std.\, probit}$ and $\hat{\beta}_{MLE}$ is also shown in the left panel of Figure~\ref{fig:mle_sim}; the EMM estimator is closer to $\hat{\beta}_{MLE}$ than the standard probit estimator is to $\hat{\beta}_{MLE}$ for all values of $\rho$.  Raw RMSE values between the estimators and the truth, shown in Figure~\ref{fig:mle_sim_2}, confirm that the EMM algorithm does perform  better than standard probit in RMSE with respect to estimation of $\beta$. The results of this simulation study suggest that the EMM algorithm improves estimation of $\beta$ over the standard probit estimator for $\rho > 0$, and that the EMM estimator is reasonably close to the MLE, signifying the approximations in the EMM algorithm are reasonable.  %It is worth noting that the approximations used in the EMM algorithm are best for large $n$, so we would expect results to improve as $n$ increases. 

In the right panel of Figure~\ref{fig:mle_sim},  the EMM estimator of $\rho$ is closer to the MLE, $\hat{\rho}_{MLE}$, than the MLE is close to the true value of $\rho$ for all values of $\rho$ examined. 
This fact suggests that the approximation error in estimating $\rho$ in the EMM algorithm is small.
%when $\rho < 0.4$. 
Further, the raw RMSE values shown in Figure~\ref{fig:mle_sim_2} illustrate that $\hat{\rho}_{EMM}$ may be as good an estimator of $\rho$ as is $\hat{\rho}_{MLE}$.
The approximations in $\hat{\rho}_{EMM}$ appear to be stable over the range of $\rho$ values examined.  
% On the other hand, when $\rho = 0.4$, the approximation error of the EMM algorithm is larger than the MSE between the MLE and the true value of $\rho$. Indeed, the approximation error of the EMM algorithm grows as the value of the true $\rho$ grows in Figure~\ref{fig:mle_sim}. 
%The difference in the MSE between  $\hat{\rho}_{MLE}$ and  $\hat{\rho}_{EMM}$ and the MSE between  $\hat{\rho}_{MLE}$ and $\rho$ again decreases as the true value of $\rho$ grows. 
% This trend and the similar trend in the left panel of Figure~\ref{fig:mle_sim} suggest that the approximations in the EMM algorithm slightly degrade as the true value of $\rho$ grows, at least for $n=15$. 
Overall, since the degradation in performance of the EMM algorithm is most pronounced in estimation of $\beta$, we postulate that the degradation may be due to the approximations in computing $E[\epsilon_{jk}  \mid \y]$ (see Appendix~\ref{sec:appx_beta}).
% \approx E[\epsilon_{jk} \epsilon_{lm} | y_{jk} y_{lm} ]$, which is expected to deteriorate as $\rho$ grows (see Section~\ref{sec:max_rho} and Appendix~\ref{sec:rho_linear_appx}).

\begin{figure}[h!]
\centering
\begin{tabular}{cc}
\includegraphics[width=.48\textwidth]{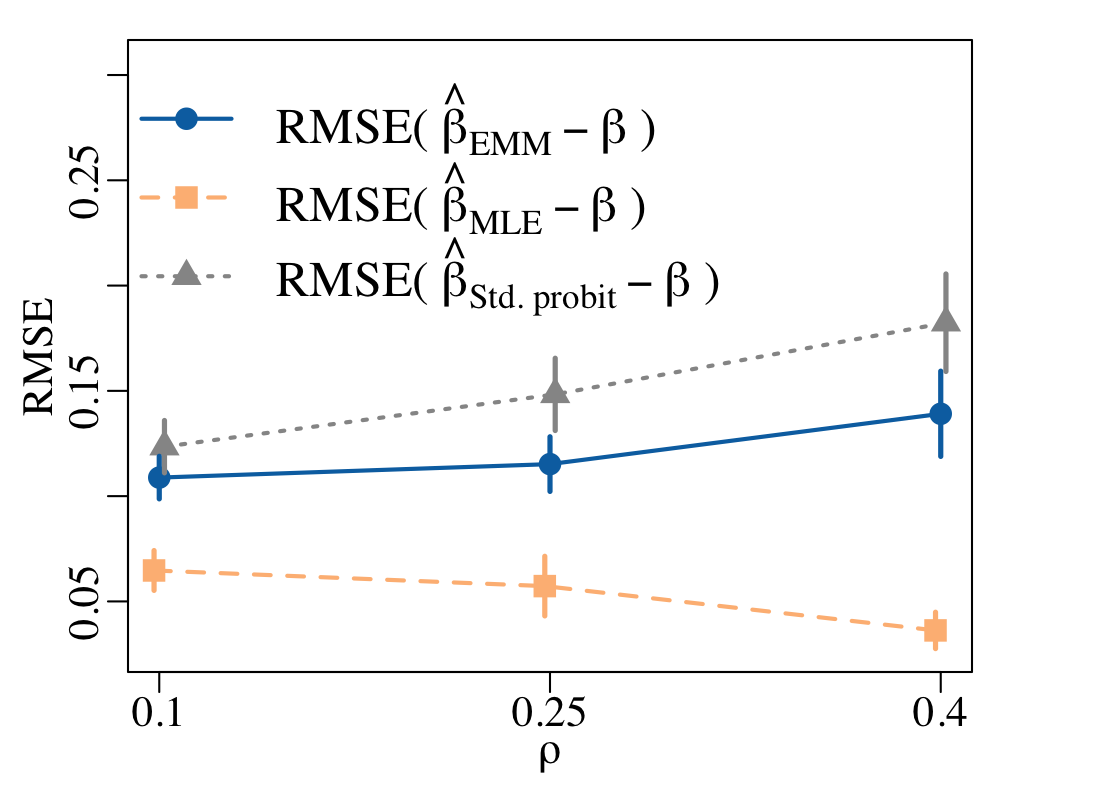}
\hspace{-.1in} & \hspace{-.1in} 
\includegraphics[width=.48\textwidth]{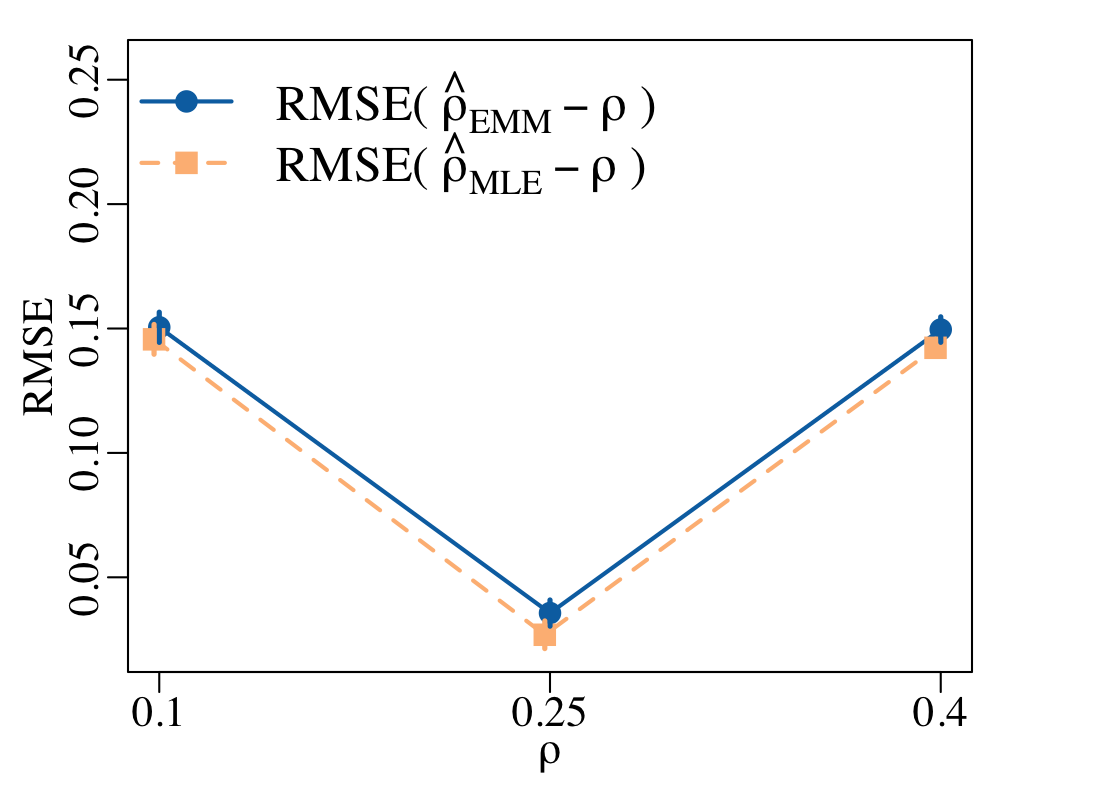} \\
\end{tabular}
  \caption{
  The left panel depicts the RMSE in estimating $\beta$ using the EMM algorithm, MLE, and standard probit regression. The right panel depicts the same for $\rho$. 
%   The left panel depicts performance in estimating $\beta$: MSE between the EMM estimator and the MLE (MSE$(\hat{\beta}_{optim} - \hat{\beta}_{EMM})$), between the MLE and the truth (MSE$(\hat{\beta}_{optim} - {\beta})$), and between the MLE and the Ordinary probit estimator (MSE$(\hat{\beta}_{optim} - \hat{\beta}_{ordinary})$). The right panel depicts performance in estimating $\rho$: MSE between the MLE and the EMM estimator  (MSE$(\hat{\rho}_{optim} - \hat{\rho}_{EMM})$) and between the MLE anad the truth (MSE$(\hat{\rho}_{optim} - \rho$).
The MSEs are plotted as a function of the true values of $\rho$, and solid vertical lines denote Monte Carlo error bars.  
 } 
    \label{fig:mle_sim_2}
  \end{figure}

\subsection{Performance in estimation of $\beta$}
\label{sec:sim_beta_est}
To evaluate the performance of the PX estimator in estimating linear coefficients $\bbeta$, we compared estimates of $\bbeta$ by the EMM algorithm to estimators of the social relations and latent eigenmodels on data generated from the PX model and data generated from the latent eigenmodel. 
We used the \texttt{amen} package in \texttt{R} to estimate the social relations model and latent eigenmodel \citep{amen}. 
We again compared these estimators to the standard probit regression model assuming independence as a baseline, which we estimated using the function \texttt{glm} in \texttt{R}. We focused on the value of $\rho=0.25$, in the center of the range of possible $\rho$ values.

In Figure~\ref{fig:mse_beta_sim}, we plot the RMSE (scaled by $n^{1/2}$ ) of the $\beta$ coefficients estimated for the PX model, standard probit model, social relations model, and latent eigenmodels. We
see that the EMM estimator for the PX model has a downward trend in $n^{1/2}$RMSE with $n$, and a reducing spread of $n^{1/2}$RMSE with $n$, for both the PX and latent eigenmodel generating models. These facts suggest that the PX estimator is consistent for $\bbeta$, at a rate $n^{1/2}$ or better, for both the PX and latent eigenmodel generating models, confirming the claims in Section~\ref{sec:consistency}. 
Further, 
the EMM estimator has the lowest median $n^{1/2}$RMSE of any of the estimators for all entries in $\bbeta$, where $n^{1/2}$RMSE is evaluated for each $\X$ realization (across the error realizations) and the median is computed across the 20 $\X$ realizations. We observe similar patterns for the correlation parameter $\rho$; see Appendix~\ref{sec:sim_cons}. 
Interestingly, the superiority of the PX estimator holds whether we generate from the PX or latent eigenmodel, which suggests that any benefit in correctly specifying the latent eigenmodel is lost in the estimating routine.
The larger $n^{1/2}$RMSEs of the \texttt{amen} estimator of the social relations and latent eigenmodels
are a result of bias; see Appendix~\ref{sec:sim_cons} for bias-variance decomposition of the MSEs.

\begin{landscape}
 \begin{figure}
\centering
\begin{tabular}{ccccc}
\begin{sideways} \hspace{1.1in} \textbf{PX} \end{sideways}
& \hspace{-.01in}
\includegraphics[width=.35\textwidth]{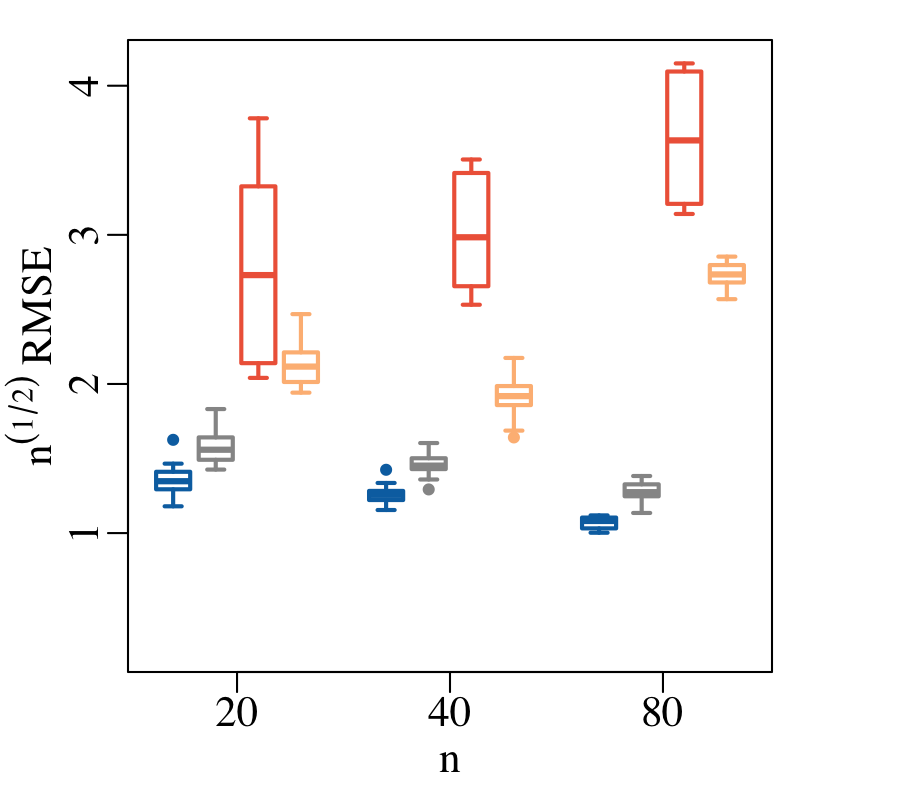}
\hspace{-.1in} & \hspace{-.1in}
\includegraphics[width=.35\textwidth]{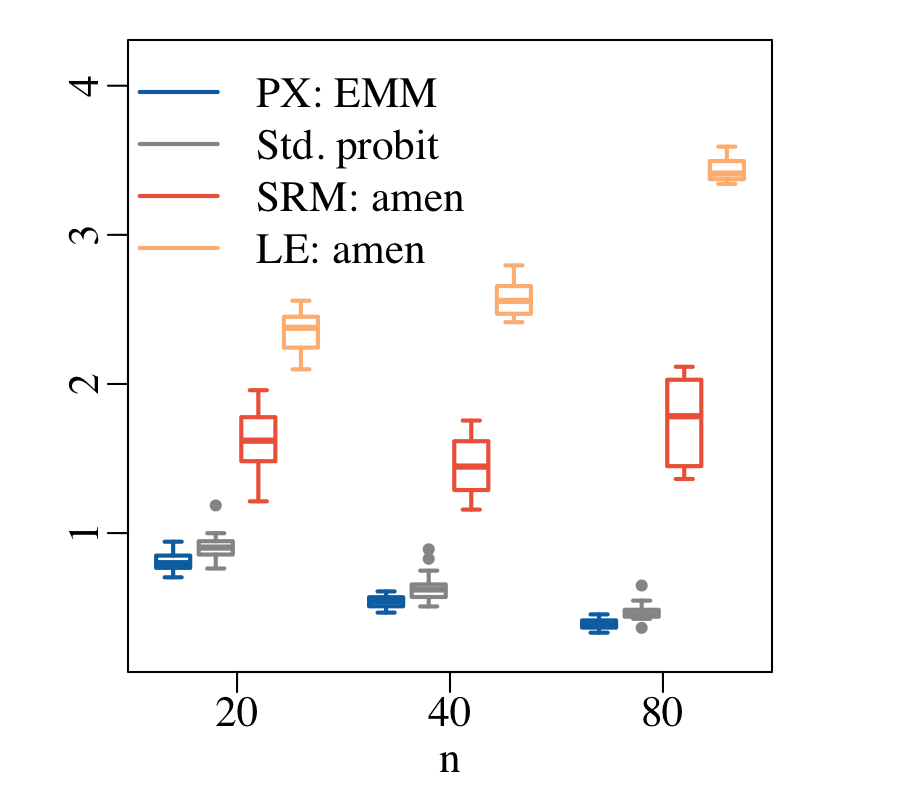}
\hspace{-.1in} & \hspace{-.1in}
\includegraphics[width=.35\textwidth]{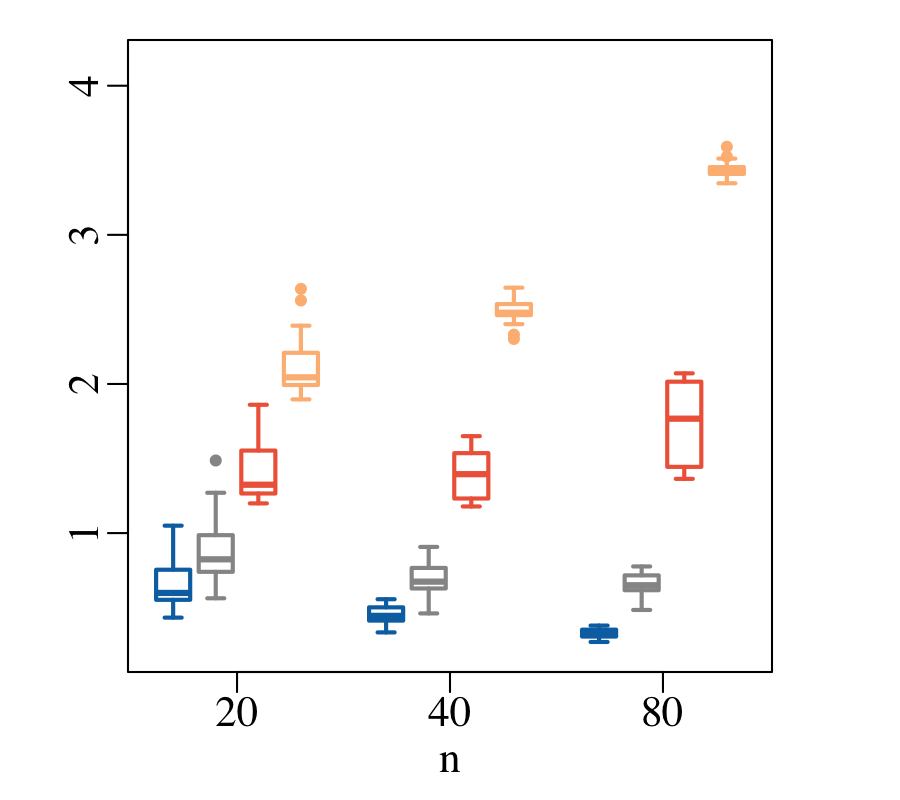}
\hspace{-.1in} & \hspace{-.1in}
\includegraphics[width=.35\textwidth]{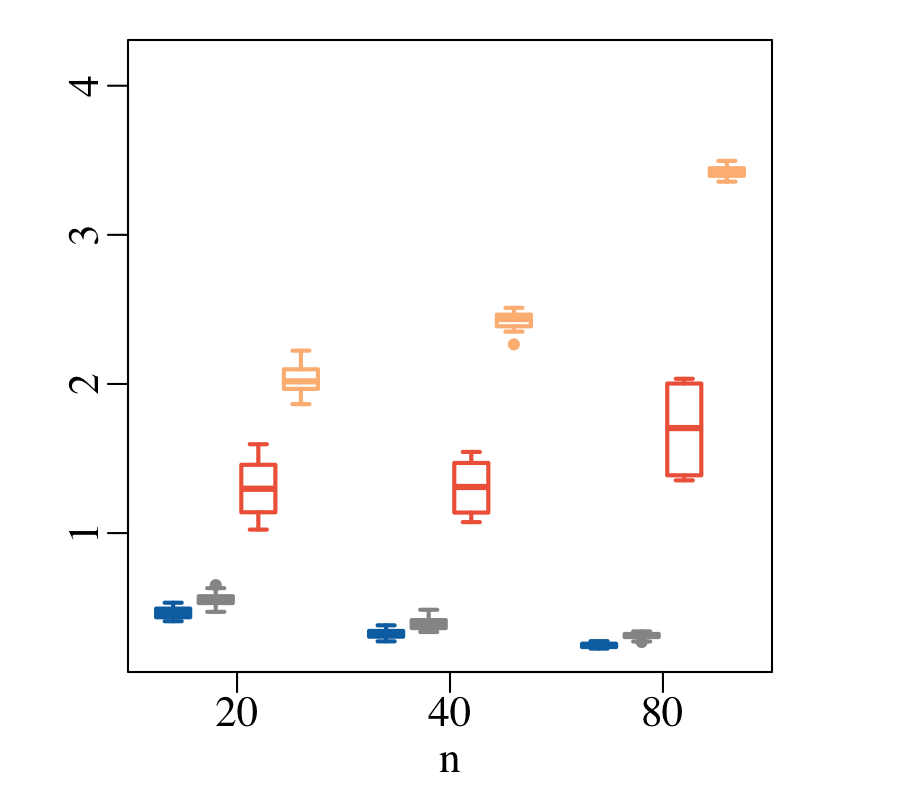} \\
\begin{sideways} \hspace{1.1in} \textbf{LE} \end{sideways}
& \hspace{-.01in}
\includegraphics[width=.35\textwidth]{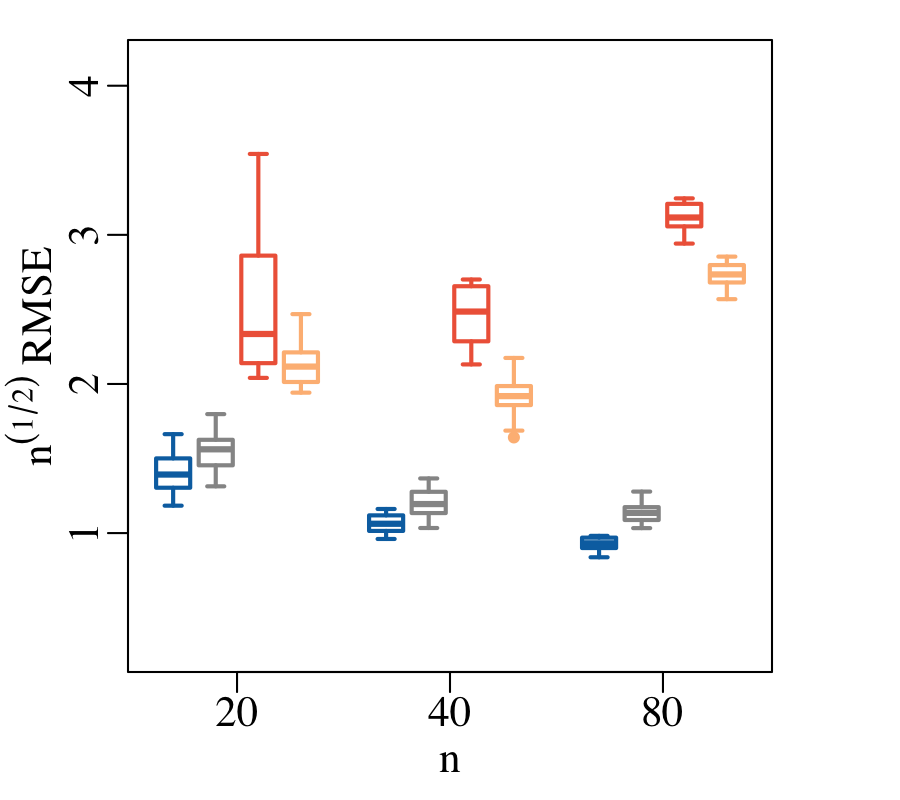}
\hspace{-.1in} & \hspace{-.1in}
\includegraphics[width=.35\textwidth]{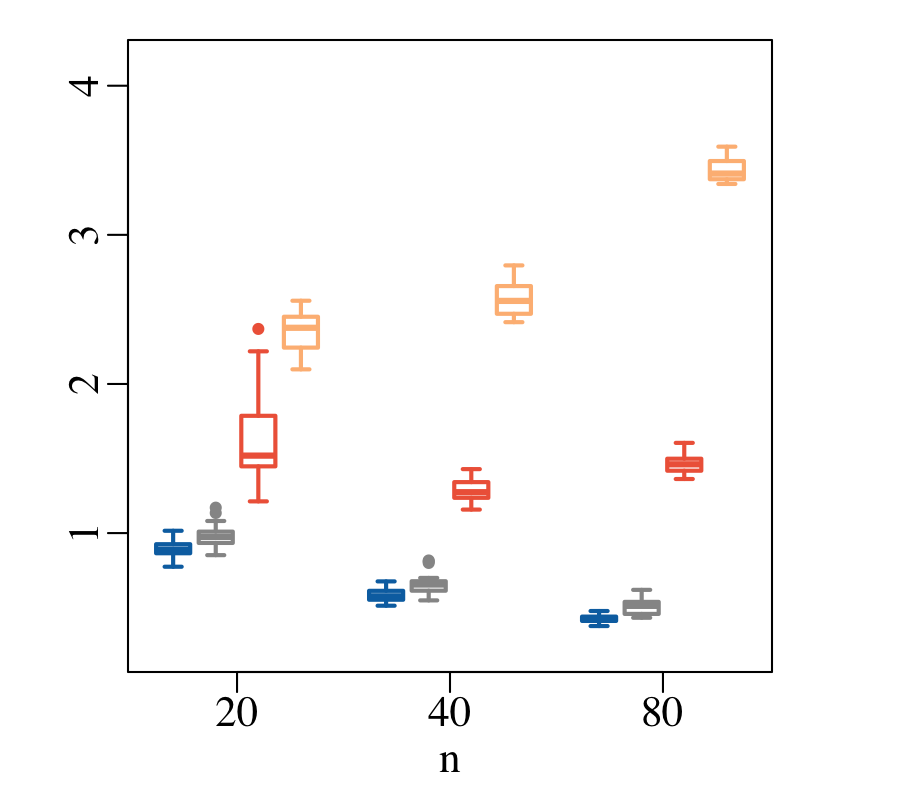}
\hspace{-.1in} & \hspace{-.1in}
\includegraphics[width=.35\textwidth]{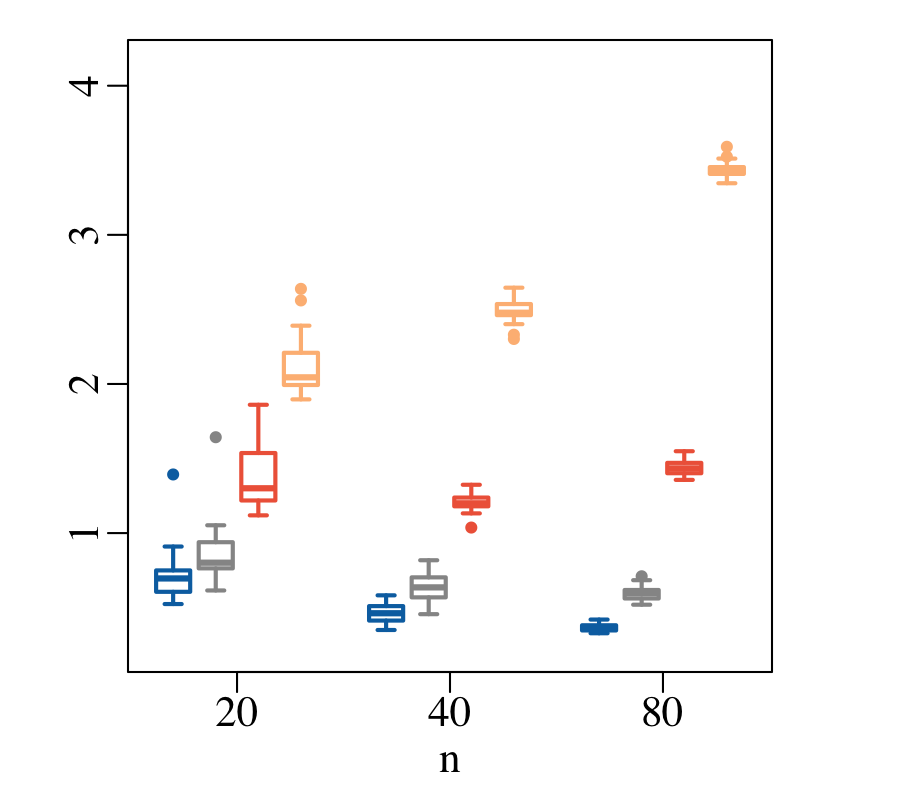}
\hspace{-.1in} & \hspace{-.1in}
\includegraphics[width=.35\textwidth]{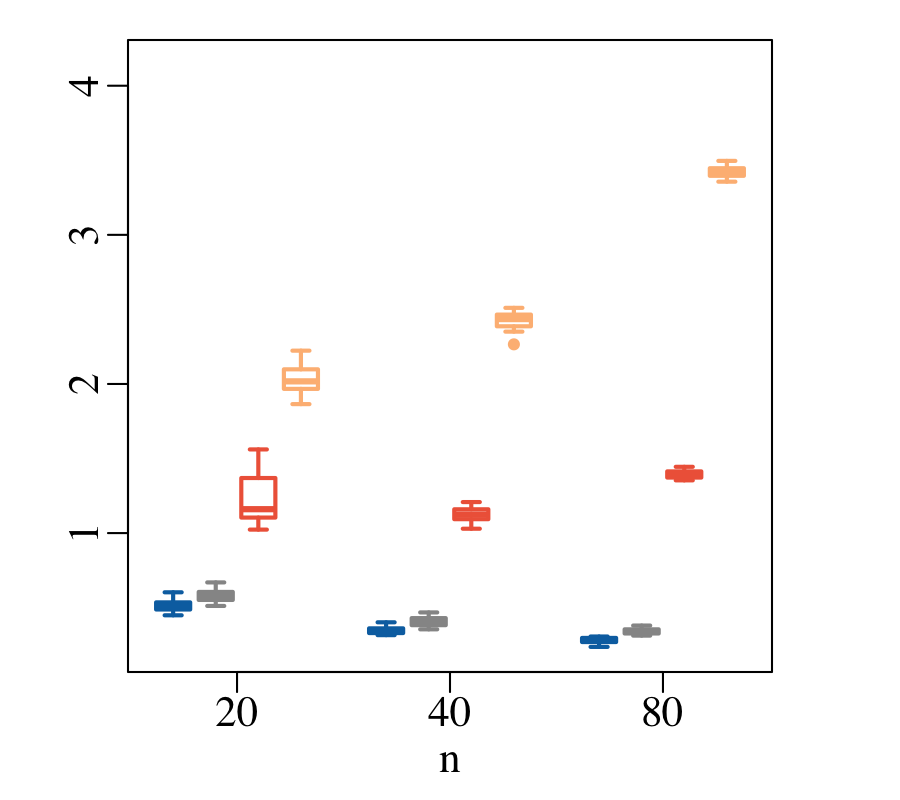} \\
\hspace{-.1in} & \hspace{-.1in}
$\beta_0$
\hspace{-.1in} & \hspace{-.1in}
$\beta_1$
\hspace{-.1in} & \hspace{-.1in}
$\beta_2$ 
\hspace{-.1in} & \hspace{-.1in}
$\beta_3$ 
\end{tabular}

\caption{Performance ($n^{1/2}$ RMSE) of estimators of ${\beta}$, for a given $\X$, when generating from the PX model (top row) and the latent eigenmodel (LE; bottom row). Variability captured by the boxplots reflects variation in RMSE with $\X$. %Note that the intercept, $\beta_0$, has MSEs on different scales than the remaining coefficients. 
} 
    \label{fig:mse_beta_sim}
  \end{figure} 
\end{landscape}

\subsection{Runtimes}
We evaluated the average runtimes of the algorithms used to estimate the simulated data. The average runtimes are plotted in Figure~\ref{fig:sim_runtime}. The improvement in runtime offered by the EMM estimation scheme over SRM and LE MCMC estimation is several orders of magnitude. Interestingly, the runtime cost of EMM appears to grow faster than the MCMC routines, and faster than standard probit regression. A contributing factor is the sum over $O(n^3)$ terms in the maximization of $\rho$ in the EMM algorithm. We have experimented with using only a random subset of $O(n^2)$ relation pairs in the maximization step, which results in gains in runtime with small cost in estimation performance. Such a tradeoff may become attractive for networks of sufficient size $n$.  

\begin{figure}[h]
\centering
\includegraphics[width=.5\textwidth]{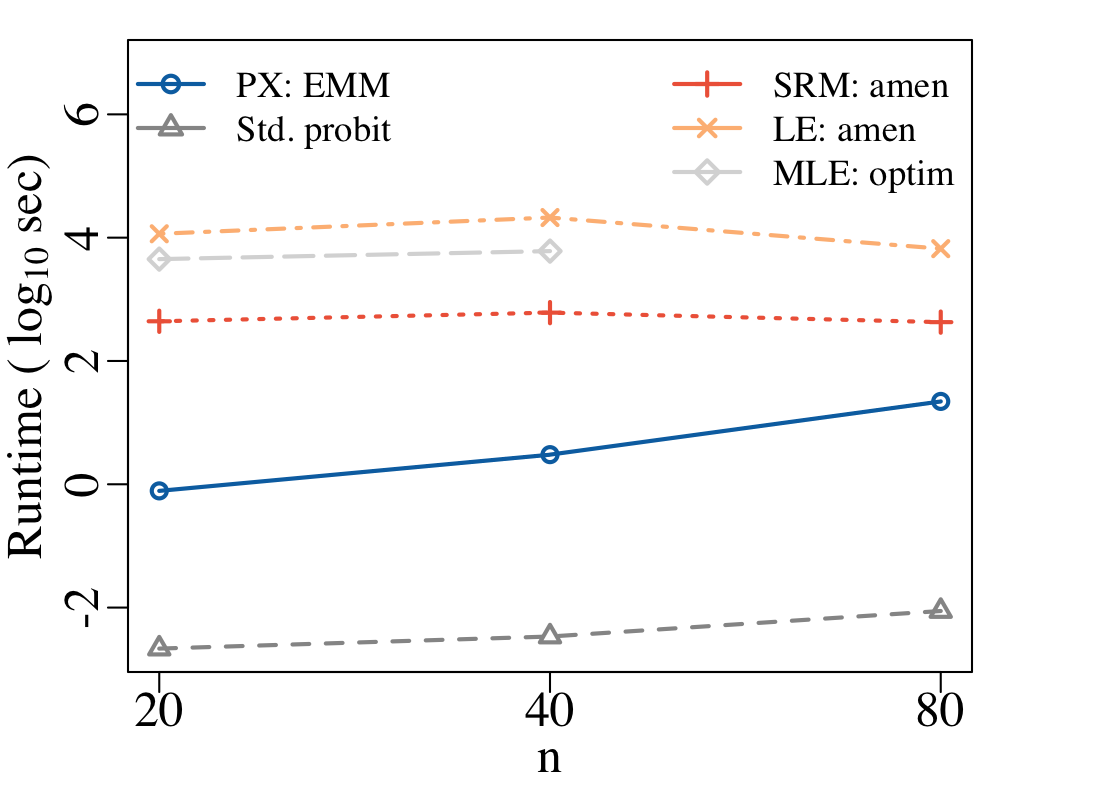}
  \caption{Average runtimes of various algorithms used on simulated data. }  
    \label{fig:sim_runtime}
  \end{figure}

\subsection{Evaluation of latent normality assumption}
To evaluate the performance of the PX model under violationof the normality assumption on the latent errors $\{ \epsilon_{ij} \}_{ij}$, we repeated the simulation study with $t$-distributed latent random variables. Specifically, we simulated from \eqref{eq:sim_gen_model}, but replaced the latent error vector $\bepsilon$ with $\sigma^{-1} \Omega^{1/2} \u$, where $\u$ consists of independently and identically distributed $t$ random variables with 5 degrees of freedom. The scaling factor $\sigma = \sqrt{5/3}$ ensures that $\u$ has unit population variance, for consistency with the Gaussian case, and $\Omega^{1/2}$ is the matrix square root of $\Omega$, with $\rho=0.25$. This model thus has the same latent mean and covariance matrix as in the Gaussian case, but the latent errors have substantially heavier tails.

\begin{figure}[h!]
\centering
\begin{tabular}{cc}
\includegraphics[width=.48\textwidth]{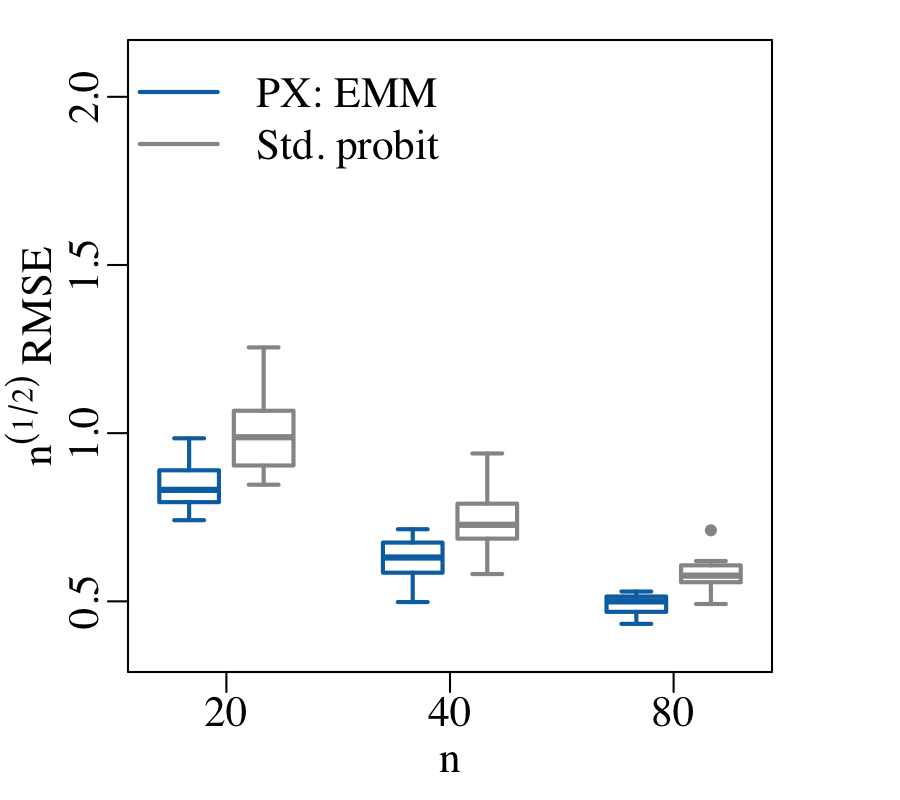}
\hspace{-.1in} & \hspace{-.1in} 
\includegraphics[width=.48\textwidth]{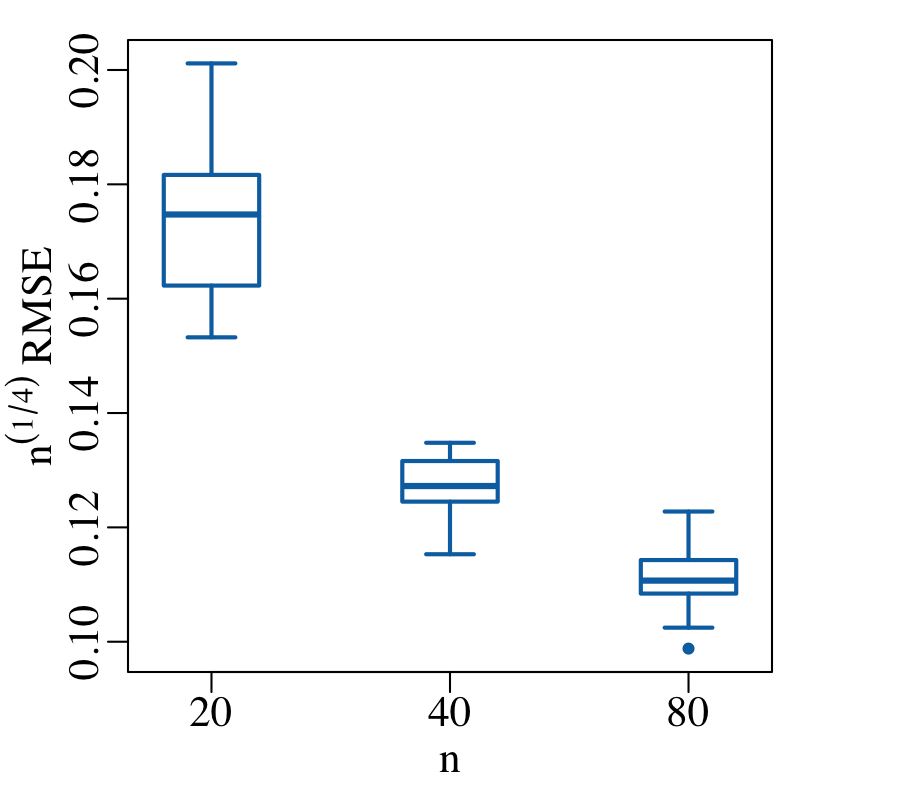} \\
\end{tabular}
  \caption{
  The left panel depicts $n^{1/2}$RMSE in estimating $\beta_1$, using the EMM algorithm and standard probit regression, under $t$ distribution of the errors. The right panel depicts $n^{1/4}$RMSE in estimating $\rho$ using the EMM algorithm in the same simulation. Variability captured by the boxplots reflects variation in RMSE with $\X$.
 } 
    \label{fig:tsim}
  \end{figure}

The left panel of Figure~\ref{fig:tsim} shows the performance of the EMM algorithm in estimating $\bbeta_1$, compared to the standard probit regression estimates. As in the Gaussian case, the EMM algorithm produces estimates with $n^{1/2}$RMSE tending to zero as $n$ grows. Also as in the Gaussian case, EMM estimation of the PX model improves estimation of $\bbeta$ over standard probit regression. We observed the same results in estimation of the remaining coefficients (see Appendix~\ref{sec:t_appx}). 
Unlike the Gaussian case, $n^{1/2}$RMSE in estimating $\rho$ did not appear to tend towards zero. However, in the right panel of Figure~\ref{fig:tsim}, the error in estimating $\rho$ scaled by $n^{1/4}$, $n^{1/4}$ RMSE, does tend towards zero. This study confirms the claim in Section~\ref{sec:consistency} that the EMM algorithm prdouces consistent estimators $\{\hat{\bbeta}, \hat{\rho} \}$, even under violation of the normality assumption of the PX model.

\section{Analysis of a network of political books}
\label{sec:data}

We live in a time of political polarization. We investigate this phenomenon by analyzing a network of $n=105$ books on American politics published around the time of the 2004 presidential election\footnote{These unpublished data were compiled by Dr. Valdis Krebs for his website \url{http://www.orgnet.com/} and are hosted, with permission, by Dr. Mark Newman at \url{http://www-personal.umich.edu/~mejn/netdata/polbooks.zip}}. 
These data were compiled by Dr. Valdis Krebs using the ``customers who bought this book also
bought these books'' list on Amazon.com. At the time, when browsing a particular book, Amazon listed the books that were bought by individuals who also bought the book in question.  Thus, a relation between two books in the network indicates that they were frequently purchased by the same buyer on Amazon. 
Political books on the best-seller list of The New York Times were used as actors in the network.
Finally, the books were labelled as conservative, liberal, or neutral based on each book's description (Figure~\ref{fig:net_pb}). Work by Dr. Krebs on a similar network was described in a 2004 \emph{New York Times} article \citep{eakin2004study}, where it was shown that there were many relations between books with similar ideologies yet relatively few across ideologies. The work by Dr. Krebs has inspired similar analyses of book purchasing networks in the fields of nanotechnology \citep{schummer2005reading} and climate science \citep{shi2017millions}.

To confirm previous work by Dr. Krebs, we develop a model that assigns a different probability of edge formation between books $i$ and $j$ depending on whether the books are ideologically aligned. By examining the network in Figure~\ref{fig:net_pb}, we observe that neutral books appear to have fewer ties than books that are labelled conservative or liberal. Thus, we add a nodal effect indicating whether  either book in a relation is labelled neutral. 
The regression model specified is
\begin{align}
\P(y_{ij}=1) &= \P (\beta_0 + \beta_1 \mathbf{1}[c(i) = c(j)] \nonumber \\
&\hspace{.25in}+ \beta_2\mathbf{1} \left[ \left\{c(i) = \text{neutral} \right\} \, \cup \, \left\{ c(j) = \text{neutral} \right\} \right] +  \bepsilon_{ij} > 0 ),    \quad \label{eq_model_pb} \\
\nonumber \bepsilon &\sim (\mbf{0}, \bSigma),
\end{align}
where $c(i)$ represents the class of book $i$ (neutral, conservative, or liberal) and the distribution and covariance matrix of $\bepsilon$ are determined by the particular model being estimated. In this section, we estimate the PX model (PX), the equivalent social relations model (SRM), the latent eigenmodel (LE), and, as a baseline, 
the standard probit regression model assuming independence of observations (which we label ``std. probit'').

\begin{figure}
\centering
\begin{tabular}{cc}
\includegraphics[width=.44\textwidth]{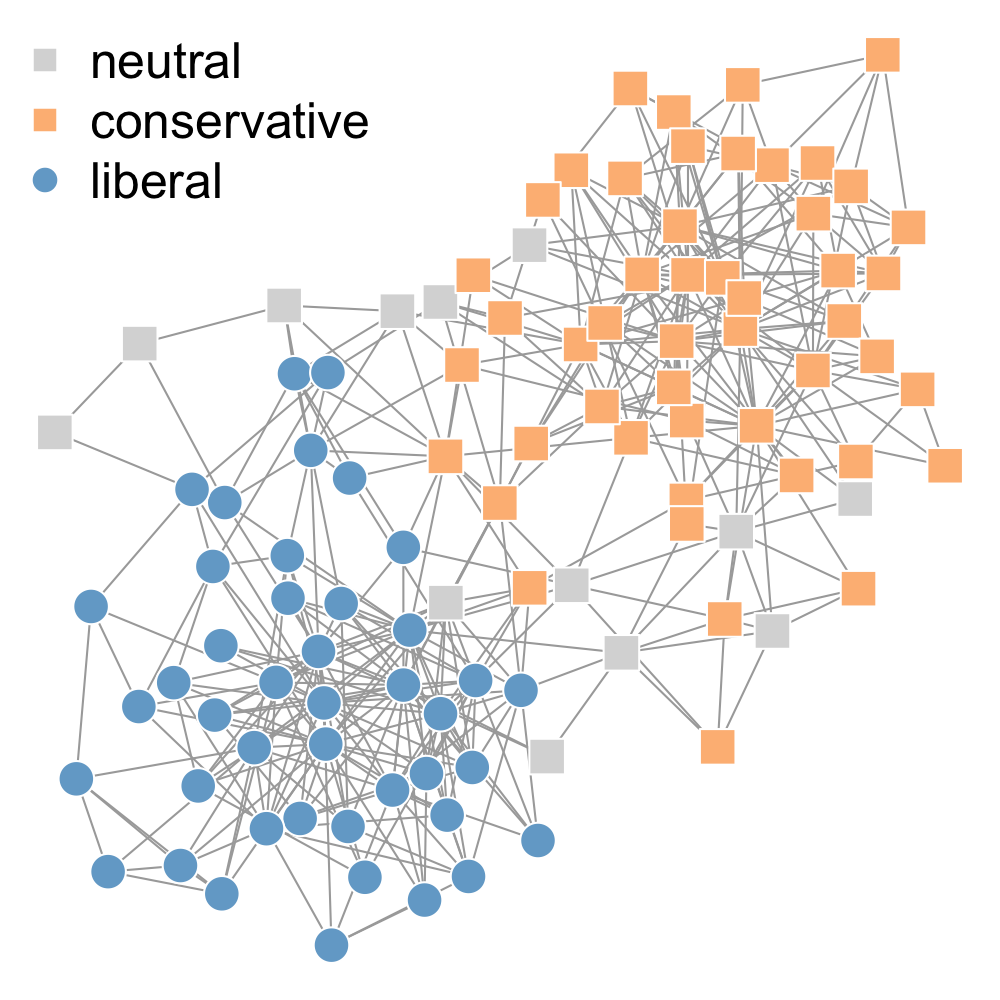}
&
\includegraphics[width=.44\textwidth]{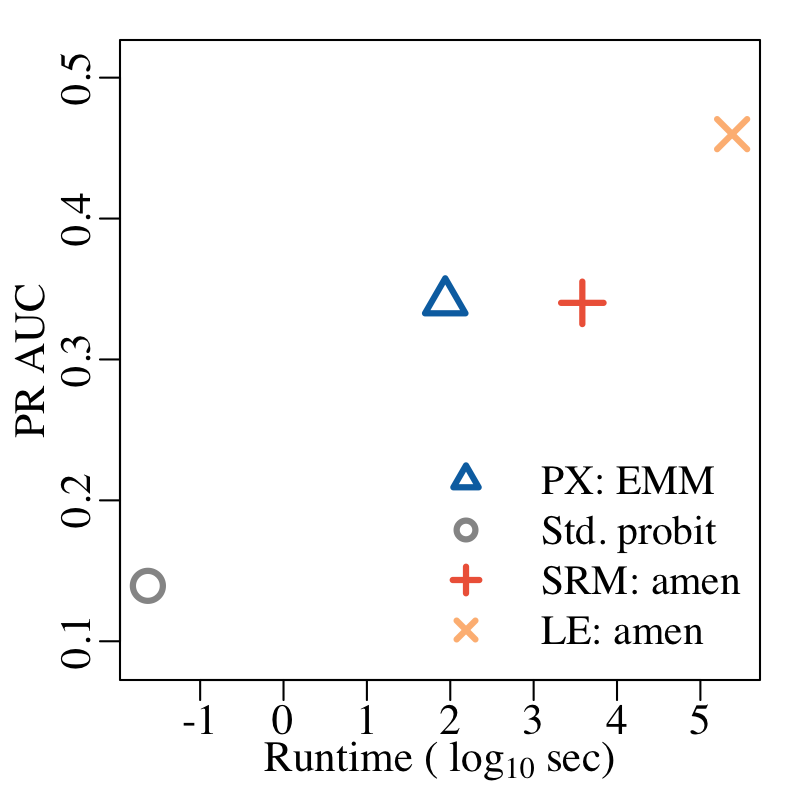}
\end{tabular}
  \caption{Krebs' political books network (left) and out-of-sample performance in 10-fold cross validation, as measured by area under the precision-recall curve (PRAUC, right), plotted against mean runtime in the cross validation. The estimators are standard probit assuming independnet observations (Std. probit), the PX model as estimated by the EMM algorithm (PX), the social relations model estimator (SRM), and the latent eigenmodel estimator (LE). }  
    \label{fig:net_pb}
  \end{figure}

We used a 10-fold cross validation to compare the out-of-sample predictive performance of the estimators and the runtimes of the algorithms for the models in question. We used the proposed EMM algorithm to estimate the PX model, the \texttt{amen} package in \texttt{R} to estimate the social relations model and latent eigenmodel \citep{amen}, and the  \texttt{glm(.)} command in the \texttt{R} package \texttt{stats} to estimate the standard probit model.
We randomly divided the $\binom{105}{2}$ relations into 10 disjoint sets, termed ``folds'',  of roughly the same size. Then, for each fold, we estimated the models on the remaining nine folds and made predictions for the data in the fold that was not used for estimation (for details of estimation of the PX model with missing data, see Appendix~\ref{sec:missing}). Repeating this operation for each fold gave a complete data set of out-of-sample predictions for each estimating model. 
The procedure to make marginal predictions from the PX model is described in Section~\ref{sec:pred_proc}. To compare with the PX model, we make marginal predictions from the social relations model and the latent eigenmodel, that is, by integrating over the random effect space. The predictions from the social relations model and the latent eigenmodel are automatically output from \texttt{amen} in the presence of missing data. The predictions from the standard probit model are marginal by default as there is no correlation structure.

We use area under the precision recall curve (PRAUC) to measure performance of the predictions relative to the observed data, although using area under the receiver operating characteristic (ROC) yields the same conclusions (see Appendix~\ref{sec:app_data}).
In Figure~\ref{fig:net_pb}, the proposed EMM estimator produces an improvement in PRAUC over standard probit prediction that is roughly equivalent to the improvement of the social relations model over standard probit, yet with an average runtime that is 45 times faster (about a minute compared with an hour). 
The latent eigenmodel produces an improvement in PRAUC over the proposed EMM algorithm and the social relations model, however, at the expense of significant increase in average runtime, that of about 3,000 times slower than EMM and taking almost three days to complete. 
Note that we selected the number of MCMC iterations for the social relations and latent eigenmodels that resulted in sets of samples from the posterior distributions (after burn-in) that had a effective sample sizes roughly equal to 100 independent samples of the $\bbeta$ parameters. Increasing the number of iterations, which may be desirable, would result in even longer runtimes for the estimators of the social relations and latent eigenmodels.
Taken together, the results of the cross validation study suggest that the PX model accounts for a large portion of the correlation in network data with estimation runtime that, depending upon stopping criterion, is orders of magnitude faster the runtime than existing approaches.

To estimate the complete data set under the mean model in \eqref{eq_model_pb}, we used the EMM algorithm for the PX model and the \texttt{amen} package for the social relations model (SRM) and latent eigenmodel (LE), which we ran for $1\times10^6$ iterations after a burn in of $5 \times 10^4$ iterations (with runtimes of roughly two hours for SRM and 17 hours for LE). The coefficient estimates in Table~\ref{tab:pb} suggest that books that share the same ideology are more likely to be frequently purchased together, as all $\hat{\beta}_1 > 0$. This positive coefficient estimate demonstrates political polarization in the network: conservative books are more likely to be purchased with other conservative books rather than with liberal books.  
The second coefficient estimate, $\hat{\beta}_2 > 0$, suggests that, relative to a random pair of ideologically misaligned books, pairs of books where at least one of the books is neutral are more likely to be purchased together. 
Neutral books are thus generally more likely to be purchased with books of disparate ideologies, and 
have a unifying effect in the book network. Returning briefly to Table~\ref{tab:pb}, the runtimes highlight that EMM reduces computational burden by order(s) of magnitude over existing approaches. 

\begin{table}[h!]
\centering
\caption{Results of fitting the Krebs political books data using the EMM estimator for the PX model and the \texttt{amen} estimator for the social relations and latent eigenmodels (SRM and LE, respectively). Point estimates for the coefficients are given to the left of the vertical bar, and runtimes (in seconds) and minimum effective sample sizes across the coefficient estimates are given to the right. }
\label{tab:pb}
\begin{tabular}{rcrcrcrc|crcr}
 & & $\hat{\beta}_0$ 
 & \hspace{.1in} &
 $\hat{\beta}_1$ 
 & \hspace{.1in} &
 $\hat{\beta}_2$ 
 & \hspace{.1in} & \hspace{.1in} & 
 runtime (s) 
 & \hspace{.1in} & 
 min$(ESS)$ \\ 
  \hline
  % -1.865  1.206  1.124
PX: EMM &\hspace{.1in} & -1.87 & & 1.21 && 1.12  &&& 68 && -- \\ 
  SRM: amen & & -2.70 & & 0.98  && 1.55 &&& 7984 && 195 \\ 
  LE: amen & & -3.90  & & 1.63 && 2.06 &&& 62565 && 26 \\
\end{tabular}
\end{table}
%%%%

%%%%
\section{Discussion}
\label{sec:disc}

In this paper we present the PX model, a probit regression model for undirected, binary networks. The PX model adds a single parameter -- latent correlation $\rho$ -- to the ordinary probit  regression model that assumes independence of observations.  
Our focus in this paper is estimation of the effects of exogenous covariates on the observed network, $\bbeta$, and prediction of unobserved network relations. Thus, we do not present uncertainty estimators for $\hat{\bbeta}$ or $\hat{\rho}$. However, practitioners estimating the PX model may require uncertainty estimators to perform inference. Development and evaluation of estimators of the uncertainty in estimators of network data  is non-trivial; indeed, entire papers are dedicated to this task for the simpler linear regression case (see, for example, \cite{aronow2015cluster, marrs2017standard}). Future development of uncertainty estimators for the PX model may draw upon existing literature for uncertainty in EM estimators \citep{louis1982finding} and the numerical approximations in this paper. 
%Thus, we leave the development of uncertainty estimators for the PX model to future work. 

A popular notion in the analysis of network data is the presence of higher-order dependencies, meaning beyond second order \citep{hoff2005bilinear}.  The representation of triadic closure, a form of transitivity -- the friend of my friend is likely to also be my friend -- is one motivation for the latent eigenmodel \citep{hoff2008modeling}. The PX model does represent triadic closure to a 
degree. One can show that, given two edges of a triangle relation exist, $y_{ij} = y_{jk} = 1$, the probability that the third edge exists, $\P(y_{ik} = 1)$, increases as $\rho$ increases. 
%Further, this increase is larger than if $y_{ij}$ and  $y_{jk}$ did not share actor $j$. 
However, the increase in probability describing triadic closure under the PX model is fixed based on the estimated value of $\rho$, which is informed only by the first two moments of the data when using the EMM estimator. It may be desirable to develop a test for whether the PX model sufficiently represents the level of triadic closure as suggested by the data. One such test might compute the empirical probability that $\P(y_{ik} = 1 \, | \, y_{ij} = y_{jk} = 1)$  and compare this statistic to its distribution under the null that the PX model is the true model with correlation parameter $\rho = \hat{\rho}$. Future work consists in theoretical development of the distributions of the test statistic(s) of choice under the null. Statistics of interest will likely be related to various clustering coefficients in the networks literature \citep{wasserman1994social, watts1998collective}.

We focus on the probit model in this paper.  
However, we find that this choice may limit the degree of covariance in the observed network $\{ y_{ij} \}_{ij}$ that the PX model can represent. For constant mean $\x_{ij}^T \bbeta = \mu$, the maximum covariance the PX model can represent is bounded by
\begin{align}
cov[y_{ij}, y_{ik}] 
\, \le \,
\lim_{\rho \rightarrow 1/2} \int_{-\mu}^\infty \int_{-\mu}^\infty dF_\rho - \Phi(\mu)^2, \quad \label{eq:px_corr_ub}
\end{align}
where $dF_\rho$ is the bivariate standard normal distribution with correlation $\rho$. The use of different latent distributions for $\bepsilon$ other than normal may allow a model analogous to the PX model to represent a larger range of observed covariances $cov[y_{ij}, y_{ik}]$. 
Future work may consider a logistic distribution for $\bepsilon$, as some researchers prefer to make inference with logistic regression models for binary data due to the ease of interpretation. 
%%%%

\section*{Acknowledgements}
 This work utilized the RMACC Summit supercomputer, which is supported by the National Science Foundation (awards ACI-1532235 and ACI-1532236), the University of Colorado Boulder, and Colorado State University. The RMACC Summit supercomputer is a joint effort of the University of Colorado Boulder and Colorado State University. This work was also partially supported by the National Science Foundation under Grant no. 1856229. \\

% % \clearpage
\bibliographystyle{nws}
\bibliography{binary_bib}

\begin{thebibliography}{}

\bibitem[\protect\citename{Airoldi {\em et~al.}\relax, }2008]{airoldi2008mixed}
Airoldi, Edoardo~M, Blei, David~M, Fienberg, Stephen~E, \& Xing, Eric~P.
  (2008).
\newblock Mixed membership stochastic blockmodels.
\newblock {\em Journal of {M}achine {L}earning {R}esearch}, {\bf 9}(Sep),
  1981--2014.

\bibitem[\protect\citename{Aldous, }1981]{aldous1981representations}
Aldous, David~J. (1981).
\newblock Representations for partially exchangeable arrays of random
  variables.
\newblock {\em Journal of {M}ultivariate {A}nalysis}, {\bf 11}(4), 581--598.

\bibitem[\protect\citename{Aronow {\em et~al.}\relax, }2015]{aronow2015cluster}
Aronow, Peter~M, Samii, Cyrus, \& Assenova, Valentina~A. (2015).
\newblock Cluster--robust variance estimation for dyadic data.
\newblock {\em Political {A}nalysis}, {\bf 23}(4), 564--577.

\bibitem[\protect\citename{Ashford \& Sowden, }1970]{ashford1970multi}
Ashford, JR, \& Sowden, RR. (1970).
\newblock Multi-variate probit analysis.
\newblock {\em Biometrics},  535--546.

\bibitem[\protect\citename{Atkinson, }2008]{atkinson2008introduction}
Atkinson, Kendall~E. (2008).
\newblock {\em An {I}ntroduction to {N}umerical {A}nalysis}.
\newblock John Wiley \& Sons.

\bibitem[\protect\citename{Bates {\em et~al.}\relax, }2015]{lme4}
Bates, Douglas, M{\"a}chler, Martin, Bolker, Ben, \& Walker, Steve. (2015).
\newblock Fitting linear mixed-effects models using {lme4}.
\newblock {\em Journal of {S}tatistical {S}oftware}, {\bf 67}(1), 1--48.

\bibitem[\protect\citename{Betzel {\em et~al.}\relax, }2018]{betzel2018non}
Betzel, Richard~F, Bertolero, Maxwell~A, \& Bassett, Danielle~S. (2018).
\newblock Non-assortative community structure in resting and task-evoked
  functional brain networks.
\newblock {\em bio{R}xiv},  355016.

\bibitem[\protect\citename{Borgs {\em et~al.}\relax, }2014]{borgs2014p}
Borgs, Christian, Chayes, Jennifer~T, Cohn, Henry, \& Zhao, Yufei. (2014).
\newblock An {L}p theory of sparse graph convergence {I}: limits, sparse random
  graph models, and power law distributions.
\newblock {\em ar{X}iv:1401.2906}.

\bibitem[\protect\citename{Breslow \& Clayton, }1993]{breslow1993approximate}
Breslow, Norman~E, \& Clayton, David~G. (1993).
\newblock Approximate inference in generalized linear mixed models.
\newblock {\em Journal of the {A}merican {S}tatistical {A}ssociation}, {\bf
  88}(421), 9--25.

\bibitem[\protect\citename{Caimo \& Friel, }2011]{caimo2011bayesian}
Caimo, Alberto, \& Friel, Nial. (2011).
\newblock Bayesian inference for exponential random graph models.
\newblock {\em Social {N}etworks}, {\bf 33}(1), 41--55.

\bibitem[\protect\citename{Chib \& Greenberg, }1998]{chib1998analysis}
Chib, Siddhartha, \& Greenberg, Edward. (1998).
\newblock Analysis of multivariate probit models.
\newblock {\em Biometrika}, {\bf 85}(2), 347--361.

\bibitem[\protect\citename{Connolly \& Liang, }1988]{connolly1988conditional}
Connolly, Margaret~A, \& Liang, Kung-Yee. (1988).
\newblock Conditional logistic regression models for correlated binary data.
\newblock {\em Biometrika}, {\bf 75}(3), 501--506.

\bibitem[\protect\citename{Dempster {\em et~al.}\relax,
  }1977]{dempster1977maximum}
Dempster, Arthur~P, Laird, Nan~M, \& Rubin, Donald~B. (1977).
\newblock Maximum likelihood from incomplete data via the {EM} algorithm.
\newblock {\em Journal of the {R}oyal {S}tatistical {S}ociety: {S}eries {B}
  ({S}tatistical {M}ethodology)}, {\bf 39}(1), 1--22.

\bibitem[\protect\citename{Dhaene, }1997]{dhaene1997pseudo}
Dhaene, Geert. (1997).
\newblock Pseudo-true values.
\newblock {\em Encompassing: Formulation, properties and testing},  7--30.

\bibitem[\protect\citename{Eakin, }2004]{eakin2004study}
Eakin, Emily. (2004).
\newblock Study finds a nation of polarized readers.
\newblock {\em The {N}ew {Y}ork {T}imes},  9--9.

\bibitem[\protect\citename{Fagiolo {\em et~al.}\relax,
  }2008]{fagiolo2008topological}
Fagiolo, Giorgio, Reyes, Javier, \& Schiavo, Stefano. (2008).
\newblock On the topological properties of the world trade web: {A} weighted
  network analysis.
\newblock {\em Physica {A}}, {\bf 387}(15), 3868--3873.

\bibitem[\protect\citename{Faust \& Wasserman, }1994]{faust1994social}
Faust, Katherine, \& Wasserman, Stanley. (1994).
\newblock {\em Social {N}etwork {A}nalysis: {M}ethods and {A}pplications}.
\newblock  Vol. 249.
\newblock Cambridge: Cambridge University Press.

\bibitem[\protect\citename{Frank \& Strauss, }1986]{frank1986markov}
Frank, Ove, \& Strauss, David. (1986).
\newblock Markov graphs.
\newblock {\em Journal of the {A}merican {S}tatistical {A}ssociation}, {\bf
  81}(395), 832--842.

\bibitem[\protect\citename{Gelman \& Hill, }2006]{gelman2006data}
Gelman, Andrew, \& Hill, Jennifer. (2006).
\newblock {\em Data {A}nalysis {U}sing {R}egression and
  {M}ultilevel/{H}ierarchical {M}odels}.
\newblock Cambridge University Press.

\bibitem[\protect\citename{Greene, }2003]{greene2003econometric}
Greene, William~H. (2003).
\newblock {\em Econometric {A}nalysis}.
\newblock Prentice Hall.

\bibitem[\protect\citename{Han {\em et~al.}\relax, }2016]{han2016using}
Han, G, McCubbins, OP, \& Paulsen, TH. (2016).
\newblock Using social network analysis to measure student collaboration in an
  undergraduate capstone course.
\newblock {\em Nacta {J}ournal}, {\bf 60}(2), 176.

\bibitem[\protect\citename{Handcock {\em et~al.}\relax,
  }2003]{handcock2003assessing}
Handcock, Mark~S, Robins, Garry, Snijders, Tom, Moody, Jim, \& Besag, Julian.
  (2003).
\newblock {\em Assessing degeneracy in statistical models of social networks}.
\newblock Tech. rept. Citeseer.

\bibitem[\protect\citename{Handcock {\em et~al.}\relax,
  }2007]{handcock2007model}
Handcock, Mark~S, Raftery, Adrian~E, \& Tantrum, Jeremy~M. (2007).
\newblock Model-based clustering for social networks.
\newblock {\em Journal of the {R}oyal {S}tatistical {S}ociety: {S}eries {A}
  ({S}tatistics in {S}ociety)}, {\bf 170}(2), 301--354.

\bibitem[\protect\citename{Handcock {\em et~al.}\relax,
  }2019]{handcock:statnet}
Handcock, Mark~S., Hunter, David~R., Butts, Carter~T., Goodreau, Steven~M.,
  Krivitsky, Pavel~N., \& Morris, Martina. (2019).
\newblock {\em ergm: Fit, simulate and diagnose exponential-family models for
  networks}.
\newblock The Statnet Project (\url{https://statnet.org}).
\newblock R package version 3.10.4.

\bibitem[\protect\citename{Heagerty \& Lele, }1998]{heagerty1998composite}
Heagerty, Patrick~J, \& Lele, Subhash~R. (1998).
\newblock A composite likelihood approach to binary spatial data.
\newblock {\em J. am. stat. assoc.}, {\bf 93}(443), 1099--1111.

\bibitem[\protect\citename{Hoff, }2008]{hoff2008modeling}
Hoff, Peter. (2008).
\newblock Modeling homophily and stochastic equivalence in symmetric relational
  data.
\newblock {\em Pages  657--664 of:} {\em Advances {N}eural {I}nformation
  {P}rocessing {S}ystems}.

\bibitem[\protect\citename{Hoff {\em et~al.}\relax, }2017]{amen}
Hoff, Peter, Fosdick, Bailey, Volfovsky, Alex, \& He, Yanjun. (2017).
\newblock {\em amen: Additive and multiplicative effects models for networks
  and relational data}.
\newblock R package version 1.3.

\bibitem[\protect\citename{Hoff, }2005]{hoff2005bilinear}
Hoff, Peter~D. (2005).
\newblock Bilinear mixed-effects models for dyadic data.
\newblock {\em Journal of the {A}merican {S}tatistical {A}ssociation}, {\bf
  100}(469), 286--295.

\bibitem[\protect\citename{Hoff {\em et~al.}\relax, }2002]{hoff2002latent}
Hoff, Peter~D, Raftery, Adrian~E, \& Handcock, Mark~S. (2002).
\newblock Latent space approaches to social network analysis.
\newblock {\em Journal of the {A}merican {S}tatistical {A}ssociation}, {\bf
  97}(460), 1090--1098.

\bibitem[\protect\citename{Holland \& Leinhardt, }1981]{holland1981exponential}
Holland, Paul~W, \& Leinhardt, Samuel. (1981).
\newblock An exponential family of probability distributions for directed
  graphs.
\newblock {\em Journal of the {A}merican {S}tatistical {A}ssociation}, {\bf
  76}(373), 33--50.

\bibitem[\protect\citename{Hoover, }1979]{hoover1979relations}
Hoover, Douglas~N. (1979).
\newblock Relations on probability spaces and arrays of random variables.
\newblock {\em Preprint, {I}nstitute for {A}dvanced {S}tudy, {P}rinceton,
  {NJ}}, {\bf 2}.

\bibitem[\protect\citename{Huber, }1967]{huber1967under}
Huber, Peter~J. (1967).
\newblock The behavior of maximum likelihood estimates under nonstandard
  conditions.
\newblock {\em Page  221 of:} {\em Proceedings of the fifth berkeley symposium
  on mathematical statistics and probability: Weather modification; university
  of california press: Berkeley, ca, usa}.

\bibitem[\protect\citename{Hunter {\em et~al.}\relax, }2008a]{handcock:ergm}
Hunter, David~R., Handcock, Mark~S., Butts, Carter~T., Goodreau, Steven~M., \&
  Morris, Martina. (2008a).
\newblock ergm: A package to fit, simulate and diagnose exponential-family
  models for networks.
\newblock {\em Journal of {S}tatistical {S}oftware}, {\bf 24}(3), 1--29.

\bibitem[\protect\citename{Hunter {\em et~al.}\relax,
  }2008b]{hunter2008goodness}
Hunter, David~R, Goodreau, Steven~M, \& Handcock, Mark~S. (2008b).
\newblock Goodness of fit of social network models.
\newblock {\em Journal of the {A}merican {S}tatistical {A}ssociation}, {\bf
  103}(481), 248--258.

\bibitem[\protect\citename{Kallenberg, }2006]{kallenberg2006probabilistic}
Kallenberg, Olav. (2006).
\newblock {\em Probabilistic {S}ymmetries and {I}nvariance {P}rinciples}.
\newblock Springer Science \& Business Media.

\bibitem[\protect\citename{Le~Cessie \& Van~Houwelingen, }1994]{le1994logistic}
Le~Cessie, Saskia, \& Van~Houwelingen, JC. (1994).
\newblock Logistic regression for correlated binary data.
\newblock {\em Journal of the {R}oyal {S}tatistical {S}ociety: {S}eries {C}
  ({A}pplied {S}tatistics)},  95--108.

\bibitem[\protect\citename{Li \& Schafer, }2008]{li2008likelihood}
Li, Yonghai, \& Schafer, Daniel~W. (2008).
\newblock Likelihood analysis of the multivariate ordinal probit regression
  model for repeated ordinal responses.
\newblock {\em Computional {S}tatistics and {D}ata {A}nalysis}, {\bf 52}(7),
  3474--3492.

\bibitem[\protect\citename{Littell {\em et~al.}\relax, }2006]{littell2006sas}
Littell, Ramon~C, Milliken, George~A, Stroup, Walter~W, Wolfinger, Russell~D,
  \& Oliver, Schabenberber. (2006).
\newblock {\em {SAS} for {M}ixed {M}odels}.
\newblock SAS publishing.

\bibitem[\protect\citename{Louis, }1982]{louis1982finding}
Louis, Thomas~A. (1982).
\newblock Finding the observed information matrix when using the em algorithm.
\newblock {\em Journal of the {R}oyal {S}tatistical {S}ociety: {S}eries {B}
  ({S}tatistical {M}ethodology)}, {\bf 44}(2), 226--233.

\bibitem[\protect\citename{Lov{\'a}sz \& Szegedy, }2006]{lovasz2006limits}
Lov{\'a}sz, L{\'a}szl{\'o}, \& Szegedy, Bal{\'a}zs. (2006).
\newblock Limits of dense graph sequences.
\newblock {\em Journal of {C}ombinatorial. {T}heory, {S}eries {B}}, {\bf
  96}(6), 933--957.

\bibitem[\protect\citename{Ma {\em et~al.}\relax, }2020]{ma2020universal}
Ma, Zhuang, Ma, Zongming, \& Yuan, Hongsong. (2020).
\newblock Universal latent space model fitting for large networks with edge
  covariates.
\newblock {\em The journal of machine learning research}, {\bf 21}(1), 86--152.

\bibitem[\protect\citename{Marrs {\em et~al.}\relax, }2017]{marrs2017standard}
Marrs, Frank~W, Fosdick, Bailey~K, \& McCormick, Tyler~H. (2017).
\newblock Standard errors for regression on relational data with exchangeable
  errors.
\newblock {\em arxiv:1701.05530}.

\bibitem[\protect\citename{Nowicki \& Snijders, }2001]{nowicki2001estimation}
Nowicki, Krzysztof, \& Snijders, Tom A~B. (2001).
\newblock Estimation and prediction for stochastic blockstructures.
\newblock {\em Journal of the {A}merican {S}tatistical {A}ssociation}, {\bf
  96}(455), 1077--1087.

\bibitem[\protect\citename{Ochi \& Prentice, }1984]{ochi1984likelihood}
Ochi, Y, \& Prentice, Ross~L. (1984).
\newblock Likelihood inference in a correlated probit regression model.
\newblock {\em Biometrika}, {\bf 71}(3), 531--543.

\bibitem[\protect\citename{Petersen {\em et~al.}\relax,
  }2008]{petersen2008matrix}
Petersen, Kaare~Brandt, Pedersen, Michael~Syskind, {\em et~al.\ }\relax.
  (2008).
\newblock The matrix cookbook.
\newblock {\em {T}echnical {U}niversity of {D}enmark}, {\bf 7}(15), 510.

\bibitem[\protect\citename{Schummer, }2005]{schummer2005reading}
Schummer, Joachim. (2005).
\newblock Reading nano: The public interest in nanotechnology as reflected in
  purchase patterns of books.
\newblock {\em Public {U}nderstanding of {S}cience}, {\bf 14}(2), 163--183.

\bibitem[\protect\citename{Schweinberger, }2011]{schweinberger2011instability}
Schweinberger, Michael. (2011).
\newblock Instability, sensitivity, and degeneracy of discrete exponential
  families.
\newblock {\em Journal of the {A}merican {S}tatistical {A}ssociation}, {\bf
  106}(496), 1361--1370.

\bibitem[\protect\citename{Sewell \& Chen, }2015]{sewell2015latent}
Sewell, Daniel~K, \& Chen, Yuguo. (2015).
\newblock Latent space models for dynamic networks.
\newblock {\em Journal of the {A}merican {S}tatistical {A}ssociation}, {\bf
  110}(512), 1646--1657.

\bibitem[\protect\citename{Shalizi \& Rinaldo, }2013]{shalizi2013consistency}
Shalizi, Cosma~Rohilla, \& Rinaldo, Alessandro. (2013).
\newblock Consistency under sampling of exponential random graph models.
\newblock {\em Annals of {S}tatistics}, {\bf 41}(2), 508.

\bibitem[\protect\citename{Shi {\em et~al.}\relax, }2017]{shi2017millions}
Shi, Feng, Shi, Yongren, Dokshin, Fedor~A, Evans, James~A, \& Macy, Michael~W.
  (2017).
\newblock Millions of online book co-purchases reveal partisan differences in
  the consumption of science.
\newblock {\em Nature {H}uman {B}ehavior}, {\bf 1}(4), 0079.

\bibitem[\protect\citename{Snijders, }2002]{snijders2002markov}
Snijders, Tom~AB. (2002).
\newblock Markov chain monte carlo estimation of exponential random graph
  models.
\newblock {\em Journal of {S}ocial {S}tructure}, {\bf 3}(2), 1--40.

\bibitem[\protect\citename{Snijders \& Kenny, }1999]{snijders1999social}
Snijders, Tom~AB, \& Kenny, David~A. (1999).
\newblock The social relations model for family data: A multilevel approach.
\newblock {\em Pers. {R}elatsh.}, {\bf 6}(4), 471--486.

\bibitem[\protect\citename{Snijders {\em et~al.}\relax, }2006]{snijders2006new}
Snijders, Tom~AB, Pattison, Philippa~E, Robins, Garry~L, \& Handcock, Mark~S.
  (2006).
\newblock New specifications for exponential random graph models.
\newblock {\em Sociological {M}ethodology}, {\bf 36}(1), 99--153.

\bibitem[\protect\citename{Stiratelli {\em et~al.}\relax,
  }1984]{stiratelli1984random}
Stiratelli, Robert, Laird, Nan, \& Ware, James~H. (1984).
\newblock Random-effects models for serial observations with binary response.
\newblock {\em Biometrics},  961--971.

\bibitem[\protect\citename{Warner {\em et~al.}\relax, }1979]{warner1979new}
Warner, Rebecca~M, Kenny, David~A, \& Stoto, Michael. (1979).
\newblock A new round robin analysis of variance for social interaction data.
\newblock {\em Journal of {P}ersonality and {S}ocial {P}sychology}, {\bf
  37}(10), 1742.

\bibitem[\protect\citename{Wasserman \& Faust, }1994]{wasserman1994social}
Wasserman, Stanley, \& Faust, Katherine. (1994).
\newblock {\em Social {N}etwork {A}nalysis: {M}ethods and {A}pplications}.
\newblock  Vol. 8.
\newblock Cambridge university press.

\bibitem[\protect\citename{Watts \& Strogatz, }1998]{watts1998collective}
Watts, Duncan~J, \& Strogatz, Steven~H. (1998).
\newblock Collective dynamics of ‘small-world’networks.
\newblock {\em Nature}, {\bf 393}(6684), 440.

\bibitem[\protect\citename{Wong, }1982]{wong1982round}
Wong, George~Y. (1982).
\newblock Round robin analysis of variance via maximum likelihood.
\newblock {\em Journal of the {A}merican {S}tatistical {A}ssociation}, {\bf
  77}(380), 714--724.

\bibitem[\protect\citename{Wu, }1983]{wu1983convergence}
Wu, CF~Jeff. (1983).
\newblock On the convergence properties of the em algorithm.
\newblock {\em The annals of statistics},  95--103.

\bibitem[\protect\citename{Zhang \& Horvath, }2005]{zhang2005general}
Zhang, Bin, \& Horvath, Steve. (2005).
\newblock A general framework for weighted gene co-expression network analysis.
\newblock {\em Statistical {A}pplications in {G}enetics and {M}olecular
  {B}iology}, {\bf 4}(1).

\bibitem[\protect\citename{Zhang {\em et~al.}\relax, }2022]{zhang2022joint}
Zhang, Xuefei, Xu, Gongjun, \& Zhu, Ji. (2022).
\newblock Joint latent space models for network data with high-dimensional node
  variables.
\newblock {\em Biometrika}, {\bf 109}(3), 707--720.

\end{thebibliography}

\clearpage
\appendix

\pagenumbering{arabic} 
\setcounter{page}{1}

\section{Details of estimation}
\label{sec:est_details}
In this section we supply details of estimation in support of Algorithm~\ref{alg:sub}, beginning with the initialization of $\rho$. 
We then provide details of computing the expectations of $\ell_\y$ need for $\beta$ maximization, and then details of computing the expectations of $\ell_\y$ need for $\rho$ maximization. We close the section with the handling of missing data in the EMM algorithm.

\subsection{Initialization of $\rho$ estimator}
\label{sec:rho_init}
An EM algorithm may take many iterations to converge, and  selecting a starting point near the optima may significantly reduce the number of iterations required. We present a method of initializing $\hat{\rho}^{(0)}$ using a mixture estimator. By examining the eigenvalues of $\bOmega$, it can be shown that $\rho$ lies in the interval $[0, 1/2)$ when $\bOmega$ is positive definite  for arbitrary $n$ \citep{marrs2017standard}. Thus %, a uniform distribution on the range of possible values of $\rho$ has mean $0.25$, and
$\hat{\rho} = 0.25$ is a natural naive initialization point as it is the midpoint of the range of possible values. However, we also allow the data to influence the initialization point by taking a random subset  $\s{A}$ of $\Theta_2$ of size $2n^2$, % such that $\s{A} \subset \Theta_2$ and $| \s{A} | = 2n^2$, 
and estimating $\rho$ using the data corresponding to relations in $\s{A}$. Then, the final initialization point is defined as a mixture between the naive estimate $\hat{\rho} = 0.25$ and the estimate based on the data. We weight the naive value as if it arose from $100n$ samples, such that the weights are even at $n=50$, and for increasing $n$, the data estimate dominates:
\begin{align}
\hat{\rho}^{(0)} = \frac{100n}{4(100n + |\s{A}|)} + \frac{|\s{A}|}{ ( 100n + |\s{A}|)} \left( \frac{1}{|\s{A}|}\sum_{jk,lm \in \s{A}}  E[ \epsilon_{jk} \epsilon_{lm} \, | \, y_{jk}, y_{lm} ] \right). \quad
\end{align}
We compute the average $\frac{1}{|\s{A}|}\sum_{jk,lm \in \s{A}}  E[ \epsilon_{jk} \epsilon_{lm} \, | \, y_{jk}, y_{lm} ]$
using the linearization approach described in Section~\ref{sec:rho_linear_appx}.

\subsection{Implementation of $\beta$ expectation step}
\label{sec:appx_beta}

Under general correlation structure, computation of the expectation $E[\bepsilon \, | \, \y]$ (step 1 in Algorithm~\ref{alg:sub}, where we drop conditioning on $\rho^{(\nu)}$ and $\bbeta^{(\nu)}$ to lighten notation) for even small networks is prohibitive, since this expectation is an $\binom{n}{2}$-dimensional truncated multivariate normal integral. We exploit the structure of $\bOmega$ to compute $E[\bepsilon \, | \, \y]$ using the law of total expectation and a Newton-Raphson algorithm.

First, we take a single relation $jk$ and use the law of total expectation to write
\begin{align}
E[\epsilon_{jk}\, | \, \y ] &= E[ E[\epsilon_{jk} \, | \, {\bepsilon}_{-jk},  y_{jk}] \, | \, \y ],  \quad   \label{eq_condl_exp_undir_binary}
\end{align}
where ${\bepsilon}_{-jk}$ is the vector of all entries in $\bepsilon$  except relation $jk$. Beginning with the innermost conditional expectation, the distribution of $\epsilon_{jk}$ given ${\bepsilon}_{-jk}$ and $y_{jk}$ is truncated univariate normal, where the untruncated normal random variable has the mean and variance of $\epsilon_{jk}$ given  ${\bepsilon}_{-jk}$. Based on the conditional multivarite normal distribution and the form of the inverse covariance matrix $\bOmega^{-1} = \sum_{i=1}^3 p_i \s{S}_i$, we may write the untruncated distribution directly as
\begin{align}
&\epsilon_{jk} \, | \, {\bepsilon}_{-jk} \sim {\rm N}(\mu_{jk}, \sigma^2_n), \label{eq_condl_distr_undir}\\
\mu_{jk} &= - {\sigma_n^2} \mathbf{1}_{jk}^T \left( p_2 \s{S}_2 + p_3 \s{S}_3 \right) \tilde{\bepsilon}_{-jk}, \nonumber \\
\sigma^2_n &= \frac{1}{p_1 }, \nonumber
\end{align}
where $\mathbf{1}_{jk}$ is the vector of all zeros with a one in the position corresponding to relation $jk$ and, for notational purposes, we define $\tilde{\bepsilon}_{-jk}$ as the vector $\bepsilon$ except with a zero in the location corresponding to relation $jk$. We note that the diagonal of the matrix $p_2 \s{S}_2 + p_3 \s{S}_3$ consists of all zeros so that $\mu_{jk}$ is free of $\epsilon_{jk}$.

We now condition on $y_{jk}$. For general $z \sim {\rm N}(\mu, \sigma^2)$ and $y = \mathbbm{1}{[z > -\eta]}$  we have that 
\begin{align}
E[z \, | \, y ] &=
\mu + \sigma \frac{\phi(\tilde{\eta})}{\Phi(\tilde{\eta}) (1 - \Phi(\tilde{\eta}))} (y - \Phi(\tilde{\eta}) ), \quad 
\end{align}
where $\tilde{\eta} := (\eta + \mu) / \sigma$.
Now, taking $z = ( \epsilon_{jk} \, | \, {\bepsilon}_{-jk} )$, we have that
\begin{align}
E[\epsilon_{jk} \, | \, {\bepsilon}_{-jk}, y_{jk} ] &= \mu_{jk} + \sigma_n 
\left( \frac{\phi(\tilde{\mu}_{jk} ) \left( y_{jk} - \Phi(\tilde{\mu}_{jk}) \right) }{\Phi(\tilde{\mu}_{jk}) (1 - \Phi(\tilde{\mu}_{jk}) )} \right), \quad 
\label{innerEXP}
% \\ &\approx \frac{1}{m} \mathbf{1}^T \bepsilon_{[-jk]}^* + v_{jk}\sqrt{1 - 2\rho},
\end{align}
where $\tilde{\mu}_{jk} : = (\mu_{jk}+ \x^T_{jk} \bbeta) / \sigma_n$. 

We now turn to the outermost conditional expectation in \eqref{eq_condl_exp_undir_binary}. Substituting the expression for $\mu_{jk}$ into \eqref{innerEXP}, we have that
\begin{align}
E[\epsilon_{jk}\, | \, \y ] &= -{\sigma_n^2} \mathbf{1}_{jk}^T \left( p_2 \s{S}_2 + p_3 \s{S}_3 \right)  E[\bepsilon \, | \, \y ]
+ \sigma_n E \left[ 
\frac{\phi(\tilde{\mu}_{jk} ) \left( y_{jk} - \Phi(\tilde{\mu}_{jk}) \right) }{\Phi(\tilde{\mu}_{jk}) (1 - \Phi(\tilde{\mu}_{jk}) )} 
 \, \Big| \, \y  \right]. \quad  \label{eq:complete_w}
\end{align}
This last conditional expectation is difficult to compute in general. 
Thus, in place of $\tilde{\mu}_{lm}$, we substitute its conditional expectation $E[\tilde{\mu}_{lm} \, | \, \y ]$. Letting $w_{lm} := E[\epsilon_{lm}\, | \, \y ]$ and $\w$ be the vector of the expectations $\{ w_{lm} \}_{lm}$, we define the following nonlinear equation for $\w$:
\begin{align}
0 \approx g(\w) := (-\I + \B) \w + \sigma_n \left( 
\frac{\phi(\tilde{\w} ) \left( \y - \Phi(\tilde{\w}) \right) }{\Phi(\tilde{\w}) (1 - \Phi(\tilde{\w}) )} \right), \quad \label{eq:appx_gfunc}
\end{align}
where we define $\B := -{\sigma_n^2} \left( p_2 \s{S}_2 + p_3 \s{S}_3 \right) $, $\tilde{\w} : = ( \B \w + \X \bbeta)  / \sigma_n$, and the functions $\phi(.)$ and $\Phi(.)$ are applied element-wise. The approximation in \eqref{eq:appx_gfunc} refers to the approximation made when replacing $\tilde{\mu}_{jk}$ with its conditional expectation $E[\tilde{\mu}_{jk} | \y ]$. We use a Newton-Raphson algorithm to update $\w$ \citep{atkinson2008introduction}, initializing the algorithm using the expectation when $\rho=0$,
\begin{align}
\w_0 := 
\frac{\phi( \X\bbeta ) \left( \y - \Phi(\X\bbeta) \right) }{\Phi(\X\bbeta) (1 - \Phi(\X\bbeta) )}. \quad 
\end{align}    
The Newton-Raphson algorithm re-estimates $\w$ based on the estimate at iteration $\nu$, $\hat{\w}^{(\nu)}$, until convergence:
\begin{align}
\hat{\w}^{(\nu+1)} = \hat{\w}^{(\nu)} - \left(\frac{\partial}{\partial \w^T} g(\hat{\w}^{(\nu)}) \right)^{-1} g(\hat{\w}^{(\nu)}).  \quad  \label{eq:newton_update}
\end{align}

The inverse in \eqref{eq:newton_update} is of a matrix that is not of the form $\sum_{i=1}^3 a_i \S_i$. To reduce the computational burden of the Netwon method updates, we numerically approximate the inverse in \eqref{eq:newton_update}. First, we define $v(w_{jk}) = \sigma_n \frac{\phi( w_{jk} ) ( y_{jk} - \Phi(w_{jk}) ) }{\Phi(w_{jk}) (1 - \Phi(w_{jk}) )}$, where we define the vector $v(\w) = \{ v(w_{jk}) \}_{jk}$, and write the derivative
\begin{align}
\frac{\partial}{\partial \w^T} g(\w) &= \B - \I + \D\B. \quad \label{eq:g_deriv}
\end{align}
where we define 
%$\D$ as a square matrix with the set
\begin{align}
\D = {\rm diag}
\left \{ \frac{-w_{jk} \phi_{jk}(y_{jk}  - \Phi_{jk}) - \phi_{jk}^2 - \phi_{jk}^2(y_{jk} - \Phi_{jk})(1 - 2\phi_{jk} \Phi_{jk})}{\Phi_{jk}(1 - \Phi_{jk})} \right\}_{jk}. \nonumber
\end{align}
where we let $\phi_{jk} = \phi(w_{jk})$ and $\Phi_{jk} = \Phi(w_{jk})$.
The term $\D \B$ arises from differentiating $v(\w)$ with respect to $\w$. Using the expression in \eqref{eq:g_deriv}, we are then able to write the second term in \eqref{eq:newton_update} as
\begin{align}
 \left(\frac{\partial}{\partial \w^T} g(\hat{\w}) \right)^{-1} g(\hat{\w}) &= \left( \B - \I + \D\B \right)^{-1}\left( (\B - \I)\w + v(\w) \right), \nonumber \\
 &=  \B^{-1} \left( \I + \D - \B^{-1} \right)^{-1} \left( (\B - \I)\w + v(\w) \right). \quad \label{eq:g_partial_newton}
\end{align}
We notice that the matrix $\I + \D$ is diagonal, but not homogeneous (in which case we compute \eqref{eq:g_partial_newton} directly, with limited computational burden, by exploiting the exchangeable structure). Instead, defining $\Q = (1 + \delta)\I - \B^{-1}$ and $\M = \D - \delta \I $, which is diagonal, we make the approximation that
\begin{align}
\left( \I + \D - \B^{-1} \right)^{-1} &= \left(\Q + \M \right)^{-1} \approx \Q^{-1} - \Q^{-1} \M \Q^{-1}, \nonumber
\end{align}
which is based on a Neumann series of matrices and relies on the absolute eigenvalues of $\M$ being small \citep{petersen2008matrix}.  We choose $\delta$ to be the mean of the minimum and maximum value of $\D$. This choice of $\delta$ minimizes the maximum absolute eigenvalue of $\M$, and thus limits the approximation error. Since the inverse of $\Q$ may be computed using the exchangeable inversion formula discussed in Appendix~\ref{sec:undir_cov_mat} (in $O(1)$ time), the following approximation represents an improvement in computation from $O(n^3)$ to $O(n^2)$ time:
\begin{align}
 \left(\frac{\partial}{\partial \w^T} g(\hat{\w}) \right)^{-1} g(\hat{\w}) &\approx  \B^{-1} \left(\Q^{-1} - \Q^{-1} \M \Q^{-1} \right) \left( (\B - \I)\w + v(\w) \right). \nonumber
\end{align}

\subsection{Approximation to $\rho$ expectation step}
\label{sec:rho_linear_appx}
The maximization of the expected likelihood with respect to  $\rho$ relies on the computation of $\gamma_i = E[\bepsilon^T \S_{i} \bepsilon \, | \, \y ] / |\Theta_i|$, for $i \in \{ 1,2,3\}$ (step 2 in Algorithm~\ref{alg:sub}).
Under general correlation structure, computation of the expectation $\{ \gamma_i \}_{i=1}^3$ for even small networks is prohibitive. To practically compute $\{ \gamma_{i} \}_{i=1}^3$, we make two approximations, which we detail in the following subsections: (1) compute expectations conditioning only on the entries in $\y$ that correspond to the entries in $\bepsilon$ being integrated, and (2) approximating these pairwise expectations as linear functions of $\rho$.

\subsubsection{Pairwise expectation}
\label{sec:pairwise_expectation}
Explicitly, the pairwise approximations to $\{ \gamma_{i} \}_{i=1}^3$ we make are:
\begin{align}
\gamma_1 = \frac{1}{| \Theta_1 |}\sum_{jk} E[ \epsilon_{jk}^2 \, | \, \y ] & \approx \frac{1}{| \Theta_1 |}\sum_{jk} E[ \epsilon_{jk}^2 \, | \, y_{jk} ], \quad \label{eq:gamma_appx} \\
\gamma_2 = \frac{1}{|\Theta_2|}\sum_{jk, lm \in \Theta_2} E[ \epsilon_{jk}  \epsilon_{lm} \, | \, \y ] & \approx  \frac{1}{|\Theta_2|}\sum_{jk, lm \in \Theta_2} E[ \epsilon_{jk}  \epsilon_{lm} \, | \, y_{jk}, y_{lm} ] , \nonumber \\
\gamma_3 = \frac{1}{|\Theta_3|}\sum_{jk, lm \in \Theta_3} E[ \epsilon_{jk}  \epsilon_{lm} \, | \, \y ] & \approx  \frac{1}{|\Theta_3|}\sum_{jk, lm \in \Theta_3} E[ \epsilon_{jk}  \epsilon_{lm} \, | \, y_{jk}, y_{lm} ] , \nonumber 
\end{align}
where $\Theta_i$ is the set of ordered pairs of relations $(jk, lm)$ which correspond entries in $\S_{i}$ that are 1, for $i \in \{1,2,3 \}$. These approximations are natural first-order approximations: recalling that $y_{jk} = \mathbbm{1}[\epsilon_{jk} > -\x^T_{jk} \bbeta]$, the approximations in \eqref{eq:gamma_appx} are based on the notion that knowing the domains of $\epsilon_{jk}$ and $\epsilon_{lm}$ is significantly more informative for $E[ \epsilon_{jk}\epsilon_{lm} \, | \, \y ] $ than knowing the domain of, for example, $\epsilon_{ab}$. 

The approximations in \eqref{eq:gamma_appx} are orders of magnitude faster to compute than the expectations when conditioning on all observations $E[ \epsilon_{jk}\epsilon_{lm} \, | \, \y ]$. In particular, when $i \in \{ 1, 3\}$, the expectations are available in closed form:
% { \small
\begin{align}
E[ \epsilon_{jk}^2 \, | \, y_{jk} ] &= 1 - \eta_{jk} \frac{\phi(\eta_{jk}) (y_{jk} - \Phi(\eta_{jk}))}{ \Phi(\eta_{jk}) (1 -  \Phi(\eta_{jk}))} , \nonumber \\
E[ \epsilon_{jk} \epsilon_{lm} \, | \, y_{jk}, y_{lm} ] &= 
\frac{\phi(\eta_{jk}) \phi(\eta_{lm}) (y_{jk} - \Phi(\eta_{jk})) (y_{lm} - \Phi(\eta_{lm}))}{ \Phi(\eta_{jk}) \Phi(\eta_{lm}) (1 -  \Phi(\eta_{jk})) (1 -  \Phi(\eta_{lm}))}, \nonumber 
 \end{align}
%  }
where we define $\eta_{jk} = \x^T_{jk} \beta$ and the indices $j,k,l$ and $m$ are distinct. 
When $i=2$,  that is, $ | \{ j, k \} \cap \{l, m \} | = 1$, the expectation depends on a two dimensional normal probability integral:
{ \small
\begin{align}
&E[ \epsilon_{jk} \epsilon_{lm} \, | \, y_{jk}, y_{lm} ] = \nonumber \\
& \rho \left(1 
- \frac{\bar{\eta}_{jk} \phi(\eta_{jk}) }{L_{jk,lm}}  \Phi \left( \frac{\bar{\eta}_{lm} - \bar{\rho} \,  \bar{\eta}_{jk}}{\sqrt{1- \rho^2}} \right) 
- \frac{\bar{\eta}_{lm} \phi(\eta_{lm})  }{L_{jk,lm}} \Phi \left( \frac{\bar{\eta}_{jk} - \bar{\rho} \, \bar{\eta}_{lm}}{\sqrt{1- \rho^2}} \right) 
\right) \quad \label{eq:gamma2_exp_appx} \\ 
& + \frac{1}{L_{jk,lm}}  \sqrt{\frac{1 - \rho^2}{2 \pi}} \phi \left( \sqrt{\frac{\eta_{jk}^2 + \eta_{lm}^2 - 2 \rho \, \eta_{jk} \eta_{lm}}{1 - \rho^2}}\right) 
, \hspace{.05in} | \{ j, k \} \cap \{l, m \} | = 1, \nonumber \\
L_{jk,lm} &= \P \left((2 y_{jk} - 1)\epsilon_{jk} > - \eta_{jk}  \cap (2 y_{lm} - 1)\epsilon_{lm} > - \eta_{lm} \right), \nonumber
\end{align}
}
where $\bar{\eta}_{jk} = (2y_{jk} - 1)\eta_{jk}$, e.g., and $\bar{\rho} = (2y_{jk} - 1)(2y_{lm} - 1)\rho$.

\subsubsection{Linearization}
The computation of $E[ \epsilon_{jk} \epsilon_{lm} \, | \, y_{jk}, y_{lm} ] $  in \eqref{eq:gamma2_exp_appx} requires the computation of $O(n^3)$ bivariate truncated normal integrals $L_{jk,lm}$, which are not generally available in closed form. We observe empirically, however, that the pairwise approximation to $\gamma_2$ described in Section~\ref{sec:pairwise_expectation} above, $\gamma_2 \approx \frac{1}{|\Theta_2|}\sum_{jk, lm \in \Theta_2} E[ \epsilon_{jk}  \epsilon_{lm} \, | \, y_{jk}, y_{lm} ]$, is approximately linear in $\rho$. This linearity is somewhat intuitive, as the sample mean $ \frac{1}{|\Theta_2|}\sum_{jk, lm \in \Theta_2} E[ \epsilon_{jk}  \epsilon_{lm} \, | \, y_{jk}, y_{lm} ]$ 
has expectation equal to $\rho$, and is thus an asymptotically linear function of $\rho$. As  the sample mean $ \frac{1}{|\Theta_2|}\sum_{jk, lm \in \Theta_2} E[ \epsilon_{jk}  \epsilon_{lm} \, | \, y_{jk}, y_{lm} ]$ concentrates around its expectation, it concentrates around a linear function of $\rho$,  and it is reasonable to approximate the sample mean $ \frac{1}{|\Theta_2|}\sum_{jk, lm \in \Theta_2} E[ \epsilon_{jk}  \epsilon_{lm} \, | \, y_{jk}, y_{lm} ]$ as a linear function of $\rho$. 
To do so, we compute the approximate values of $\gamma_2$ at $\rho=0$ and if $\rho=1$. In particular,
{
\begin{align}
\gamma_2 &\approx a_2  + b_2 \rho, \label{eq:gamma2_lin} \\
a_2 &= \frac{1}{|\Theta_2|}\sum_{jk, lm \in \Theta_2} E[ \epsilon_{jk} \, | \, y_{jk} ] E[ \epsilon_{lm} \, | \,  y_{lm} ], \nonumber \\ 
&= \frac{1}{|\Theta_2|}\sum_{jk, lm \in \Theta_2}  \frac{\phi(\eta_{jk}) \phi(\eta_{lm}) (y_{jk} - \Phi(\eta_{jk})) (y_{lm} - \Phi(\eta_{lm}))}{ \Phi(\eta_{jk}) \Phi(\eta_{lm}) (1 -  \Phi(\eta_{jk})) (1 -  \Phi(\eta_{lm}))}, \nonumber \\
c_2 &=   \frac{1}{|\Theta_2|}\sum_{jk, lm \in \Theta_2} E[ \epsilon_{jk} \epsilon_{lm} \, | \, y_{jk}, y_{lm} ] \, \Big|_{\rho = 1}, \nonumber \\
b_2 &= c_2 - a_2. \nonumber
\end{align}
}
To compute $c_2$, we must compute the value of $E[ \epsilon_{jk} \epsilon_{lm} \, | \, y_{jk}, y_{lm} ]$ when $\rho = 1$. Computing $E[ \epsilon_{jk} \epsilon_{lm} \, | \, y_{jk}, y_{lm} ]$ is simple when the values $y_{jk} = y_{lm}$, as in this case $E[ \epsilon_{jk} \epsilon_{lm} \, | \, y_{jk}, y_{lm} ] = E[ \epsilon_{jk}^2 \, | \, y_{jk} = y_{lm} ]$ since, when $\rho = 1$, $\epsilon_{jk} = \epsilon_{lm}$.  Approximations must be made in the cases when  $y_{jk} \neq y_{lm}$. There are two such cases. In the first, there is overlap between the domains of $\epsilon_{jk}$ and $\epsilon_{lm}$ indicated by $y_{jk} = \mathbbm{1}[\epsilon_{jk} > -\eta_{jk} ]$ and $y_{jk} = \mathbbm{1}[\epsilon_{lm} > -\eta_{lm}]$, respectively. We define the domain for $\epsilon_{jk}$ indicated by $y_{jk}$ as $U_{jk} := \{ u \in \R : u > (1 - 2y_{jk})\eta_{jk} \}$. As an example, there is overlap between $U_{jk}$ and $U_{lm}$ when $y_{jk} = 1, y_{lm} = 0$ and $\eta_{lm} < \eta_{jk} $. Then, the dersired expectation may be approximated $E[ \epsilon_{jk} \epsilon_{lm} \, | \, y_{jk}, y_{lm} ] \approx E[ \epsilon_{jk}^2 \, | \, \epsilon_{jk} \in U_{jk} \cap U_{lm} ]$. In the second case, when $y_{jk} \neq y_{lm}$ and $U_{jk} \cap U_{lm} = \varnothing$, we make the approximation by integrating over the sets $U_{jk}$ and $U_{lm}$. That is, by taking 
\begin{align}
&E[ \epsilon_{jk} \epsilon_{lm} \, | \, y_{jk}, y_{lm} ] \nonumber \\
& \hspace{.25in}\approx E[ \epsilon_{jk}^2 \, | \, \epsilon_{jk} \in U_{jk} ] \; \P(\epsilon_{jk} \in U_{jk}) + E[ \epsilon_{lm}^2 \, | \, \epsilon_{lm} \in U_{lm} ] \; \P(\epsilon_{lm}  \in U_{lm}). \nonumber
\end{align}
To summarize, we compute $c_2$ in \eqref{eq:gamma2_lin} when $\rho=1$ by using the following approximation to $E[\epsilon_{jk} \epsilon_{lm} \, | \, \y ] \, \Big|_{\rho=1}$ :
{\small
\begin{align}
 \begin{cases}
    E[\epsilon_{jk}^2 \, | \,  \epsilon_{jk} > {\rm max}(-\eta_{jk}, -\eta_{lm}) ], & y_{jk} = 1 \text{ and }y_{lm} = 1,\\
    E[\epsilon_{jk}^2 \, | \,  \epsilon_{jk} < {\rm min}(-\eta_{jk}, -\eta_{lm}) ], & y_{jk} = 0 \text{ and }y_{lm} = 0, \\
   E[ \epsilon_{jk}^2 \, | \, \epsilon_{jk} \in U_{jk} \cap U_{lm} ], & U_{jk} \cap U_{lm} \neq \emptyset, \\
    E[ \epsilon_{jk}^2 \, | \, \epsilon_{jk} \in U_{jk} ] \; \P(\epsilon_{jk} \in U_{jk}) + E[ \epsilon_{lm}^2 \, | \, \epsilon_{lm} \in U_{lm} ] \; \P(\epsilon_{lm}  \in U_{lm})
    & U_{jk} \cap U_{lm} = \emptyset.
  \end{cases} 
  \nonumber
\end{align}    
}

\subsection{Missing data}
\label{sec:missing}
In this subsection, we describe estimation of the PX model in the presence of missing data. We present the maximization of $\ell_\y$ with respect to $\bbeta$ first. Second, we discuss maximization of $\ell_\y$ with respect to $\rho$. Finally, we give a note on prediction from the PX model when data are missing. \\

\noindent\textbf{Update $\bbeta$:} \\
To maximize $\ell_\y$ with respect to  $\bbeta$ (Step 1 of Algorithm~\ref{alg:sub}) in the presence of missing data, we impute the missing values of $\X$ and $\y$. We make the decision to impute missing values since much of the speed of estimation of the PX model relies on exploitation of the particular network structure, and, when data are missing, this structure is more difficult to leverage. We impute entries in $\X$ with the mean value of the covariates. For example, if $x_{jk}^{(1)}$ is missing, we replace it with the sample mean $\frac{1}{|\s{M}^c|}\sum_{lm \in \s{M}^c} x_{lm}^{(1)}$, where the superscript ${(1)}$ refers to the first entry in $\x_{jk}$ and $\s{M}$ is the set of relations for which data are missing. 
%We do the same for the remaining entries in all missing rows of $\X$, $\{ x_{jk}^{(i)} \}_{ jk \in \s{M}}$. 
If $y_{jk}$ is missing, we impute $y_{jk}$ with $\mathbbm{1}[w_{jk} > -\bar{\eta}]$, where $\bar{\eta} = \frac{1}{|\s{M}^c|}\sum_{lm \in \s{M}^c} \x_{lm}^T \hat{\bbeta}$ and we compute $\w = E[\bepsilon \, | \, \y]$ using the procedure in Section~\ref{sec:appx_beta}. We initialize this procedure at $\w^{(0)}$, where any missing entries $jk \in \s{M}$ are initialized with $w^{(0)}_{jk} = 0$. 
Given the imputed $\X$ and $\y$, the estimation routine may be accomplished as described in Algorithm~\ref{alg:sub}. \\

\noindent\textbf{Update $\rho$:} \\
To maximize $\ell_\y$ with respect to $\rho$ (Step 2 of Algorithm~\ref{alg:sub}), we approximate $\{ \gamma_i \}_{i=1}^3$ using only observed values. 
Using the pairwise expressions in \eqref{eq:gamma_appx}, the expressions for the expectation step under missing data are
{\small
\begin{align}
\gamma_1 &\approx  \frac{1}{|\s{M}^c|}
\sum_{jk\in \s{M}^c} E[  \epsilon_{jk}^2 \, | \,  y_{jk} ], \label{eq:gamma_appx_missing} \\
% \nonumber \\
\gamma_2 & \approx  \frac{1}{|\s{A}^{(s)}|}  \sum_{jk, lm \in \s{A}^{(s)}} E[ \epsilon_{jk}  \epsilon_{lm} \, | \, y_{jk}, y_{lm} ].  \nonumber \\
\gamma_3 & \approx \frac{\sum_{jk, lm \in \Theta_3}  E[ \epsilon_{jk}  \, | \, y_{jk} ]  E[ \epsilon_{lm} \, | \,  y_{lm} ] \mathbbm{1}[jk \in \s{M}^c]\mathbbm{1}[lm \in \s{M}^c]}
{\sum_{jk, lm \in \Theta_3} \mathbbm{1}[jk \in \s{M}^c]\mathbbm{1}[lm \in \s{M}^c]}, \nonumber \\  
& \approx \frac{1}{|\Theta_3|} \left(  \Bigg( \frac{|\Theta_1|}{|\s{M}^c|} \sum_{jk\in \s{M}^c} E[  \epsilon_{jk} \, | \,  y_{jk} ]\right)^2 
 - \frac{|\Theta_1|}{|\s{M}^c |}  \sum_{jk\in \s{M}^c} E[  \epsilon_{jk} \, | \,  y_{jk} ]^2  \nonumber \\
& \hspace{.5in} - \frac{|\Theta_2|}{|\s{A}^{(s)}|}  \sum_{jk, lm \in \s{A}^{(s)}} E[  \epsilon_{jk} \, | \,  y_{jk} ] E[  \epsilon_{lm} \, | \,  y_{lm} ] \Bigg), \nonumber
\end{align}
}
where we only subsample pairs of relations that are observed such that $\s{A}^{(s)} \subset \Theta_2 \cap \s{M}^c$. 
Then, given the values of  $\{ \gamma_i \}_{i=1}^3$  in \eqref{eq:gamma_appx_missing}, the maximization of $\ell_\y$ with respect to $\rho$ (Step 2 in Algorithm~\ref{alg:sub}) may proceed as usual. \\

\noindent\textbf{Prediction:} \\
Joint prediction in the presence of missing data 
is required for out-of-sample evaluation of the EMM estimator, for example, for cross validation studies in Section~\ref{sec:data}. 
In this setting, model estimation is accomplished by imputing values in $\X$ and $\y$ earlier in this section under the `\textbf{Update $\bbeta$}' subheading. Then, prediction may be performed by proceeding as described in Section~\ref{sec:pred_proc} with the full observed $\X$ matrix and imputing the missing values in $\y$ (again as described above in this section under the `\textbf{Update $\bbeta$}' subheading).

%%%%
\section{Parameters of undirected exchangeable network covariance   {matrices}}
\label{sec:undir_cov_mat}
In this section, we give a $3\times3$ matrix equation to invert $\bOmega$ rapidly. This equation also gives a basis to compute the partial derivatives $\left\{ \frac{\partial \phi_i}{ \partial p_j} \right\}$, which we require for the EMM algorithm. 

We define an \emph{undirected exchangeable network covariance matrix} as those square, positive definite matrices of the form
\begin{align}
    \bOmega(\bphi) = \sum_{i=1}^3 \phi_i \s{S}_i.  \nonumber
\end{align}
We find empirically that the inverse matrix of any undirected exchangeable network covariance matrix has the same form, that is $\bOmega^{-1} = \sum_{i=1}^3 \p_i \s{S}_i$. Using this fact and the particular forms of the binary matrices $\{\s{S}_i \}_{i=1}^3$, one can see that there are only three possible row-column inner  products in the matrix multiplication $\bOmega \bOmega^{-1}$, those pertaining to row-column pairs of the form $(ij,ij)$, $(ij,ik)$, and $(ij,kl)$ for distinct indices $i,j,k,$ and $l$. Examining the three products in terms of the parameters in $\bphi$ and $\bp$, and the fact that $\bOmega \bOmega^{-1} = \I$, we get the following matrix equation for the parameters $\bp$ given $\bphi$
\begin{align}
&\C(\bphi) \bp = [1,0,0]^T, \quad 
\label{eq:C_undir}
\end{align}
where the matrix $\C(\bphi)$ is given by
{\small
\begin{align}
\left[ {\begin{array}{ccc}
 \phi_1& 2(n-2) \phi_2 & \frac{1}{2}(n-2)(n-3) \phi_3 \\
\phi_2 &  \phi_1 + (n-2) \phi_2  + (n - 3) \phi_3 & (n - 3) \phi_2  +  \left( \frac{1}{2}(n-2)(n-3)-n +3 \right) \phi_3  \\
 \phi_3 & 4 \phi_2 + (2n - 8)\phi_3 & \phi_1 + (2n - 8)\phi_2 + \left( \frac{1}{2}(n-2)(n-3)-2n + 7 \right) \phi_3 \\
  \end{array} } \right]. \nonumber 
\end{align}
}
% We observe empirically that the eigenvalues of $\C(\bphi)$ are contained within those of $\bOmega$, and 
Then, we may  invert $\bOmega$ with a $3 \times 3$ inverse to find the parameters $\bp$ of $\bOmega^{-1}$. Explicitly solving these linear equations, the expressions for $\bp$ are given by
    \begin{align}
        p_1 &= 1 - (2n-4)p_2, \label{eq:p_vals_exact} \\
        p_2 &= \frac{1 + (n-3)p_3}{(2n-4)\rho - n + 2 - 1/\rho}, \nonumber \\
        p_3 &= \frac{-4 \rho^2}{(n-3)4 \rho + (1 + (2n-8)\rho)((2n-4)\rho -n + 2 -1/\rho)}. \nonumber
    \end{align}
Taking only the largest terms in $n$, one may approximate the values in $\bp$ as follows, which will be useful in following theoretical development:
    \begin{align}
        p_1 &\approx \frac{1}{1-2\rho} + O(n^{-1}), \label{eq:p_vals_appx} \\
        p_2 &\approx \frac{-1}{n(1-2\rho)} + O(n^{-2}), \nonumber \\
        p_3 &\approx \frac{2}{n^2(1-2\rho)} + O(n^{-3}). \nonumber
    \end{align}

The equation in \eqref{eq:C_undir} allows one to compute the partial derivatives $\left\{ \frac{\partial \phi_i}{ \partial p_j} \right\}$.
First, based on \eqref{eq:C_undir}, we can write $\C(\bp) \bphi =  [1,0,0]^T$.
Then, we note that the
matrix function $\C(\bphi)$ in  \eqref{eq:C_undir} is linear in the terms $\bphi$, and thus, we may write $\C(\p) = \sum_{j = 1}^3 p_j \A_j^{(n)}$ for some matrices $\left\{ \A_j^{(n)}  \right\}_{j=1}^3$
that depend on $n$.  Differentiating both sides of $\C(\bp) \bphi =  [1,0,0]^T$ with respect to $p_j$ and solving gives
\begin{align}
\frac{\partial \bphi}{\partial p_j} &= - \C(\bp)^{-1} \A_{j}^{(n)} \C(\bp)^{-1} [1,0,0]^T, \nonumber
\end{align}
which holds for all $j \in \{1,2,3 \}$.
%%%%

\section{Theoretical support}
\label{sec:theory}
In this section, we outline proofs suggesting that the estimators resulting from the EMM algorithm are consistent. 

\subsection{Consistency of $\hat{\beta}_{EMM}$}
The estimator of $\beta$ resulting from the EMM algorithm, $\hat{\beta}_{EMM}$, depends on the estimated value of $\rho$, $\hat{\rho}_{EMM}$, through the covariance matrix $\bOmega$. Explicitly, given $\bOmega$, the EMM estimator 
\begin{align}
\hat{\beta}_{EMM} &= (\X^T \bOmega^{-1} \X)^{-1} \X^T \bOmega^{-1} \widehat{E[\z \mid \y]}, \quad
\end{align}
where $\widehat{E[\z \mid \y]}$ represents the estimation and approximation of $E[\z \mid \y]$ described in the EMM algorithm. This estimator is difficult to analyze in general, because, in principle, $\widehat{E[z_{jk} \mid \y]}$ depends on every entry in $\y$, and the effects of the approximations are difficult to evaluate.  Instead of direct analysis, to evaluate consistency of $\hat{\beta}_{EMM}$, we define a bounding estimator that is easier to analyze,
\begin{align}
\hat{\beta}_{bound} &= (\X^T \bOmega^{-1} \X)^{-1} \X^T \bOmega^{-1} \u, \quad u_{jk} = E[z_{jk} \mid y_{jk}]. \quad
\end{align}
It is immediately clear that $\hat{\beta}_{bound}$ is unbiased, since $E[u_{jk}] = \x_{jk}^T \bbeta$. Further, the approximations made in the EMM algorithm are meant to bound $|| \hat{\beta}_{EMM} - {\beta}^*_{MLE} ||_2^2  \le || \hat{\beta}_{bound} - {\beta}^*_{MLE} ||_2^2 $, where ${\beta}^*_{MLE}$ is the true maximum likelihood estimator. That is, the expectation estimator we compute $\widehat{E[\z \mid \y]}$ takes into account correlation information through $\Omega$, and is thus closer to the true expectation, $E[\z \mid \y]$, than $\u$. Then, we also have that $\hat{\beta}_{EMM}$ is closer to ${\beta}^*_{MLE}$ than $\hat{\beta}_{bound}$. Then, consistency of $\hat{\beta}_{bound}$ implies consistency of $\hat{\beta}_{EMM}$, since we assume that the true MLE is consistent. 

We now establish consistency of $\hat{\beta}_{bound}$. We make the following assumptions:
\begin{enumerate}
\item The true model follows a latent variable model,
\begin{align}
&\P(y_{ij} = 1 ) = \P \left( \x_{ij}^T \bbeta + \epsilon_{ij} > 0 \right), \quad \\
&E[\epsilon_{jk}] = 0. \nonumber 
\end{align}
where $\bepsilon$ is not necessarily normally distributed. 
    \item The design matrix $\X$ is such that the expressions $n^{-(1+i)}\X^T \s{S}_i \X$, for $i \in \{1,2,3 \}$, converge in probability to constant matrices. 
    \item The fourth moments of $\X$ and $\bepsilon$ are bounded, $|| \x_{jk} ||_4 \le C_1 < \infty$ and  $E[\epsilon_{jk}^4] \le C_2 < \infty$. 
    \item The estimator of $\rho$ is such that $\Omega(\hat{\rho})$ converges in probability to some positive definite matrix. 
    \item The independence assumption for relations that do not share an actor holds, such that $\epsilon_{jk}$ is independent $\epsilon_{lm}$ whenever actors $j$, $k$, $l$, and $m$ are distinct. 
\end{enumerate}
The first assumption defines the meaning of the true coefficient $\beta$. The second assumption is a standard condition required for most regression problems; a similar condition is required for consistency of any estimator which accounts for correlation in generalized linear model. 
We evaluate the second assumption in the following section, when we analyze $\hat{\rho}_{EMM}$. The fourth assumption defines the minimal independence structure. 

We start by noticing that $\u = \X\beta + \epsilon$, such that 
\begin{align}
\hat{\beta}_{bound} &= \beta + \left(n^{-2} \sum_{i=1}^3 p_i \X^T \s{S}_i \X \right)^{-1} \left(n^{-2} \sum_{i=1}^3 p_i \X^T \s{S}_i \v \right), \quad v_{jk} = E[\epsilon_{jk} \mid y_{jk}]. \quad
\end{align}
Then, as noted in the previous paragraph, the bounding estimator is unbiased, $E[\hat{\beta}_{bound}] = \beta$. It remains to establish sufficient conditions for which $\hat{\beta}_{bound}$ converges to its expectation in probability. Noting the orders of $\{ p_i \}_i$ in \eqref{eq:p_vals_appx}, we immediately have that $n^{-2} \X^T \Omega^{-1} \X$ converges in probability to a constant. A sufficient condition to establish that $\left(n^{-2} \sum_{i=1}^3 p_i \X^T \s{S}_i \v \right)$ converges in probability to its expectation (zero) is that its variance tends to zero. Expanding this variance expression,
\begin{align}
\text{var} \left(n^{-2} \sum_{i=1}^3 p_i \X^T \s{S}_i \v \right)  & = 
n^{-4} \sum_{i=1}^3 \sum_{j=1}^3 p_i p_j \X^T \s{S}_i E[\v \v^T] \s{S}_j \X,  \quad \label{eq:var_expression} \\
& = n^{-4} \sum_{i=1}^3 \sum_{j=1}^3 p_i p_j \sum_{jk, lm \in \Theta_i}\sum_{rs, tu \in \Theta_j} \x_{jk} \x_{rs}^T E[v_{lm} v_{tu}]. \nonumber
\end{align}
By assumption, every term in the sum expression in \eqref{eq:var_expression} is bounded. Also by assumption, the expectation $E[v_{lm} v_{tu}]$ is zero whenever the relations $lm$ and $tu$ do not share an actor. Using the expressions in \eqref{eq:p_vals_appx} ($p_i \propto n^2 |\Theta_i|^{-1}$) and counting terms, 
\begin{align}
\text{var} \left(n^{-2} \sum_{i=1}^3 p_i \X^T \s{S}_i \v \right)  & \propto n^{-4} \sum_{i=1}^3 \sum_{j=1}^3 \frac{n^2}{|\Theta_i|} \frac{n^2}{|\Theta_j|} \frac{|\Theta_i| |\Theta_j|}{n} = O(n^{-1}). \nonumber
\end{align}
Thus, the variance of $\hat{\beta}_{bound}$ converges to zero, so that $\hat{\beta}_{bound}$ converges in probability to the true $\bbeta$, as does $\hat{\beta}_{EMM}$.

\subsection{Consistency of $\hat{\rho}_{EMM}$}
Using the expressions in \eqref{eq:p_vals_appx} and differentiating the expected log-likelihood with respect to $\rho$, the maximum likelihood estimator is
\begin{align}
    \hat{\rho}_{MLE} &= \frac{1}{2} + \frac{1}{n^3} E[\bepsilon^T \s{S}_2 \bepsilon \mid \y ] - \frac{1}{n^2} E[\bepsilon^T \bepsilon \mid \y ] - \frac{2}{n^4} E[\bepsilon^T \s{S}_3\bepsilon \mid \y ] + O(n^{-1}). \quad \label{eq:em_appx_rho}
\end{align}
In the EMM algorithm, we approximate the expectations in \eqref{eq:em_appx_rho} using pairwise conditioning. Then, we have that
\begin{align}
    \hat{\rho}_{EMM} &= \frac{1}{2} + \frac{1}{n^3} \sum_{jk, lm \in \Theta_2} E[\epsilon_{jk} \epsilon_{lm} \mid y_{jk}, y_{lm} ] - \frac{1}{n^2} \sum_{jk} E[\epsilon_{jk}^2 \mid y_{jk} ]\ldots \label{eq:rho_bc_em_E} \\
    & \ldots - \frac{2}{n^4} \sum_{jk, lm \in \Theta_3} E[\epsilon_{jk} \mid y_{jk}] E[\epsilon_{lm} \mid y_{lm}] + O(n^{-1}). \quad \nonumber
\end{align}
According to the exchangeability assumption of the errors, the pairwise expectations are known, and the EMM estimator of $\rho$ is unbiased, 
$E[\hat{\rho}_{EMM}] = E[\epsilon_{jk} \epsilon_{lm}] = \rho$.
The EMM estimator $\hat{\rho}_{EMM}$,  converges to its expectation when the sums of conditional expectations in \eqref{eq:rho_bc_em_E} converge to their expectations. This occurs when the variances of these sums tend to zero. This fact can be established by similar counting arguments as in the previous subsection. For example, 
{ \small 
\begin{align}
    \text{var}\left( \frac{1}{n^3} \sum_{jk, lm \in \Theta_2} E[\epsilon_{jk} \epsilon_{lm} \mid y_{jk}, y_{lm} ] \right) &=  
    n^{-6} \sum_{jk, lm \in \Theta_2} \sum_{jk, lm \in \Theta_2} (E[E[\epsilon_{jk} \epsilon_{lm} \mid y_{jk}, y_{lm} ] E[\epsilon_{rs} \epsilon_{tu} \mid y_{rs}, y_{tu} ]] - \rho^2),\nonumber \\
    &= n^{-6} \frac{|\Theta_2| |\Theta_2|}{n} = O(n^{-1}), \nonumber
\end{align}
}
since $E[\epsilon_{jk} \epsilon_{lm} \mid y_{jk}, y_{lm} ]$ is independent $E[\epsilon_{rs} \epsilon_{tu} \mid y_{rs}, y_{tu} ]$ whenever all the indices $\{j,k,l,m,r,s,t,u\}$ are distinct. Thus, each of the sums of expectations in \eqref{eq:rho_bc_em_E} has variance that tends to zero, so that they converge to their marginal expectations, and $\hat{\rho}_{EMM}$ is consistent.

\subsection{Consistency under misspecification}
In the discussion of consistency of the EMM estimator, we did not require the assumption of latent normality, nor of exchangeability of the latent errors (we do require a small assumption that the sequence of constants $n^{-3} E[\epsilon^T \s{S}_2 \epsilon_{lm}]$ converges to some constant on $[0,1/2)$). Hence, when the data generating mechanism is non-Gaussian and non-exchangeable, we expect $\hat{\rho}_{EMM}$ to converge to the pseudo-true $\rho$. The pseudo-true $\rho$ is the value which minimizes the Kullback-Leibler divergence from the modeled (Gaussian, exchangeable) distribution to the true distribution \citep{huber1967under, dhaene1997pseudo}. In the discussion of consistency of $\hat{\beta}_{EMM}$, we only require that $\hat{\rho}_{EMM}$ converges to a fixed value on the interval $[0, 1/2)$, such that $\Omega(\rho)$ is positive definite. Again, when the data generating mechanism is non-Gaussian and non-exchangeable, we expect $\hat{\beta}_{EMM}$ to converge to the pseudo-true $\beta$. When the true data generating mechanism is Gaussian (but not necessarily exchangeable), the limiting pseudo-true value for $\hat{\beta}_{EMM}$ should be the true value.

\section{Simulation studies}
\label{sec:app_sims}
In this section we present details pertaining to the second simulation study in Section~\ref{sec:sim}.

\subsection{Evaluation of estimation of $\beta$}
\label{sec:sim_cons}
See Section~\ref{sec:sim_beta_est} for a description of the simulation study to evaluate performance in estimating $\bbeta$. We provide further details in the rest of this paragraph. We generated each $\{ x_{1i} \}_{i=1}^n$ as iid Bernoulli$(1/2)$ random variables, such that the second covariate is an indicator of both $x_{1i} = x_{1j} = 1$. Each of $\{ x_{2i} \}_{i=1}^n$ and $\{ x_{3ij} \}_{ij}$ were generated from iid standard normal random variables. 
We fixed  $\bbeta = [\beta_0, \beta_1, \beta_2, \beta_3]^T = [-1, 1/2, 1/2, 1/2]^T$
throughout the simulation study. 
When generating from the latent eigenmodel in \eqref{eq_binary_eigent}, we set $\Lambda = \I$, $\sigma^2_a = 1/6$, $\sigma^2_u = 1/\sqrt{6}$, and $\sigma^2_\xi = 1/3$.

To further investigate the source of poor performance of the \texttt{amen} estimators of the social relations and latent eigenmodels, we computed the bias and the variance of estimators when generating from the PX model and the latent eigenmodel in 
Figures~\ref{fig:bias_var_px}~and~\ref{fig:bias_var_le}, respectively.
Figures~\ref{fig:bias_var_px}~and~\ref{fig:bias_var_le} show that the variances of the \texttt{amen} estimators of the social relations and latent eigenmodels are similar to the PX model, however, that the bias of the \texttt{amen} estimators are substantially larger.

\begin{figure}
\centering
\begin{tabular}{ccccc}
\begin{sideways} \hspace{.3in} $n^{1/2}$ Bias \end{sideways}
& \hspace{-.01in}
\includegraphics[width=.23\textwidth]{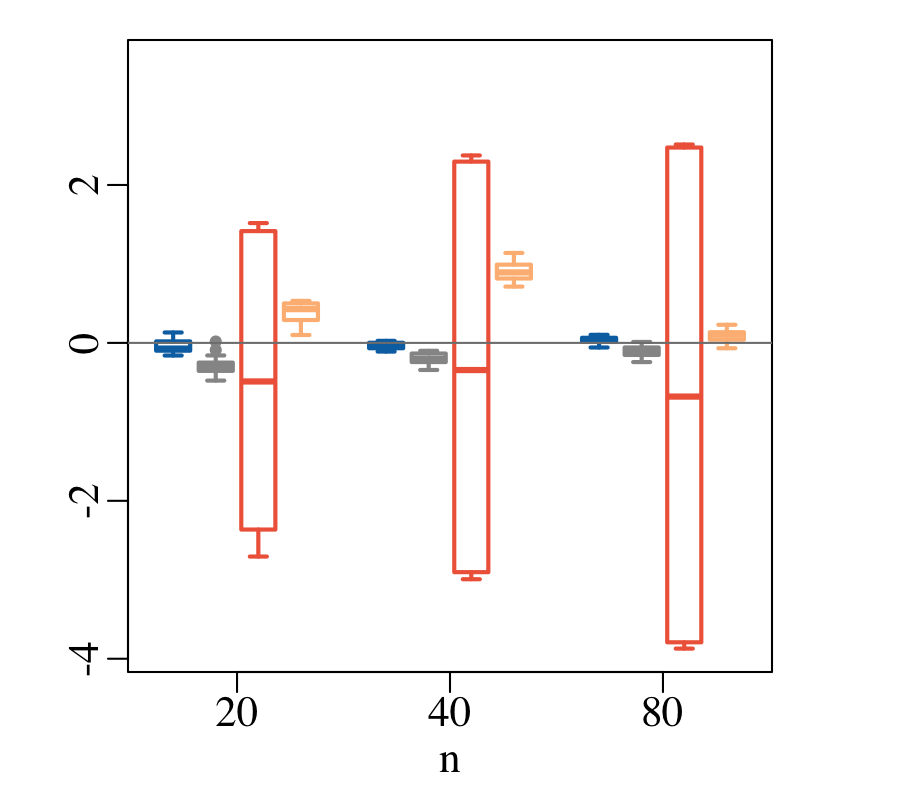}
\hspace{-.1in} & \hspace{-.1in}
\includegraphics[width=.23\textwidth]{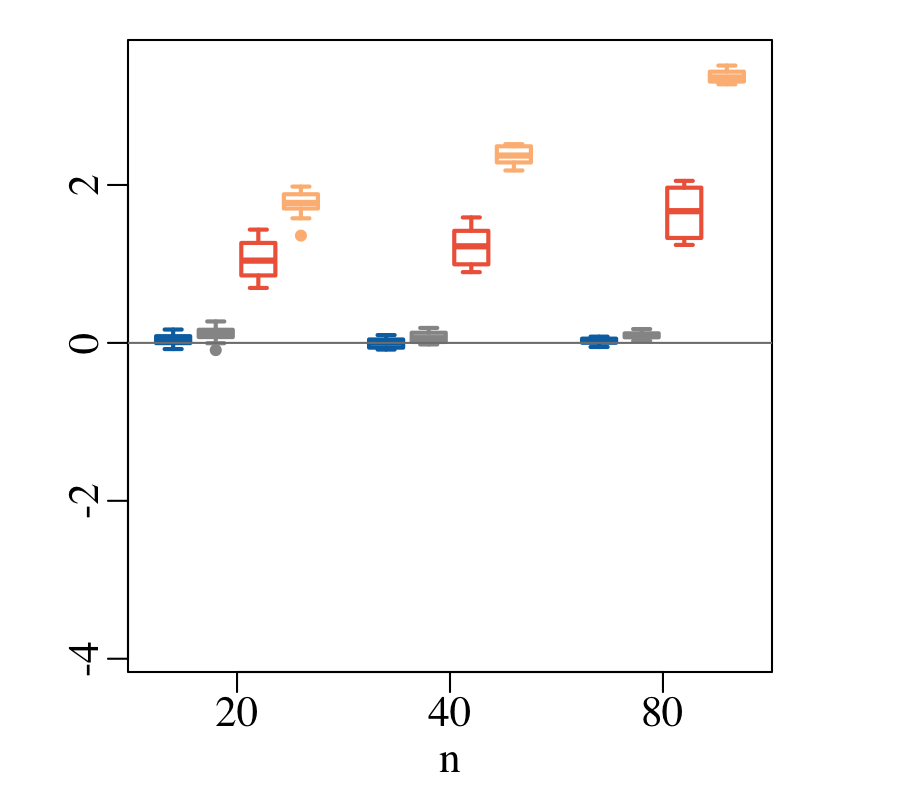}
\hspace{-.1in} & \hspace{-.1in}
\includegraphics[width=.23\textwidth]{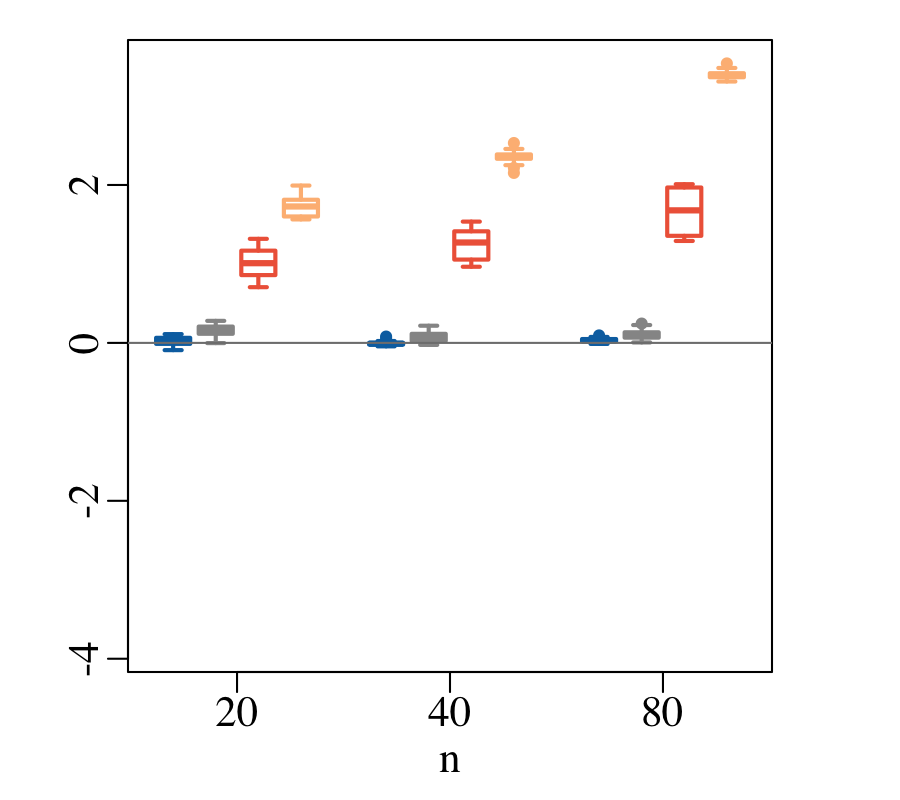}
\hspace{-.1in} & \hspace{-.1in}
\includegraphics[width=.23\textwidth]{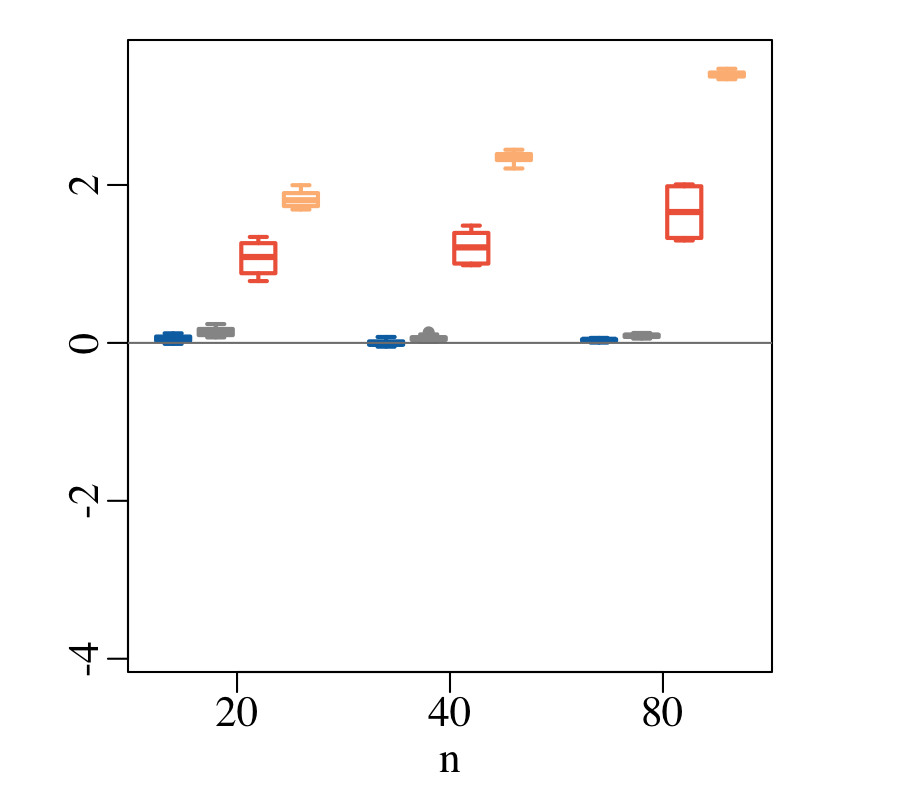} \\
\begin{sideways} \hspace{.3in} $n$ Variance \end{sideways}
& \hspace{-.01in}
\includegraphics[width=.23\textwidth]{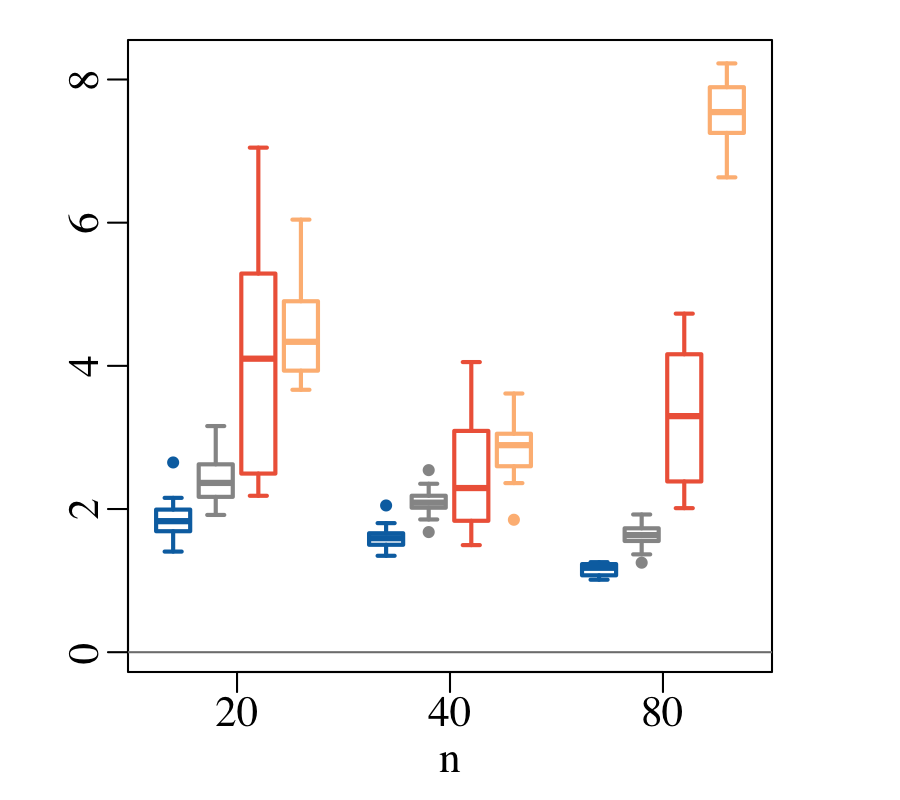}
\hspace{-.1in} & \hspace{-.1in}
\includegraphics[width=.23\textwidth]{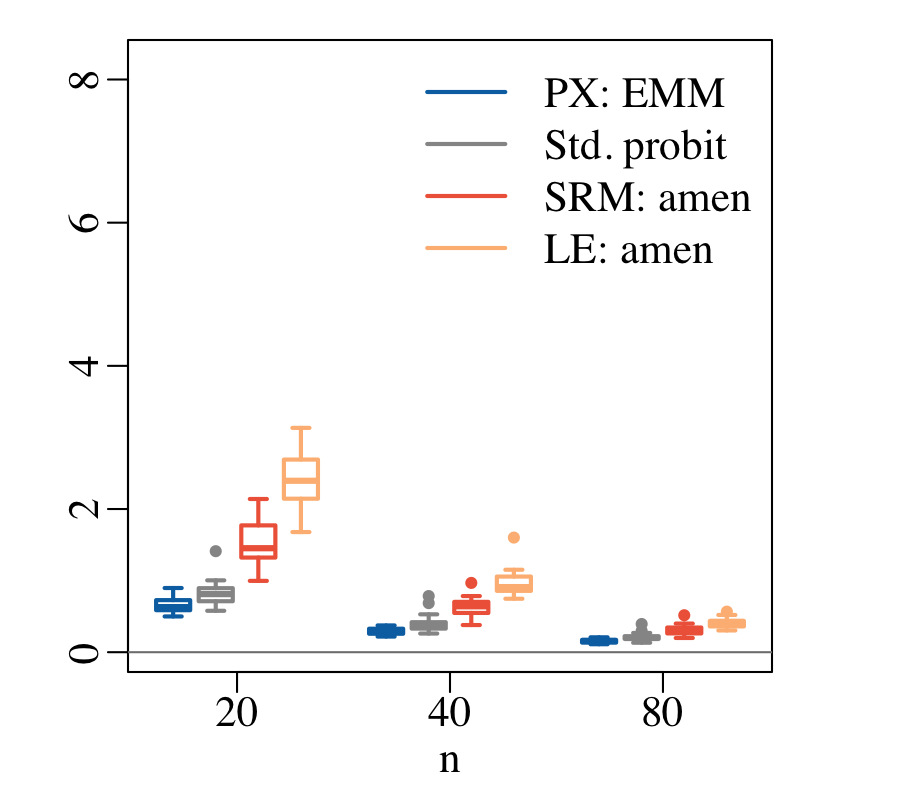}
\hspace{-.1in} & \hspace{-.1in}
\includegraphics[width=.23\textwidth]{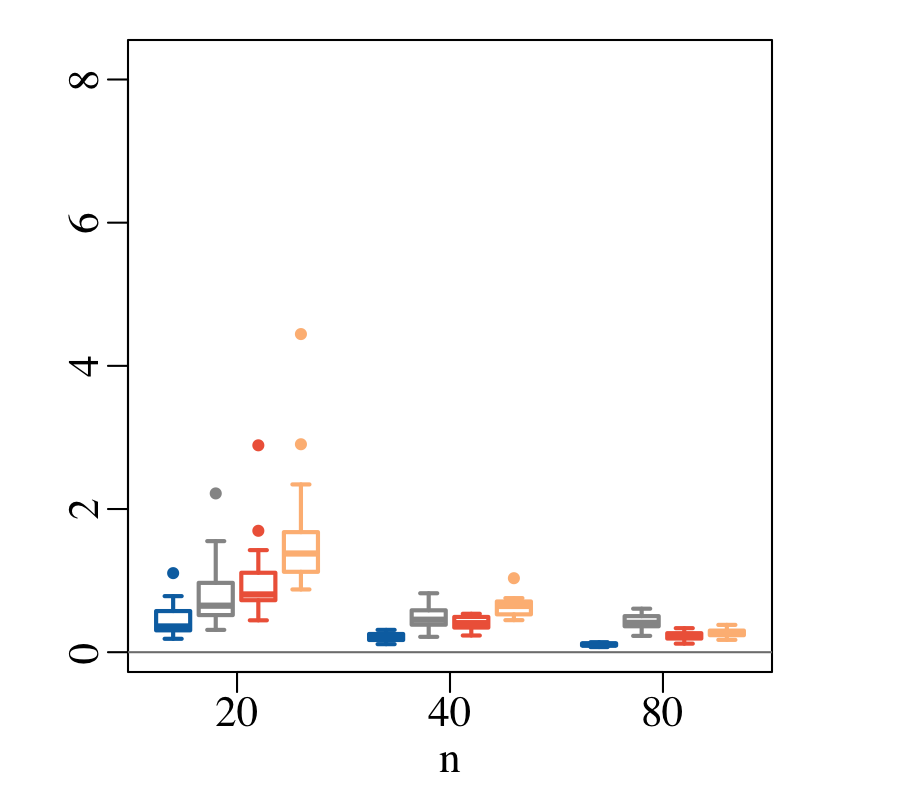}
\hspace{-.1in} & \hspace{-.1in}
\includegraphics[width=.23\textwidth]{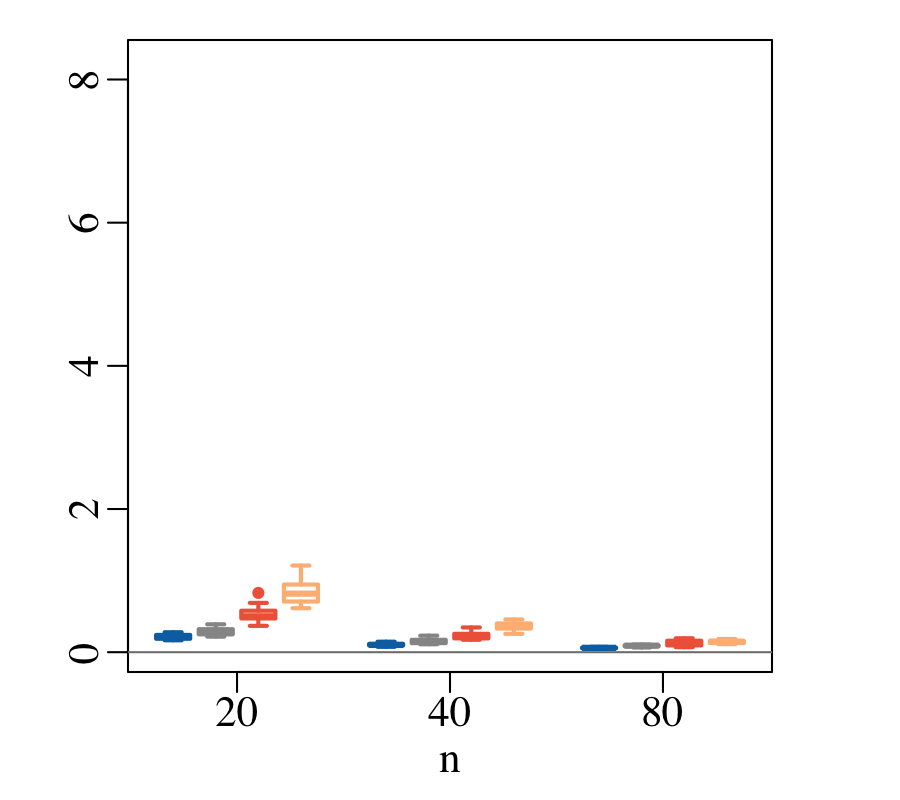} \\
\hspace{-.1in} & \hspace{-.1in}
$\beta_0$
\hspace{-.1in} & \hspace{-.1in}
$\beta_1$
\hspace{-.1in} & \hspace{-.1in}
$\beta_2$ 
\hspace{-.1in} & \hspace{-.1in}
$\beta_3$ 
\end{tabular}
  \caption{\textbf{PX model: } Scaled bias and variance of estimators of ${\bbeta}$ for a given $\X$ when generating from the PX model. Variability captured by the boxplots reflects variation with $\X$. }
  % Note that the intercept, $\beta_0$, has biases and variances on different scales than the remaining coefficients. } 
    \label{fig:bias_var_px}
  \end{figure}

\begin{figure}
\centering
\begin{tabular}{ccccc}
\begin{sideways} \hspace{.3in} $n^{1/2}$ Bias \end{sideways}
& \hspace{-.01in}
\includegraphics[width=.23\textwidth]{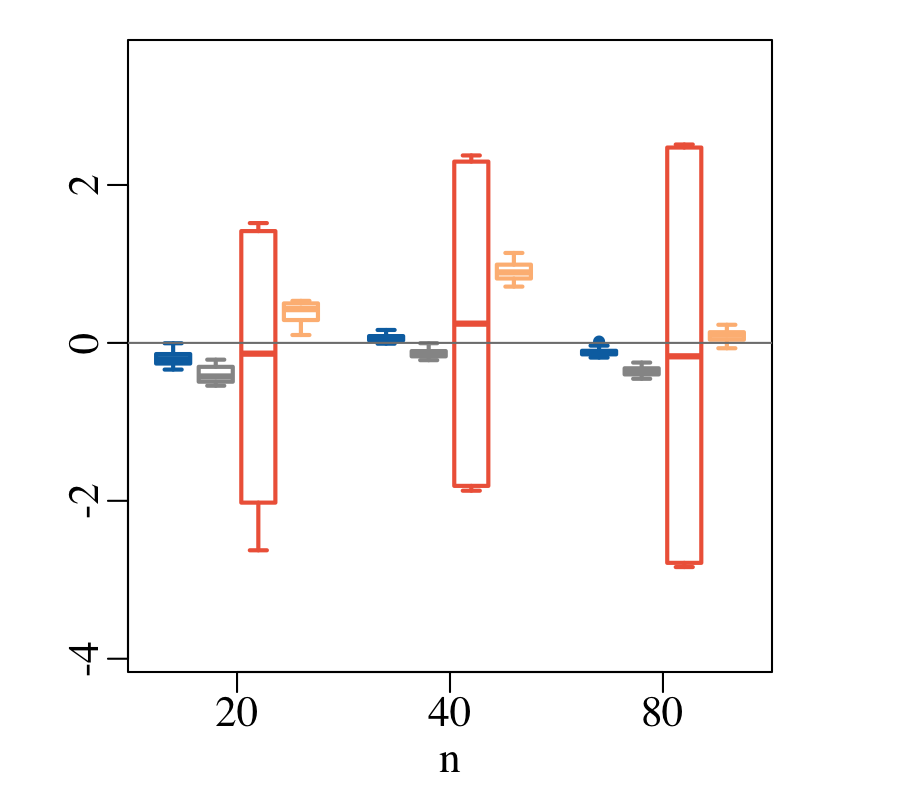}
\hspace{-.1in} & \hspace{-.1in}
\includegraphics[width=.23\textwidth]{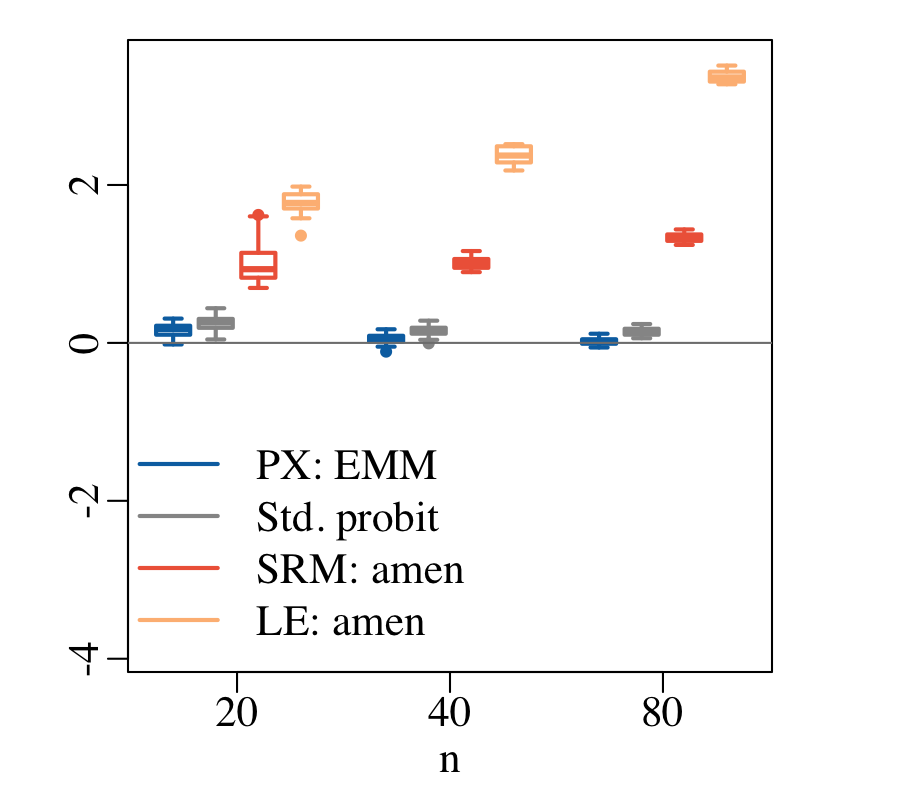}
\hspace{-.1in} & \hspace{-.1in}
\includegraphics[width=.23\textwidth]{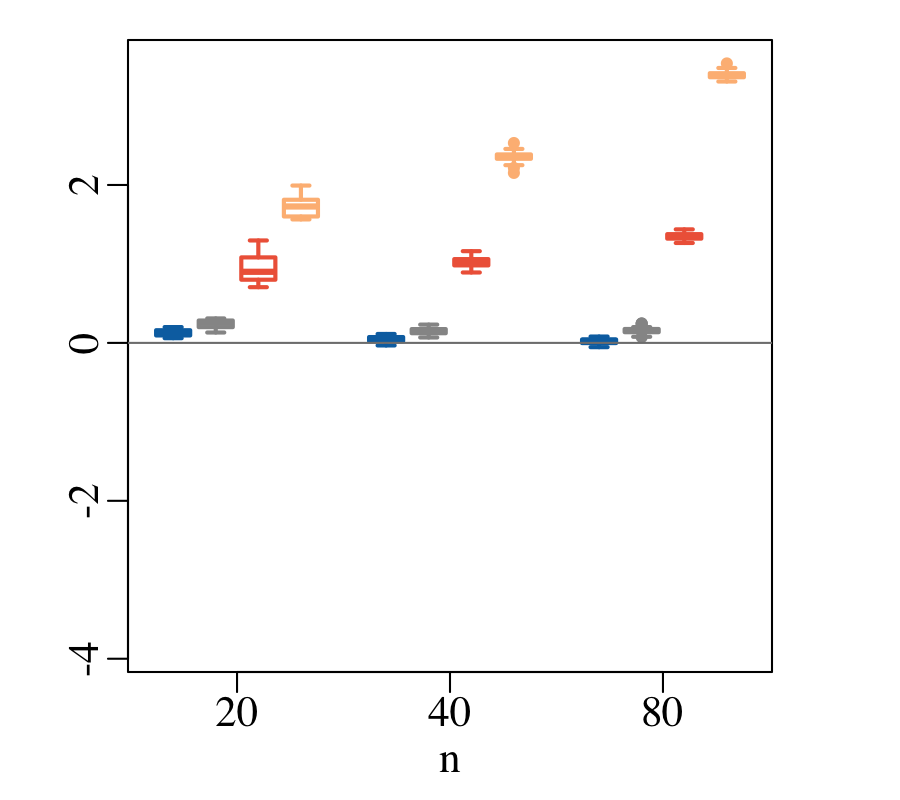}
\hspace{-.1in} & \hspace{-.1in}
\includegraphics[width=.23\textwidth]{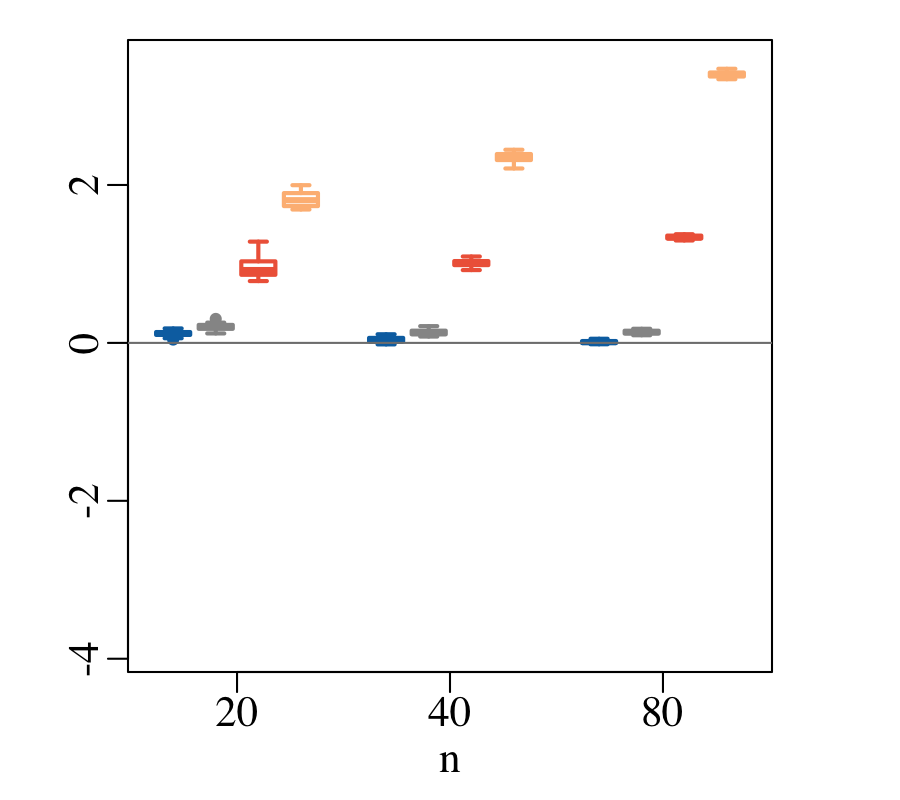} \\
\begin{sideways} \hspace{.3in} $n$ Variance \end{sideways}
& \hspace{-.01in}
\includegraphics[width=.23\textwidth]{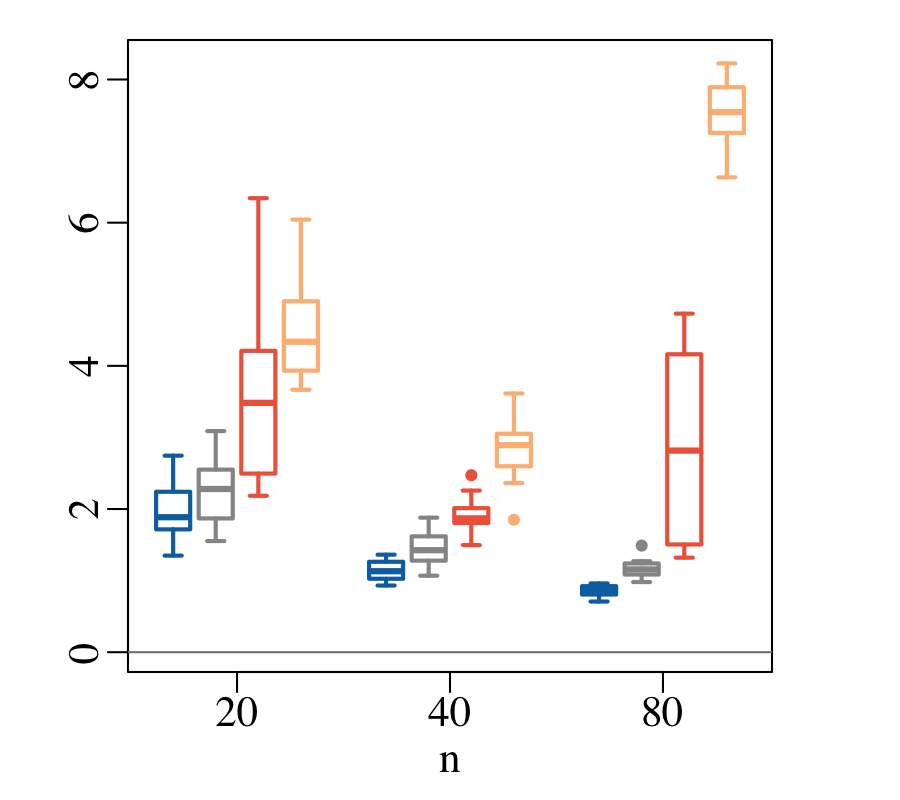}
\hspace{-.1in} & \hspace{-.1in}
\includegraphics[width=.23\textwidth]{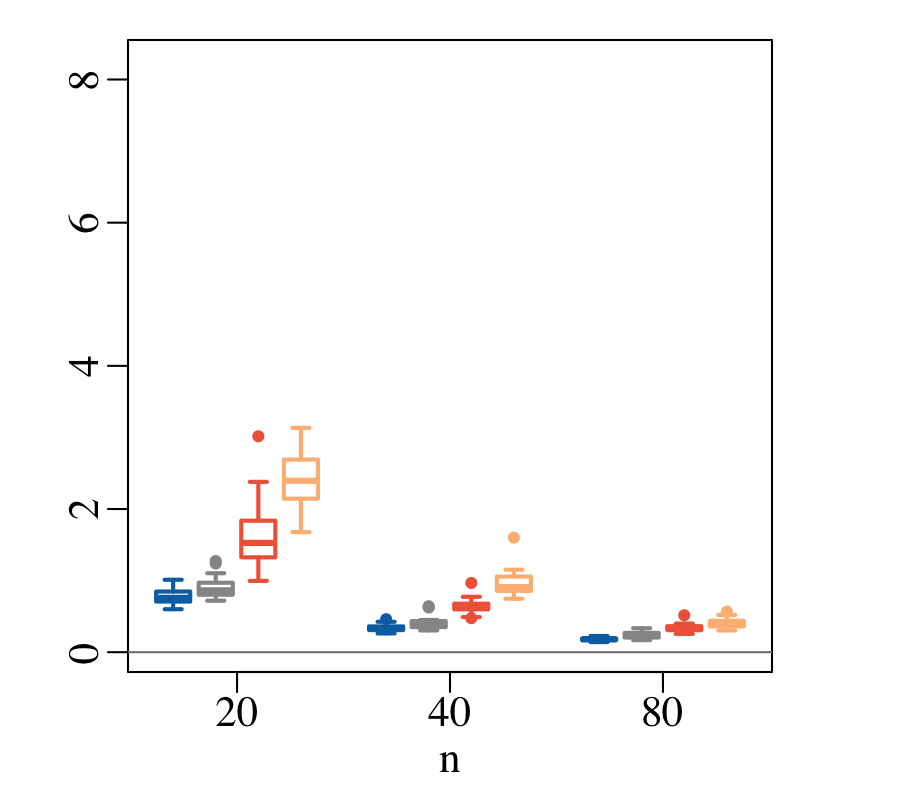}
\hspace{-.1in} & \hspace{-.1in}
\includegraphics[width=.23\textwidth]{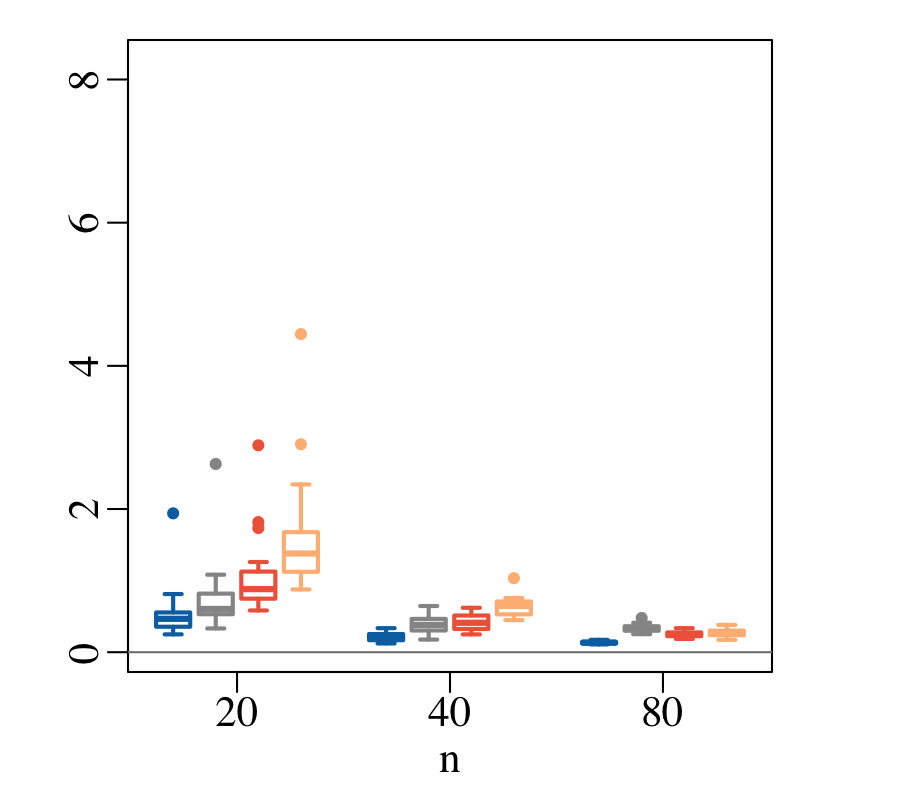}
\hspace{-.1in} & \hspace{-.1in}
\includegraphics[width=.23\textwidth]{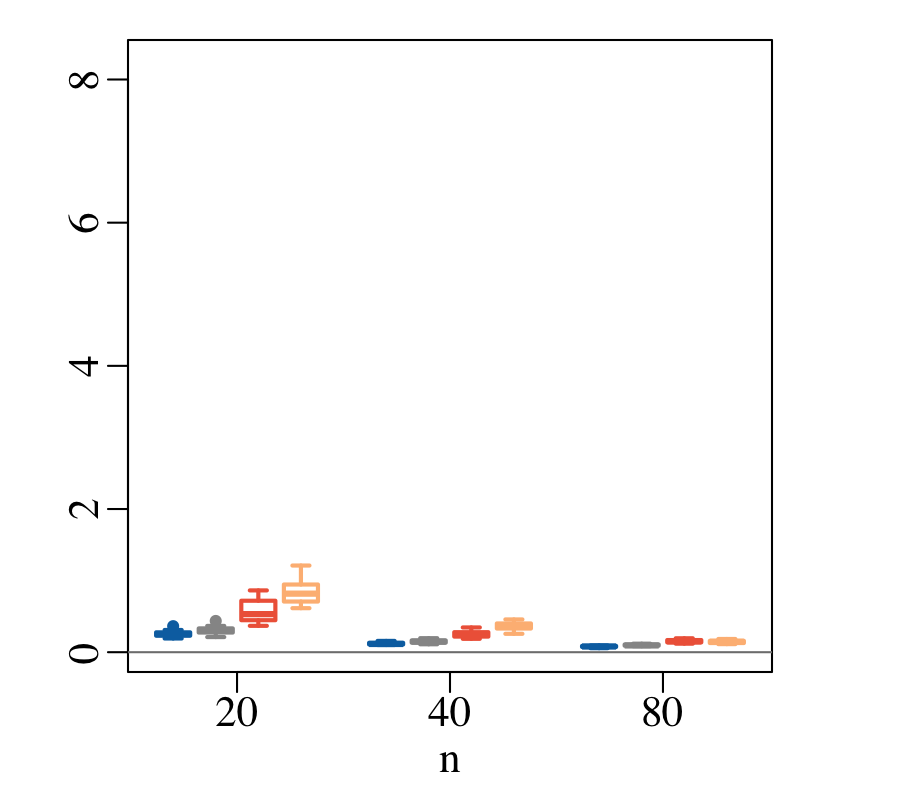} \\
\hspace{-.1in} & \hspace{-.1in}
$\beta_0$
\hspace{-.1in} & \hspace{-.1in}
$\beta_1$
\hspace{-.1in} & \hspace{-.1in}
$\beta_2$ 
\hspace{-.1in} & \hspace{-.1in}
$\beta_3$ 
\end{tabular}
  \caption{\textbf{LE model: } Scaled bias and variance of estimators of ${\bbeta}$ for a given $\X$ when generating from the latent eigenmodel. Variability captured by the boxplots reflects variation with $\X$.  } 
    \label{fig:bias_var_le}
  \end{figure}

Both the EMM estimator of the PX model and \texttt{amen} estimator of the social relations model provide estimates of $\rho$. We computed the RMSE for each estimator, for each $\X$ realization, when generating from the PX model.
In Figure~\ref{fig:rho_sim}, the RMSE plot for $\hat{\rho}$ shows that the MSE, and the spread of the MSE, decreases with $n$ for the EMM estimator, suggesting that the EMM estimator of $\rho$ is consistent. As with the $\bbeta$ parameters, the \texttt{amen} estimator displays substantially larger RMSE than the EMM estimator of $\rho$.

\begin{figure}
\centering
\includegraphics[width=.4\textwidth]{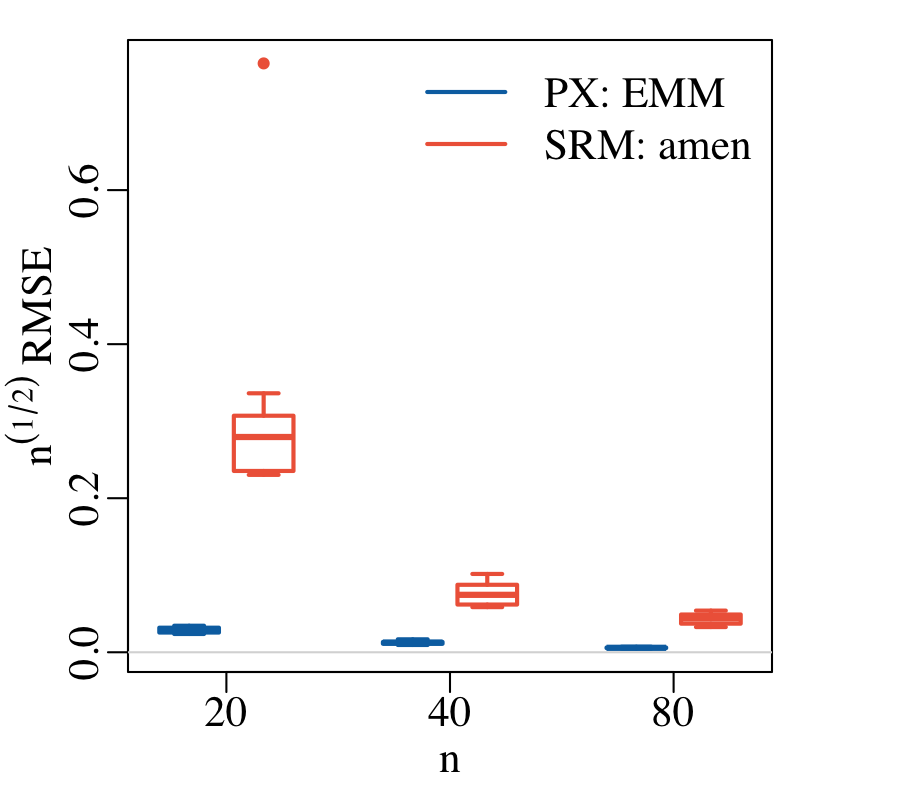}
  \caption{RMSE, scaled by $n^{1/2}$, of the EMM estimator and \texttt{amen} estimator of the social relations model of $\rho$ when generating from the PX model. Variability captured by the boxplots reflects variation in $n^{1/2}$RMSE with $\X$.    } 
    \label{fig:rho_sim}
  \end{figure}
  
  \subsection{Remaining coefficients in $t$ simulation}
  \label{sec:t_appx}
We simulated from the PX model, modified to have heavier-tailed $t_5$ error distribution.  The scaled RMSE when estimating all entries in ${\bbeta}$ is given in Figure~\ref{fig:tsim_appx}. All coefficient estimators, for both PX: EMM and standard probit regression, appear consistent, but the PX: EMM has lower RSME.

\begin{figure}[h]
\centering
\begin{tabular}{ccccc}
\includegraphics[width=.25\textwidth]{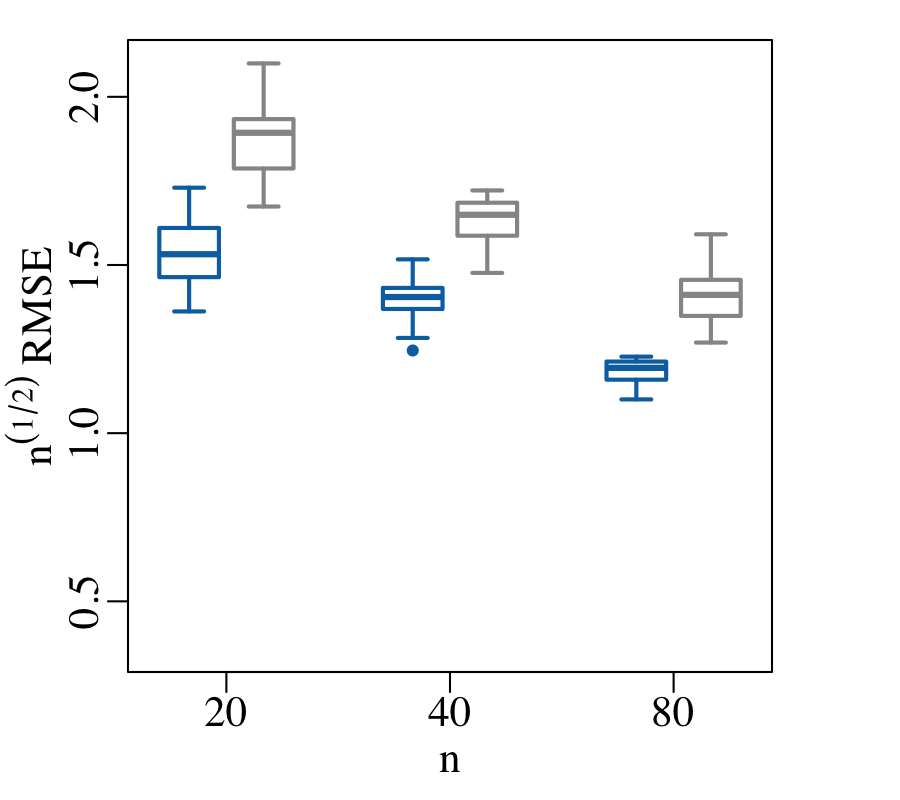}
\hspace{-.1in} & \hspace{-.1in}
\includegraphics[width=.25\textwidth]{figures/box_rmse_tPX_beta2_rho2_pop.png}
\hspace{-.1in} & \hspace{-.1in}
\includegraphics[width=.25\textwidth]{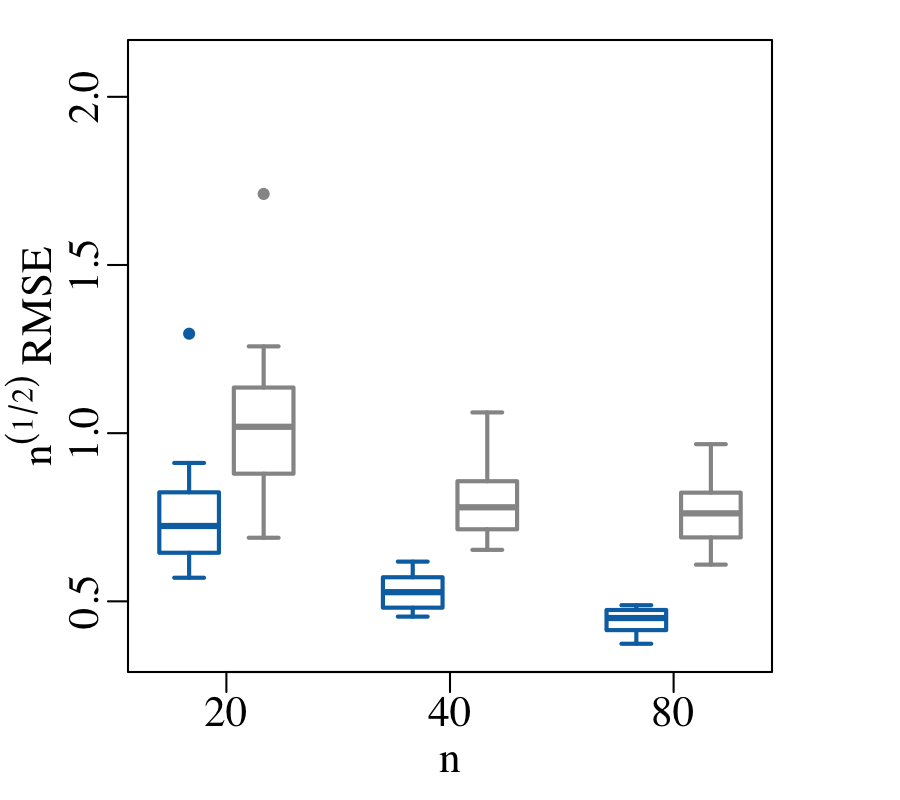}
\hspace{-.1in} & \hspace{-.1in}
\includegraphics[width=.25\textwidth]{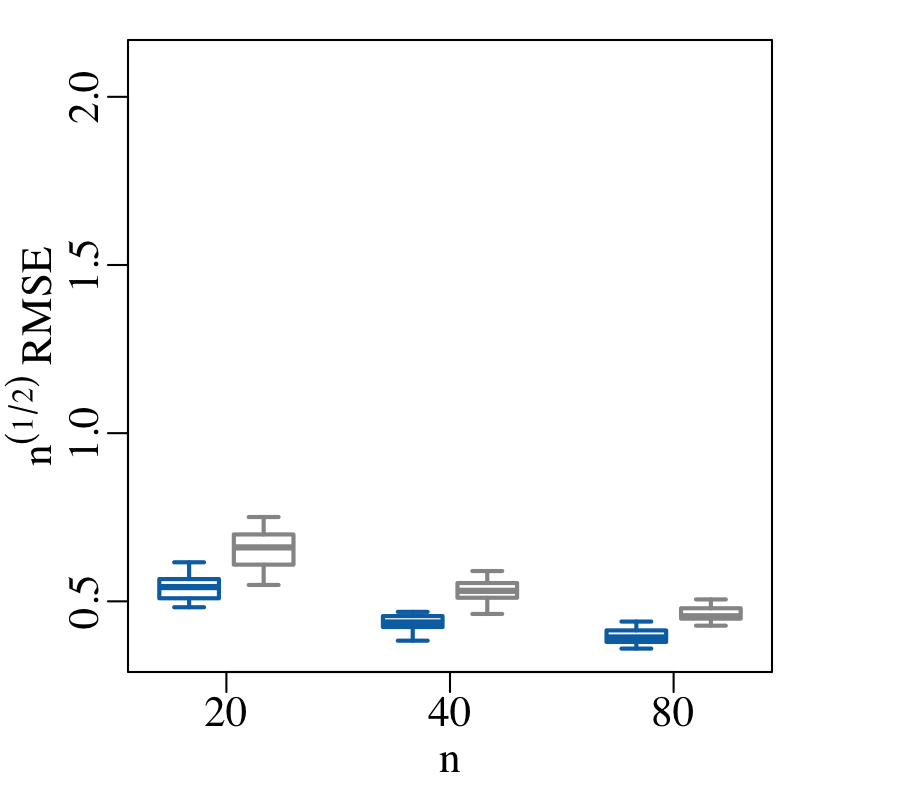} \\
$\beta_0$
\hspace{-.1in} & \hspace{-.1in}
$\beta_1$
\hspace{-.1in} & \hspace{-.1in}
$\beta_2$ 
\hspace{-.1in} & \hspace{-.1in}
$\beta_3$ 
\end{tabular}
  \caption{\textbf{$t$ model: } Scaled RMSE, for PX: EMM and standard probit regression, when generating from the PX model modified to have latent errors with heavier-tailed distribution. } 
    \label{fig:tsim_appx}
  \end{figure}

%%%%
\section{Analysis of political books network}
\label{sec:app_data}
In this section, we present additional predictive results and verify the efficacy of an approximation made by the EMM algorithm when analyzing the political books network data set. 

\subsection{Prediction performance using ROC AUC}
In Section~\ref{sec:data}, we use area under the precision-recall curve to evaluation predictive performance on the political books network data set. 
Figure~\ref{fig:pb_roc} shows the results of the cross validation study, described in Section~\ref{sec:data}, as measured by area under the receiver operating characteristic (ROC AUC). The conclusions are the same as those given in Section~\ref{sec:data}: the PX model appears to account for the inherent correlation in the data with estimation runtimes that are orders of magnitude faster than existing approaches.

\begin{figure}[h]
\centering
% \begin{tabular}{cc}
% \includegraphics[width=.44\textwidth]{figures/net_polbooks.png}
% &
\includegraphics[width=.44\textwidth]{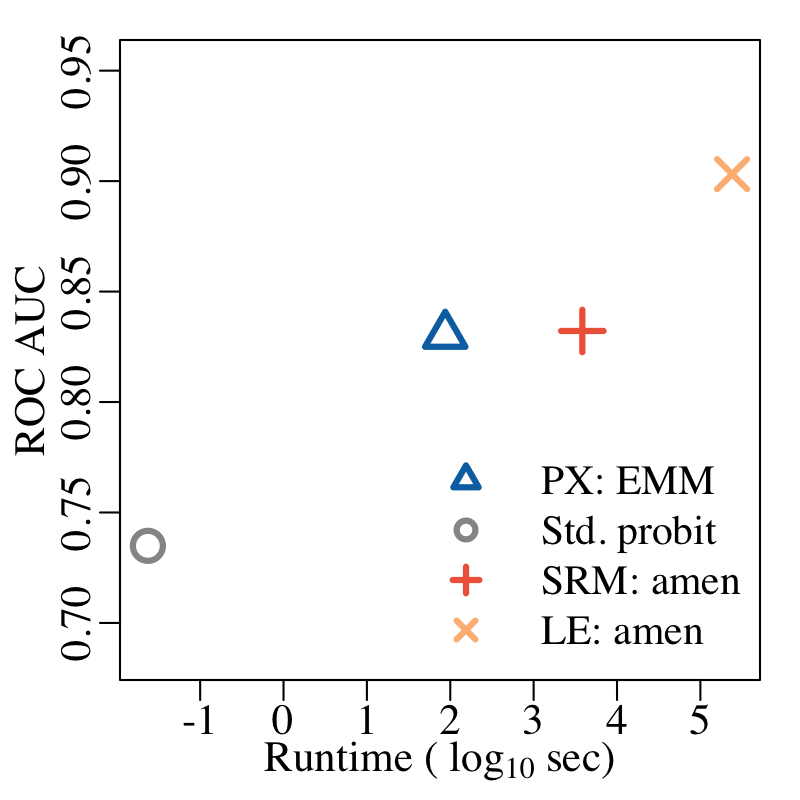}
% \end{tabular}
  \caption{Out-of-sample performance in 10-fold cross validation, as measured by area under the precision-recall curve (ROC AUC), plotted against mean runtime in the cross validation for Krebs' political books network. The estimators are standard probit assuming independent observations (Std. probit), the proposed PX estimator as estimated by EMM (PX: EMM), the social relations model as estimated by \texttt{amen} (SRM: amen), and the latent eigenmodel as estimated by \texttt{amen} (LE: amen). }
    \label{fig:pb_roc}
  \end{figure}

\subsection{Linear approximation in $\rho$ in EMM algorithm}
In Section~\ref{sec:max_rho}, we discuss a series of approximations to the E-step of an EM algorithm to maximize $\ell_\y$ with respect to $\rho$. One approximation is a linearization of the sample average $\frac{1}{|\Theta_2|} \sum_{jk, lm \in \Theta_2} E[\epsilon_{jk} \epsilon_{lm} \, | \, y_{jk}, y_{lm} ]$ with respect to $\rho$. In
Figure~\ref{fig:pb_lin}, we confirm that this approximation is reasonable for the political books network data set. Figure~\ref{fig:pb_lin} shows that the linear approximation to $\frac{1}{|\Theta_2|} \sum_{jk, lm \in \Theta_2} E[\epsilon_{jk} \epsilon_{lm} \, | \, y_{jk}, y_{lm} ]$ (dashed blue line), as described in detail in Section~\ref{sec:rho_linear_appx}, agrees  well with the true average of the pairwise expectations (solid orange line).

  \begin{figure}[h]
\centering
\includegraphics[width=.44\textwidth]{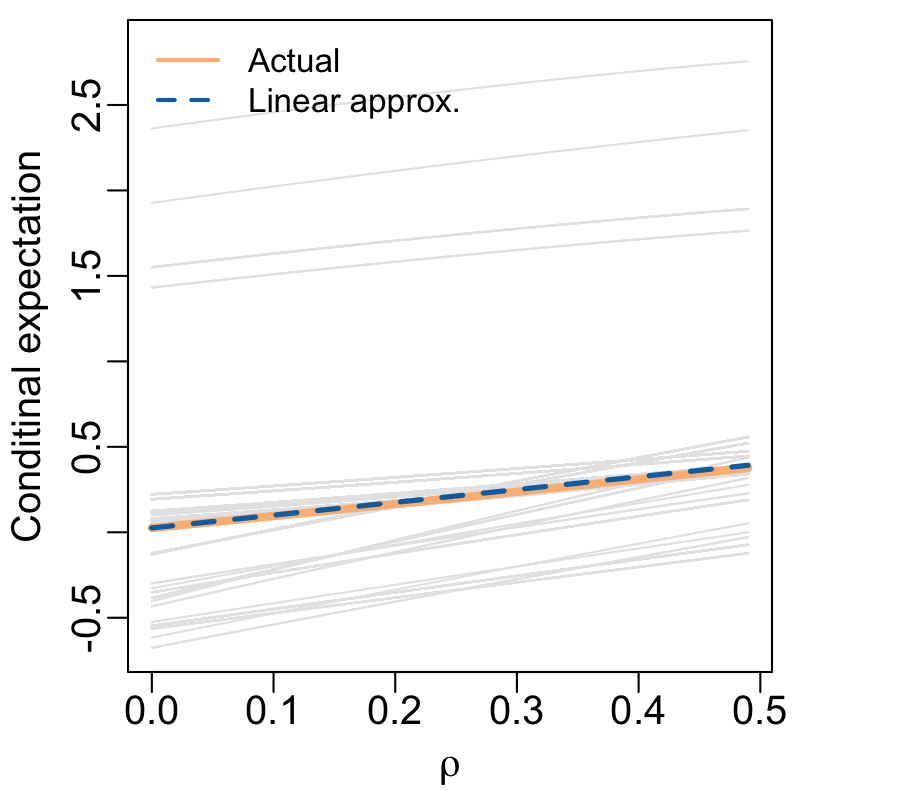}
  \caption{The average of all pairwise expectations $\frac{1}{|\Theta_2|} \sum_{jk, lm \in \Theta_2} E[\epsilon_{jk} \epsilon_{lm} \, | \, y_{jk}, y_{lm} ]$ is shown in orange, and the linear approximation to this average, described in Section~\ref{sec:estimation}, is shown in dashed blue.
  In addition, pairwise conditional expectations  $E[\epsilon_{jk} \epsilon_{lm} \, | \, y_{jk}, y_{lm} ]$ are shown in light gray, for a random subset of 500 relation pairs $(jk, lm) \in \Theta_2$.   }  
    \label{fig:pb_lin}
  \end{figure}

\end{document}